\documentclass[conference,letterpaper,onecolumn]{IEEEtran}

\usepackage{cite}

\usepackage[cmex10]{amsmath}
\usepackage{graphicx,amsfonts,amssymb,amsthm,bbm,hyperref,verbatim}
\hypersetup{colorlinks = true, linkcolor = black, citecolor = blue}
\usepackage[caption=false, font=footnotesize]{subfig} 
\usepackage{physics} 
\usepackage{bigints} 
\usepackage[margin=0.9in]{geometry}
\usepackage[all]{hypcap} 

\def\X{{\mathcal{X}}}
\def\Y{{\mathcal{Y}}}
\def\U{{\mathcal{U}}} 
\def\Z{{\mathcal{Z}}} 
\def\Simplex{{\mathcal{P}}}
\def\Gauss{{\mathcal{N}}}
\def\E{{\mathbb{E}}}
\def\VAR{{\mathbb{V}\mathbb{A}\mathbb{R}}}
\def\indep{\perp\!\!\!\perp}
\def\R{{\mathbb{R}}}
\def\C{{\mathbb{C}}}
\def\P{{\mathbb{P}}}
\def\N{{\mathbb{N}}}
\def\etaKL{{\eta_{\textsf{\tiny KL}}}}
\def\etaKLp{{\eta_{\textsf{\tiny KL}}^{(p)}}}
\def\etaChi{{\eta_{\chi^2}}}
\def\etaTV{{\eta_{\textsf{\tiny TV}}}}
\def\1{{\textbf{1}}}
\def\0{{\textbf{0}}}
\def\I{{\mathbbm{1}}}
\DeclareMathOperator*{\argmin}{arg\,min}
\DeclareMathOperator*{\esssupp}{ess\,supp}

\newtheorem{theorem}{Theorem}
\newtheorem{lemma}[theorem]{Lemma}
\newtheorem{proposition}[theorem]{Proposition}
\newtheorem{corollary}[theorem]{Corollary}

\theoremstyle{definition}
\newtheorem{definition}{Definition}

\theoremstyle{remark}

\renewcommand{\proofname}{\bfseries \emph{Proof}}

\IEEEoverridecommandlockouts

\begin{document}

\title{Linear Bounds between Contraction\\Coefficients for $f$-Divergences}

\author{Anuran Makur and Lizhong Zheng$^{\ast}$
\thanks{$^{\ast}$A. Makur and L. Zheng are with the Department of Electrical Engineering and Computer Science, Massachusetts Institute of Technology, Cambridge, MA 02139, USA (e-mail: a\_makur@mit.edu; lizhong@mit.edu).

This research was supported in part by the National Science Foundation under Award 1216476 and in part by the Hewlett-Packard Fellowship. 

This work was presented in part at the 2015 53rd Annual Allerton Conference on Communication, Control, and Computing \cite{BoundsbetweenContractionCoefficients}.}}%

\maketitle

\thispagestyle{plain}
\pagestyle{plain}

\begin{abstract}
Data processing inequalities for $f$-divergences can be sharpened using constants called ``contraction coefficients'' to produce strong data processing inequalities. For any discrete source-channel pair, the contraction coefficients for $f$-divergences are lower bounded by the contraction coefficient for $\chi^2$-divergence. In this paper, we elucidate that this lower bound can be achieved by driving the input $f$-divergences of the contraction coefficients to zero. Then, we establish a linear upper bound on the contraction coefficients for a certain class of $f$-divergences using the contraction coefficient for $\chi^2$-divergence, and refine this upper bound for the salient special case of Kullback-Leibler (KL) divergence. Furthermore, we present an alternative proof of the fact that the contraction coefficients for KL and $\chi^2$-divergences are equal for a Gaussian source with an additive Gaussian noise channel (where the former coefficient can be power constrained). Finally, we generalize the well-known result that contraction coefficients of channels (after extremizing over all possible sources) for all $f$-divergences with non-linear operator convex $f$ are equal. In particular, we prove that the so called ``less noisy'' preorder over channels can be equivalently characterized by any non-linear operator convex $f$-divergence.
\end{abstract}

\tableofcontents
\thispagestyle{plain} 
\hypersetup{linkcolor = red}

\section{Introduction}
\label{Introduction}

Strong data processing inequalities for Kullback-Leibler divergence (KL divergence or relative entropy) and mutual information \cite{AhlswedeGacsHypercontraction, ErkipCover, RenyiCorrHypercontractivity, HypercontractivityBooleanFunc, NoninteractiveSimulations}, and more generally (Csisz\'{a}r) $f$-divergences \cite{ContractionCoefficients,GraphSDPI,SDPIandSobolevInequalities}, have been studied extensively in various contexts in information theory. They are obtained by tightening traditional data processing inequalities using distribution dependent constants known as ``contraction coefficients.'' Contraction coefficients for $f$-divergences come in two flavors: those pertaining to source-channel pairs, and those pertaining only to channels. The broad goal of this paper is to study various inequalities and equalities between contraction coefficients in both settings. In the source-channel pair setting, we will primarily establish general bounds on contraction coefficients for certain classes of $f$-divergences, as well as specific bounds on the contraction coefficient for KL divergence, in terms of the contraction coefficient for $\chi^2$-divergence (or squared Hirschfeld-Gebelein-R\'{e}nyi maximal correlation). On the other hand, in the channel only setting, we will prove an appropriate generalization of the well-known result that the contraction coefficient for KL divergence is equal to the contraction coefficient for any $f$-divergence with non-linear operator convex $f$ \cite{ChoiRuskaiSeneta1994}.

\subsection{Outline}

We briefly delineate the remainder of our discussion. We will first provide an overview of the burgeoning literature on contraction coefficients in section \ref{Overview of Contraction Coefficients}. This section will compile formal definitions and key properties of both the aforementioned variants of contraction coefficients, and briefly outline their genesis in the study of ergodicity. Then, we will state and explain our main results, and discuss related literature in section \ref{Main Results}. In section \ref{Proofs of Linear Bounds between Contraction Coefficients}, we will present some useful bounds between $f$-divergences and $\chi^2$-divergence, and use them to prove linear upper bounds on contraction coefficients of source-channel pairs for a certain class of $f$-divergences and KL divergence. Following this, we will prove the equivalence between certain contraction coefficients of Gaussian sources with additive Gaussian noise channels in section \ref{Proof of Gaussian Contraction Coefficients}. In section \ref{Proof of Equivalent Characterizations of the Less Noisy Preorder}, we will prove equivalent characterizations of the less noisy preorder over channels using non-linear operator convex $f$-divergences by generalizing the main result of \cite{ChoiRuskaiSeneta1994}. Finally, we will conclude our discussion and propose future research directions in section \ref{Conclusion}.

\section{Overview of Contraction Coefficients}
\label{Overview of Contraction Coefficients}

In this section, we will define contraction coefficients for $f$-divergences and present some well-known facts about them. We begin by introducing some preliminary definitions and notation pertaining to $f$-divergences in subsection \ref{f-Divergence}, and then give a brief prelude on contraction coefficients and strong data processing inequalities in the ensuing subsections.

\subsection{$f$-Divergence}
\label{f-Divergence}

Consider a discrete sample space $\X \triangleq \{1,\dots,|\X|\}$ with $2 \leq |\X| < +\infty$, where we let singletons in $\X$ be natural numbers without loss of generality. Let $\Simplex_{\X} \subseteq \left(\R^{|\X|}\right)^{\! *}$ denote the probability simplex in $\left(\R^{|\X|}\right)^{\! *}$ of all probability mass functions (pmfs) on $\X$, where $\left(\R^{|\X|}\right)^{\! *}$ is the dual vector space of $\R^{|\X|}$ consisting of all row vectors of length $|\X|$. We perceive $\Simplex_{\X}$ as the set of all possible probability distributions of a random variable $X$ with range $\X$, and construe each pmf $P_X \in \Simplex_{\X}$ as a row vector $P_X = (P_X(1),\dots,P_X(|\X|)) \in \left(\R^{|\X|}\right)^{\! *}$. We also let $\Simplex_{\X}^{\circ} \triangleq \left\{P_X \in \Simplex_{\X} : \forall x \in \X, \, P_X(x) > 0\right\}$ denote the relative interior of $\Simplex_{\X}$. A popular notion of ``distance'' between pmfs in information theory is the $f$-divergence, which was independently introduced by Csisz\'{a}r in \cite{OriginalfDivergenceGerman,OriginalfDivergence} and by Ali and Silvey in \cite{AliSilveyDistance}. 

\begin{definition}[$f$-Divergence \cite{OriginalfDivergenceGerman,OriginalfDivergence,AliSilveyDistance}]
\label{Def: f-Divergence}
Given a convex function $f:(0,\infty) \rightarrow \R$ that satisfies $f(1) = 0$, we define the \textit{$f$-divergence} of a pmf $P_X \in \Simplex_{\X}$ from a pmf $R_X \in \Simplex_{\X}$ as:
\begin{align}
\label{Eq: f-Divergence}
D_{f}\!\left(R_X || P_X\right) & \triangleq \sum_{x \in \X}{P_X(x) f\!\left(\frac{R_X(x)}{P_X(x)}\right)} \\
& = \E_{P_X}\!\left[f\!\left(\frac{R_X(X)}{P_X(X)}\right)\right]
\end{align}
where $\E_{P_X}\!\left[\cdot\right]$ denotes the expectation with respect to $P_X$, and we assume that $f(0) = \lim_{t \rightarrow 0^{+}}{f(t)}$, $0f(0/0) = 0$, and for all $r > 0$, $0 f(r/0) = \lim_{p \rightarrow 0^{+}}{pf(r/p)} = r \lim_{p \rightarrow 0^{+}}{p f(1/p)}$, based on continuity and other considerations (see \cite[Section 3]{LieseVajda2006} for details). 
\end{definition}

The $f$-divergences generalize several well-known divergence measures that are used in information theory, statistics, and probability theory. We present some examples below:

\begin{enumerate}
\item \textit{Total variation (TV) distance}: When $f(t) = \frac{1}{2}|t-1|$, the corresponding $f$-divergence is the TV distance:
\begin{align}
\left\|R_X - P_X\right\|_{\textsf{\tiny TV}} & \triangleq \max_{A \subseteq \X}{\left|R_X(A) - P_X(A)\right|} \\
& = \frac{1}{2}\left\|R_X - P_X\right\|_{1}
\label{Eq: Total Variation Distance}
\end{align}
where $P_X(A) = \sum_{x \in A}{P_X(x)}$ for any $A \subseteq \X$, $\left\|\cdot\right\|_{p}$ denotes the $\ell^p$-norm for $p \in [1,\infty]$, and the second equality as well as several other characterizations of TV distance are proved in \cite[Chapter 4]{MarkovMixing}.
\item \textit{KL divergence}: When $f(t) = t \log(t)$ (where $\log(\cdot)$ is the natural logarithm with base $e$ throughout this paper), the corresponding $f$-divergence is the KL divergence:
\begin{equation}
\label{Eq:KL Divergence}
D\!\left(R_X||P_X\right) \triangleq \sum_{x \in \X}{R_X(x) \log\!\left(\frac{R_X(x)}{P_X(x)}\right)} .
\end{equation}
\item \textit{$\chi^2$-divergence}: When $f(t) = (t - 1)^2$ or $f(t) = t^2 - 1$, the corresponding $f$-divergence is the $\chi^2$-divergence:
\begin{equation}
\label{Eq:Chi-Squared Divergence}
\chi^2(R_X||P_X) \triangleq \sum_{x \in \X}{\frac{\left(R_X(x) - P_X(x)\right)^2}{P_X(x)}} .
\end{equation}
\item \textit{Hellinger divergence of order $\alpha \in (0,\infty)\backslash\!\{1\}$} \cite[Definition 2.10]{LieseVajda1987}: When $f(t) = \frac{t^{\alpha} - 1}{\alpha-1}$, the corresponding $f$-divergence is the Hellinger divergence (or \textit{Tsallis divergence}) of order $\alpha$:
\begin{equation}
\label{Eq: Tsallis Divergence}
\mathcal{H}_{\alpha}(R_X||P_X) \triangleq \frac{1}{\alpha - 1}\left(\sum_{x \in \X}{R_X(x)^{\alpha} P_X(x)^{1-\alpha}} - 1\right)
\end{equation}
where $\frac{1}{2}\mathcal{H}_{\frac{1}{2}}(R_X||P_X)$ is the \textit{squared Hellinger distance}, $\mathcal{H}_{2}(R_X||P_X) = \chi^2(R_X||P_X)$ is the $\chi^2$-divergence, and $\alpha = 1$ corresponds to KL divergence, $\mathcal{H}_{1}(R_X||P_X) = D(R_X||P_X)$, by analytic extension, cf. \cite[Section II]{fDivergenceBounds}.
\item \textit{Vincze-Le Cam divergence of order $\lambda \in (0,1)$ \cite{LeCam1986,Vincze1981,GyorfiVajda2001}}: When $f(t) = \lambda \bar{\lambda} \frac{(t-1)^2}{\lambda t + (1-\lambda)}$, where $\bar{\lambda} = 1 - \lambda$, the corresponding $f$-divergence is the Vincze-Le Cam divergence of order $\lambda$:
\begin{align}
\label{Eq: Le Cam Divergence}
\textsf{\small LC}_{\lambda}(R_X||P_X) & \triangleq \lambda \bar{\lambda} \sum_{x \in \X}{\frac{(R_X(x) - P_X(x))^2}{\lambda R_X(x) + \bar{\lambda} P_X(x)}} \\
& = \frac{\lambda}{\bar{\lambda}} \chi^2(R_X||\lambda R_X + \bar{\lambda} P_X)
\label{Eq: Chi-Squared Characterization of Le Cam Divergence}
\end{align}
where the special case of $\lambda = \frac{1}{2}$ is known as the \textit{Vincze-Le Cam distance} or \textit{triangular discrimination}.
\end{enumerate}

Although $f$-divergences are not valid metrics in general,\footnote{We often distinguish $f$-divergences that are metrics by dubbing them as ``distances'' (e.g. TV distance, Hellinger distance), while the term ``divergence'' is reserved for $f$-divergences that are not metrics (e.g. KL divergence, $\chi^2$-divergence).} they satisfy several useful properties. To present some of these properties, we let $\Y \triangleq \{1,\dots,|\Y|\}$ denote another discrete alphabet with $2 \leq |\Y| < +\infty$, and corresponding probability simplex $\Simplex_{\Y}$ of possible pmfs of a random variable $Y$ with range $\Y$. Furthermore, we let $\Simplex_{\Y|\X}$ denote the set of $|\X| \times |\Y|$ row stochastic matrices in $\R^{|\X| \times |\Y|}$. Throughout our discussion, the discrete channel of conditional pmfs $\left\{P_{Y|X = x} \in \Simplex_{\Y} : x \in \X\right\}$ will correspond to a transition probability matrix $W \in \Simplex_{\Y|\X}$ (where the $x$th row of $W$ is $P_{Y|X = x}$). We interpret $W:\Simplex_{\X} \rightarrow \Simplex_{\Y}$ as a map that takes input pmfs $P_X \in \Simplex_{\X}$ to output pmfs $P_Y = P_X W \in \Simplex_{\Y}$. Some well-known properties of $f$-divergences are presented next, cf. \cite{OriginalfDivergenceGerman,OriginalfDivergence}:

\begin{enumerate}
\item \textit{Non-negativity and reflexivity}: For every $R_X,P_X \in \Simplex_{\X}$, $D_{f}(R_X||P_X) \geq 0$ (by Jensen's inequality) with equality if $R_X = P_X$. Furthermore, if $f$ is strictly convex at unity,\footnote{Strict convexity of $f:(0,\infty) \rightarrow \R$ at unity implies that for every $x,y \in (0,\infty)$ and $\lambda \in (0,1)$ such that $\lambda x + (1-\lambda) y = 1$, $\lambda f(x) + (1-\lambda) f(y) > f(1)$. The aforementioned examples of $f$-divergences have this property.} then equality holds if and only if $R_X = P_X$.
\item \textit{Affine invariance}: Consider any affine function $\alpha(t) = a(t-1)$ with $a \in \R$. Clearly, $D_{\alpha}(R_X||P_X) = 0$ for every $R_X,P_X \in \Simplex_{\X}$. Hence, $f$ and $f+\alpha$ define the same $f$-divergence, i.e. $D_{f + \alpha}(R_X||P_X)\! =\! D_{f}(R_X||P_X)$ for every $R_X,P_X \in \Simplex_{\X}$.\footnote{Note that $f+\alpha:(0,\infty) \rightarrow \R$ is the function $(f+\alpha)(t) = f(t) + a(t-1)$.} 
\item \textit{Csisz\'{a}r duality}: Let the \textit{Csisz\'{a}r conjugate} function of $f$ be $f^{*} : (0,\infty) \rightarrow \R$, $f^{*}(t) = t f\!\left(\frac{1}{t}\right)$, which is also convex and satisfies $f^{**} = f$. Then, $D_{f^{*}}(P_X||R_X) = D_{f}(R_X||P_X)$ for every $R_X,P_X \in \Simplex_{\X}$.\footnote{Note also that $f$ is strictly convex at unity if and only if $f^*$ is strictly convex at unity.}  
\item \textit{Joint convexity}: The map $(R_X,P_X) \mapsto D_{f}(R_X||P_X)$ is convex in the pair of input pmfs.
\item \textit{Data processing inequality (DPI)}: For every $W \in \Simplex_{\Y|\X}$, and every $R_X,P_X \in \Simplex_{\X}$, we have (by the convexity of perspective functions):
\begin{equation}
\label{Eq:DPI for f-Divergence}
D_{f}\!\left(R_X W||P_X W\right) \leq D_{f}\!\left(R_X||P_X\right) 
\end{equation}
where equality holds if and only if $Y$ is a sufficient statistic of $X$ for performing inference about the pair $(R_X,P_X)$ (see e.g. \cite[Section 3.1]{InfoTheoryNotes}).  
\end{enumerate}

While \cite{OriginalfDivergence} and \cite[Section 2]{MutualfInformation} contain the original presentation of these properties, we also refer readers to \cite[Section 6]{InfoTheoryNotes} for a more didactic presentation. Note that due to property $2$, we only consider $f$-divergences with non-linear $f$.

We next define a notion of ``information'' between random variables corresponding to any $f$-divergence that also exhibits a DPI. For random variables $X$ and $Y$ with joint pmf $P_{X,Y}$ (consisting of $(P_X,W)$), the \textit{mutual $f$-information} between $X$ and $Y$ is defined as \cite{StochasticMatricesAndContractionCoefficients}:
\begin{align}
I_{f}\!\left(X;Y\right) & \triangleq D_{f}\!\left(P_{X,Y} || P_X P_Y\right) \\
& = \sum_{x \in \X}{P_X(x) D_{f}\!\left(P_{Y|X=x}||P_Y\right)}
\end{align}
where $P_X P_Y$ denotes the product distribution specified by the marginal pmfs $P_X$ and $P_Y$ (also see \cite[Equation (V.8)]{SDPIandSobolevInequalities}, \cite[Equation (11)]{Calmonetal2017}). When $f(t) = t\log(t)$, mutual $f$-information corresponds to standard mutual information. Moreover, mutual $f$-information possesses certain natural properties of information measures. For example, if $X$ and $Y$ are independent, then $I_{f}(X;Y) = 0$, and the converse holds when $f$ is strictly convex at unity. 

Now suppose $U$ is another random variable with discrete alphabet $\U \triangleq \{1,\dots,|\U|\}$ such that $2 \leq |\U| < +\infty$. If $(U,X,Y)$ are jointly distributed and form a Markov chain $U \rightarrow X \rightarrow Y$, then they satisfy the DPI \cite{MutualfInformation}:\footnote{Although Csisz\'{a}r studies a different notion known as \textit{$f$-informativity} in \cite{MutualfInformation}, \eqref{Eq:DPI for Mutual f-Information} can be distilled from the proof of \cite[Proposition 2.1]{MutualfInformation}.} 
\begin{equation}
\label{Eq:DPI for Mutual f-Information}
I_{f}\!\left(U;Y\right) \leq I_{f}\!\left(U;X\right)
\end{equation}
where equality holds if and only if $Y$ is a sufficient statistic of $X$ for performing inference about $U$ (i.e. $U \rightarrow Y \rightarrow X$ also forms a Markov chain). 
Needless to say, the DPIs \eqref{Eq:DPI for f-Divergence} and \eqref{Eq:DPI for Mutual f-Information} are generalizations of the better known DPIs for KL divergence and mutual information that can be found in standard texts on information theory, e.g. \cite{CoverThomas}. Finally, note that although we cite \cite{OriginalfDivergenceGerman,OriginalfDivergence} and \cite{MutualfInformation} for the DPIs \eqref{Eq:DPI for f-Divergence} and \eqref{Eq:DPI for Mutual f-Information} respectively, both DPIs were also independently proved in \cite{ZivZakai1973,ZakaiZiv1975}.\footnote{In particular, Ziv and Zakai studied \textit{generalized information functionals} in \cite{ZakaiZiv1975}, and a specialization of \cite[Theorem 5.1]{ZakaiZiv1975} yields $D_{f}(P_{U}P_{Y}||P_{U,Y}) \leq D_{f}(P_{U}P_{X}||P_{U,X})$ for any Markov chain $U \rightarrow X \rightarrow Y$. By the Csisz\'{a}r duality property of $f$-divergences, this implies \eqref{Eq:DPI for Mutual f-Information}.}  

We end this subsection with a brief exposition of the ``local quadratic behavior'' of $f$-divergences. Local approximations of $f$-divergences are geometrically appealing because they transform neighborhoods of stochastic manifolds, with certain $f$-divergences as the distance measures, into inner product spaces with the Fisher information metric \cite{EuclideanInfoTheory,InfoCouplingConf,InformationGeometry}. Consider any reference pmf $P_X \in \Simplex_{\X}^{\circ}$ (which forms the ``center of the local neighborhood'' of pmfs that we will be concerned with), and any other pmf $R_X \in \Simplex_{\X}$. Let us define the \textit{spherical perturbation} vector of $R_X$ from $P_X$ as:
\begin{equation}
\label{Eq:Spherical Perturbation}
K_X \triangleq (R_X - P_X) \, \textsf{\small diag}\!\left(\sqrt{P_X}\right)^{\!-1} 
\end{equation}
where $\sqrt{\cdot}$ denotes the entry-wise square root when applied to a vector, and $\textsf{\small diag}(\cdot)$ denotes a diagonal matrix with its input vector on the principal diagonal. Using $K_X$, we can construct a trajectory of spherically perturbed pmfs:
\begin{align}
\label{Eq: Spherically perturbed pmf}
R_X^{(\epsilon)} & = P_X + \epsilon \, K_X \, \textsf{\small diag}\!\left(\sqrt{P_X}\right) \\
& = (1 - \epsilon) P_X + \epsilon R_X
\end{align}
which is parametrized by $\epsilon \in (0,1)$, and corresponds to the convex combinations of $R_X$ and $P_X$. Note that $K_X$ provides the direction of the trajectory \eqref{Eq: Spherically perturbed pmf}, and $\epsilon$ controls how close $R_X^{(\epsilon)}$ and $P_X$ are. Furthermore, \eqref{Eq: Spherically perturbed pmf} clarifies why $K_X$ is called a ``spherical perturbation'' vector; $K_X$ is proportional to the first order perturbation term as $\epsilon \rightarrow 0$ of $\sqrt{R_X^{(\epsilon)}}$ from $\sqrt{P_X}$, which are embeddings of the pmfs $R_X^{(\epsilon)}$ and $P_X$ as vectors on the unit sphere in $\left(\R^{|\X|}\right)^{\! *}$.

Now suppose the function $f:(0,\infty) \rightarrow \R$ that defines our $f$-divergence is twice differentiable at unity with $f^{\prime \prime}(1) > 0$. Then, Taylor's theorem can be used to show that this $f$-divergence is locally proportional to $\chi^2$-divergence, cf. \cite[Section 4]{InfoTheoryStats} (or \cite{CoverThomas} for the KL divergence case):
\begin{align}
\label{Def:Local f-Divergence}
D_{f}(R_X^{(\epsilon)}||P_X) & = \frac{f^{\prime\prime}(1)}{2}\epsilon^2 \chi^2(R_X||P_X) + o\!\left(\epsilon^2\right) \\
& =  \frac{f^{\prime\prime}(1)}{2}\epsilon^2 \left\|K_X\right\|_2^2 + o\!\left(\epsilon^2\right)
\label{Def:Spherical Local f-Divergence}
\end{align}
where we use the Bachmann-Landau asymptotic little-$o$ notation.\footnote{Given two functions $f(\epsilon)$ and $g(\epsilon)$ such that $g(\epsilon)$ is non-zero, we write $f(\epsilon) = o\!\left(g(\epsilon)\right)$ if and only if $\lim_{\epsilon \rightarrow 0}{f(\epsilon)/g(\epsilon)} = 0$.} The local approximation in \eqref{Def:Spherical Local f-Divergence} is somewhat more flexible than the version in \eqref{Def:Local f-Divergence}. Indeed, we can construct a trajectory \eqref{Eq: Spherically perturbed pmf} using a spherical perturbation vector $K_X \in \left(\R^{|\X|}\right)^{\! *}$ that satisfies the orthogonality constraint $\sqrt{P_X} K_X^T = 0$, but is not of the form \eqref{Eq:Spherical Perturbation}. For sufficiently small $\epsilon \neq 0$ (depending on $P_X$ and $K_X$), the vectors $R_X^{(\epsilon)}$ defined by \eqref{Eq: Spherically perturbed pmf} are in fact valid pmfs in $\Simplex_{\X}$.\footnote{For larger values of $\epsilon$ (in magnitude), although $R_X^{(\epsilon)}$ always sums to $1$ since $\sqrt{P_X} K_X^T = 0$, it may not be entry-wise non-negative.} So, the approximation in \eqref{Def:Spherical Local f-Divergence} continues to hold because it is concerned with the regime where $\epsilon \rightarrow 0$.

It is also straightforward to verify that $f$-divergences with $f^{\prime \prime}(1) > 0$ are locally symmetric, i.e. $D_{f}(R_X^{(\epsilon)}||P_X) = D_{f}(P_X||R_X^{(\epsilon)}) + o(\epsilon^2)$. Hence, they resemble the standard Euclidean metric within a ``neighborhood'' of pmfs around a reference pmf in $\Simplex_{\X}^{\circ}$. Note that the advantage of using spherical perturbations $\{K_X \in \left(\R^{|\X|}\right)^{\! *} : \sqrt{P_X} K_X^T = 0\}$ over additive perturbations (e.g. $R_X - P_X$) is that they form an inner product space equipped with the standard Euclidean inner product. This allows us to recast \eqref{Def:Local f-Divergence} using the $\ell^2$-norm of $K_X$ instead of a weighted $\ell^2$-norm of the additive perturbation $K_X \textsf{\small diag}(\sqrt{P_X})$. Consequently, we benefit from more polished notation and simpler algebra later on\textemdash{}see our proof of Theorem \ref{Thm:Local Approximation of Contraction Coefficients}. Finally, we remark that perturbation ideas like \eqref{Eq: Spherically perturbed pmf} have also been exploited in various other contexts in information theory, and we refer readers to \cite{EuclideanInfoTheory,InfoCouplingConf,GohariAnantharam2012}, and \cite{CoordGaussian} for a few examples.

\subsection{Contraction Coefficients of Source-Channel Pairs}
\label{Contraction Coefficients of Source-Channel Pairs}

The DPIs, \eqref{Eq:DPI for f-Divergence} and \eqref{Eq:DPI for Mutual f-Information}, can be maximally tightened into so called \textit{strong data processing inequalities} (SDPIs) by inserting in pertinent constants known as contraction coefficients. There are two variants of contraction coefficients: the first depends on a source-channel pair, and the second depends solely on a channel. We introduce the former kind of coefficient in this subsection, and defer a discussion of the latter kind to the next subsection. 

\begin{definition}[Contraction Coefficient of Source-Channel Pair] 
\label{Def:Contraction Coefficient}
For any input pmf $P_X \in \Simplex_{\X}$ and any discrete channel $W \in \Simplex_{\Y|\X}$ corresponding to a conditional distribution $P_{Y|X}$, the \textit{contraction coefficient} for a particular $f$-divergence is:
$$ \eta_{f}\!\left(P_{X},P_{Y|X}\right) \triangleq \sup_{\substack{R_X \in \Simplex_{\X}:\\0 < D_{f}(R_X || P_X) < +\infty}} \frac{D_{f}(R_X W||P_X W)}{D_{f}(R_X||P_X)} $$
where the supremum is taken over all pmfs $R_X$ such that $0 < D_{f}(R_X || P_X) < +\infty$. Furthermore, if $X$ or $Y$ is a constant almost surely ($a.s.$), we define $\eta_{f}\!\left(P_{X},P_{Y|X}\right) = 0$. 
\end{definition}

Using Definition \ref{Def:Contraction Coefficient}, we may write the following SDPI from the DPI for $f$-divergences in \eqref{Eq:DPI for f-Divergence}:
\begin{equation}
\label{Eq:SDPI}
D_{f}\!\left(R_X W||P_X W\right) \leq \eta_{f}\!\left(P_{X},P_{Y|X}\right) D_{f}\!\left(R_X||P_X\right)
\end{equation}
which holds for every $R_X \in \Simplex_{\X}$, with fixed $P_X \in \Simplex_{\X}$ and $W \in \Simplex_{\Y|\X}$. The next proposition illustrates that the DPI for mutual $f$-information can be improved in a similar fashion. 

\begin{proposition}[Mutual $f$-Information Contraction Coefficient {\cite[Theorem V.2]{SDPIandSobolevInequalities}}]
\label{Prop:Contraction Coefficient of Mutual f-Information}
For any input pmf $P_X \in \Simplex_{\X}$, any discrete channel $P_{Y|X} \in \Simplex_{\Y|\X}$, and any convex function $f:(0,\infty) \rightarrow \R$ that is differentiable, has uniformly bounded derivative in some neighborhood of unity, and satisfies $f(1) = 0$, we have:
$$ \eta_{f}\!\left(P_{X},P_{Y|X}\right) = \sup_{\substack{P_{U|X}:\,U \rightarrow X \rightarrow Y \\ 0 < I_{f}(U;X) < +\infty}}{\frac{I_{f}(U;Y)}{I_{f}(U;X)}} $$
where the supremum is taken over all conditional distributions $P_{U|X} \in \Simplex_{\U|\X}$ and finite alphabets $\U$ of $U$ such that $U \rightarrow X \rightarrow Y$ form a Markov chain.\footnote{It suffices to let $|\U| = 2$ in the extremization \cite[Theorem V.2]{SDPIandSobolevInequalities}.}
\end{proposition}

Proposition \ref{Prop:Contraction Coefficient of Mutual f-Information} is proved in \cite[Theorem V.2]{SDPIandSobolevInequalities}. The special case of this result for KL divergence was proved in \cite{RenyiCorrHypercontractivity} (which tackled the finite alphabet case) and \cite{ContractionCoefficients} (which derived the general alphabet case). Intuitively, the variational problem in Proposition \ref{Prop:Contraction Coefficient of Mutual f-Information} determines the probability model that makes $Y$ as close to a sufficient statistic of $X$ for $U$ as possible (see the comment after \eqref{Eq:DPI for Mutual f-Information}). Furthermore, the result illustrates that under regularity conditions, the contraction coefficient for any $f$-divergence gracefully unifies the DPIs for the $f$-divergence and the corresponding mutual $f$-information as the tightest factor that can be inserted into either one of them. Indeed, when the random variables $U \rightarrow X \rightarrow Y$ form a Markov chain, we can write the SDPI version of \eqref{Eq:DPI for Mutual f-Information}:
\begin{equation}
\label{Eq: SDPI for Mutual f-Information}
I_{f}\!\left(U;Y\right) \leq \eta_{f}\!\left(P_{X},P_{Y|X}\right) I_{f}\!\left(U;X\right)
\end{equation}
which holds for every conditional distribution $P_{U|X}$, with fixed $P_X \in \Simplex_{\X}$ and $P_{Y|X} \in \Simplex_{\Y|\X}$. Note that even if the conditions of Proposition
\ref{Prop:Contraction Coefficient of Mutual f-Information} do not hold, \eqref{Eq: SDPI for Mutual f-Information} is still true (but $\eta_{f}\!\left(P_{X},P_{Y|X}\right)$ may not be the tightest possible constant that can be inserted into \eqref{Eq:DPI for Mutual f-Information}).  

There are two contraction coefficients that will be particularly important to our study. The first is the contraction coefficient for KL divergence:
\begin{equation}
\label{Eq:Contraction Coefficient for KL Divergence}
\etaKL\!\left(P_{X},P_{Y|X}\right) = \sup_{\substack{R_X \in \Simplex_{\X}:\\0 < D(R_X || P_X) < +\infty}} \frac{D(R_X W||P_X W)}{D(R_X||P_X)} .
\end{equation}
This quantity is related to the fundamental notion of \textit{hypercontractivity} in statistics \cite{AhlswedeGacsHypercontraction}.\footnote{Hypercontractivity refers to the phenomenon that some conditional expectation operators are contractive even when their input functional space has a (probabilistic) $\mathcal{L}^q$-norm while their output functional space has a (probabilistic) $\mathcal{L}^p$-norm with $1 \leq q < p$ (see e.g. \cite{HypercontractivityBooleanFunc}). This notion has found applications in information theory because hypercontractive quantities are often imparted with tensorization properties which permit single letterization.} In fact, the authors of \cite{AhlswedeGacsHypercontraction} and \cite{HypercontractivityBooleanFunc} illustrate how $\etaKL\!\left(P_{X},P_{Y|X}\right)$ can be defined as the chordal slope of the lower boundary of the hypercontractivity ribbon at infinity in the discrete and finite setting. 

The contraction coefficient for KL divergence elucidates a striking dichotomy between the extremizations in Definition \ref{Def:Contraction Coefficient} and Proposition \ref{Prop:Contraction Coefficient of Mutual f-Information}. To delineate this contrast, we first specialize Proposition \ref{Prop:Contraction Coefficient of Mutual f-Information} for KL divergence \cite{RenyiCorrHypercontractivity,ContractionCoefficients}:
\begin{equation}
\label{Eq: MI Contraction}
\etaKL\!\left(P_X,P_{Y|X}\right) = \sup_{\substack{P_{U},P_{X|U}: \, U \rightarrow X \rightarrow Y\\I(U;X) > 0}}{\frac{I(U;Y)}{I(U;X)}}
\end{equation}
where the optimization is (equivalently) over all $P_{U} \in \Simplex_{\U}$ with $\U = \{0,1\}$ (without loss of generality, cf. \cite[Appendix B]{ContractionCoefficients}) and all $P_{X|U} \in \Simplex_{\X|\U}$ such that marginalizing yields $P_{X}$. We next recall an example from \cite{RenyiCorrHypercontractivity} where $\U = \left\{0,1\right\}$, $X \sim \textsf{\small Bernoulli}\!\left(\frac{1}{2}\right)$, and $P_{Y|X}$ is an asymmetric erasure channel. In this numerical example, the supremum in \eqref{Eq: MI Contraction} is achieved by the sequences of pmfs $\{P^{(k)}_{X|U = 0} \in \Simplex_{\X}:k \in \N\}$, $\{P^{(k)}_{X|U = 1} \in \Simplex_{\X}:k \in \N\}$, and $\{P^{(k)}_U \in \Simplex_{\U}:k \in \N\}$ (where $\N \triangleq \{0,1,2,\dots\}$) satisfying the following conditions:
\begin{align}
\lim_{k \rightarrow \infty}{P^{(k)}_U(1)} & = 0 \, , \\
\lim_{k \rightarrow \infty}{D(P^{(k)}_{X|U=0}||P_X)} & = 0 \, , \\
\liminf_{k \rightarrow \infty}{D(P^{(k)}_{X|U=1}||P_X)} & > 0 \, , \\
\limsup_{k \rightarrow \infty}{\frac{D(P^{(k)}_{Y|U=0}||P_Y)}{D(P^{(k)}_{X|U=0}||P_X)}} & < \etaKL\!\left(P_{X},P_{Y|X}\right) , \\
\lim_{k \rightarrow \infty}{\frac{D(P^{(k)}_{Y|U=1}||P_Y)}{D(P^{(k)}_{X|U=1}||P_X)}} & = \etaKL\!\left(P_{X},P_{Y|X}\right) .
\end{align}
This example conveys that in general, although \eqref{Eq: MI Contraction} is maximized when $I(U;X) \rightarrow 0$ \cite{ErkipCover}, \eqref{Eq:Contraction Coefficient for KL Divergence} is often achieved by a sequence of pmfs $\big\{R_X^{(k)} \in \Simplex_{\X}\backslash\{P_X\}:k \in \N\big\}$ that does not tend to $P_X$ (due to the non-concave nature of this extremal problem). At first glance, this is counter-intuitive because the DPI \eqref{Eq:DPI for f-Divergence} is tight when $R_X = P_X$. However, Theorem \ref{Thm:Local Approximation of Contraction Coefficients} (presented in subsection \ref{Local Approximation of Contraction Coefficients}) will portray that maximizing the ratio of KL divergences with the constraint that $D(R_X||P_X) \rightarrow 0$ actually achieves $\etaChi\!\left(P_{X},P_{Y|X}\right)$, which is often strictly less than $\etaKL\!\left(P_{X},P_{Y|X}\right)$ \cite{RenyiCorrHypercontractivity}. Therefore, there is a stark contrast between the behaviors of the optimization problems in \eqref{Eq:Contraction Coefficient for KL Divergence} and \eqref{Eq: MI Contraction}.  

The second important contraction coefficient is the contraction coefficient for $\chi^2$-divergence:
\begin{equation}
\label{Eq:Contraction Coefficient for Chi Squared Divergence}
\etaChi\!\left(P_{X},P_{Y|X}\right) = \sup_{\substack{R_X \in \Simplex_{\X}:\\0 < \chi^2(R_X||P_X) < +\infty}} \frac{\chi^2(R_X W||P_X W)}{\chi^2(R_X||P_X)}
\end{equation}
which is closely related to a generalization of the Pearson correlation coefficient between $X$ and $Y$ known as the \textit{Hirschfeld-Gebelein-R\'{e}nyi maximal correlation}, or simply maximal correlation \cite{Hirschfeld1935,Gebelein1941,RenyiCorrelation}. We next define maximal correlation, which was proven to be a measure of statistical dependence satisfying seven natural axioms (some of which will be given in Proposition \ref{Prop:Contraction Coefficient Properties} later) that such measures should exhibit \cite{RenyiCorrelation}.

\begin{definition}[Maximal Correlation {\cite{Hirschfeld1935,Gebelein1941,RenyiCorrelation}}] 
\label{Def:Maximal Correlation}
For two jointly distributed random variables $X \in \X$ and $Y \in \Y$, the maximal correlation between $X$ and $Y$ is given by:
$$ \rho(X;Y) \triangleq \sup_{\substack{{f: \X \rightarrow \R, \, g: \Y \rightarrow \R \,:} \\ {\E[f(X)] = \E[g(Y)] = 0} \\ {\E\left[f^2(X)\right] = \E\left[g^2(Y)\right] = 1}}} \E\left[f(X)g(Y)\right] $$
where the supremum is taken over all Borel measurable functions $f$ and $g$ satisfying the constraints. Furthermore, if $X$ or $Y$ is a constant $a.s.$, there exist no functions $f$ and $g$ that satisfy the constraints, and we define $\rho(X;Y) = 0$.
\end{definition}

It can be shown that the contraction coefficient for $\chi^2$-divergence is precisely the squared maximal correlation \cite{ChiSquaredContractionMaximalCorrelation}:
\begin{equation}
\label{Eq: Maximal Correlation as Contraction Coefficient}
\etaChi\!\left(P_{X},P_{Y|X}\right) = \rho^2(X;Y) \, .
\end{equation}
Furthermore, the next proposition portrays that maximal correlation can be represented as a singular value; this was originally observed in \cite{RenyiCorrelation} in a slightly different form (also see \cite{TensorizationRenyiCorr,RenyiCorrHypercontractivity,Calmonetal2017} and \cite[Theorem 3.2.4]{MastersThesis}).

\begin{proposition}[Singular Value Characterization of Maximal Correlation {\cite{RenyiCorrelation}}] 
\label{Prop:Singular Value Characterization of Maximal Correlation}
Given the random variables $X \in \X$ and $Y \in \Y$ with joint pmf $P_{X,Y}$ (consisting of $(P_X,W)$), we may define a \emph{divergence transition matrix} (DTM):\footnote{Note that for every $x \in \X$ and $y \in \Y$ with $P_X(x) > 0$ and $P_Y(y) > 0$, $\left[B\right]_{x,y} = P_{X,Y}(x,y)/\!\sqrt{P_X(x)P_Y(y)}$, and $\left[B\right]_{x,y} = 0$ otherwise.}
\begin{equation}
\label{Eq: DTM}
B \triangleq \textsf{\small diag}\!\left(\sqrt{P_X}\right) W \textsf{\small diag}\!\left(\sqrt{P_Y}\right)^{\dagger} 
\end{equation}
where $\dagger$ denotes the Moore-Penrose pseudoinverse. Then, the maximal correlation $\rho(X;Y)$ is the second largest singular value of $B$.
\end{proposition}

\begin{proof}
See Appendix \ref{App: Proof of Singular Value Characterization of Maximal Correlation}.
\end{proof}

From Proposition \ref{Prop:Singular Value Characterization of Maximal Correlation} and \eqref{Eq: Maximal Correlation as Contraction Coefficient}, we see that the contraction coefficient for $\chi^2$-divergence is in fact the squared second largest singular value of the DTM $B$. We can write this using the Courant-Fischer variational characterization of eigenvalues or singular values (cf. \cite[Theorems 4.2.6 and 7.3.8]{BasicMatrixAnalysis}) as:
\begin{equation}
\label{Eq: Contraction Coefficient as Singular Value}
\etaChi\!\left(P_{X},P_{Y|X}\right) = \max_{\substack{x \in \R^{|\X|}\backslash\!\{\0\} : \\ \sqrt{P_X} x = 0}}{\frac{\left\|B^T x\right\|_2^2}{\left\|x\right\|_2^2}}
\end{equation}
where $\textbf{0}$ denotes the zero vector of appropriate dimension, and $\sqrt{P_X}^T$ is the right singular vector of $B^T$ corresponding to its maximum singular value of unity. 

Singular value decompositions (SVDs) of DTMs and their relation to $\chi^2$-divergence have been well-studied in statistics. For instance, one direction of work concerns the analysis and identification of so called Lancaster distributions \cite{StructBivarDistLancaster,Lancaster}. In particular, given a joint distribution $\P_{X,Y}$ over a product measurable space $\X \times \Y$ such that $\chi^2(\P_{X,Y}||\P_X \times \P_Y) < \infty$,\footnote{Note that $\P_{X}$ and $\P_{Y}$ are marginal distributions of $\P_{X,Y}$, and $\P_{X} \times \P_{Y}$ denotes their product distribution. Furthermore, $\chi^2(\P_{X,Y}||\P_X \times \P_Y) < \infty$ implies that $\P_{X,Y}$ is absolutely continuous with respect to $\P_{X} \times \P_{Y}$.} Lancaster proved in \cite{StructBivarDistLancaster} that there exist orthonormal bases, $\{f_j \in \mathcal{L}^2 \! \left(\X,\P_X\right) : 0 \leq j < |\X|\}$ and $\{g_k \in \mathcal{L}^2 \! \left(\Y,\P_Y\right) : 0 \leq k < |\Y|\}$,\footnote{Here, $\mathcal{L}^2 \! \left(\X,\P_X\right)$ (respectively $\mathcal{L}^2 \! \left(\Y,\P_Y\right)$) is the Hilbert space of square integrable functions from $\X$ (respectively $\Y$) to $\R$ with inner product defined by $\P_X$ (respectively $\P_Y$).} and some sequence $\{\sigma_k \geq 0:0 \leq k < \min\{|\X|,|\Y|\} \}$ of non-negative correlations, such that $\P_{X,Y}$ is a \textit{Lancaster distribution} exhibiting the decomposition:
\begin{equation}
\label{Eq:Expansion}
\frac{d\,\P_{X,Y}}{d \, (\P_X \times \P_Y)}(x,y) = \sum_{k = 0}^{\min\{|\X|,|\Y|\}}{\sigma_k f_k(x) g_k(y)} 
\end{equation}
where $d \, \P_{X,Y}/d \, (\P_X \times \P_Y)$ is the Radon-Nikodym derivative of $\P_{X,Y}$ with respect to $\P_{X} \times \P_{Y}$. When $\X$ and $\Y$ are finite, the decomposition in \eqref{Eq:Expansion} precisely captures the SVD structure of $B$ corresponding to $P_{X,Y}$. We refer readers to \cite[Section II-D]{PolynomialSVDofDTM} for further references.

While Lancaster decompositions are usually studied in infinite-dimensional settings, the related area of applied statistics known as \textit{correspondence analysis} deals with understanding the dependence between categorical variables. In particular, \textit{simple} correspondence analysis views a bivariate pmf $P_{X,Y}$ as a contingency table, and decomposes the dependence between $X$ and $Y$ into so called \textit{principal inertia components} using the SVD of $B$, cf. \cite{Greenacre1984}, \cite[Section 2]{GreenacreHastie1987}, and the references therein. More recently, the authors of \cite{Calmonetal2017} have studied principal inertia components (which are eigenvalues of the Gramian matrix $B^T B$) in the context of information and estimation theory. They generalize the first principal inertia component (i.e. squared maximal correlation) into a quantity known as $k$-correlation for $k \in \{1,\dots,\min\{|\X|,|\Y|\} - 1\}$ (which is the Ky Fan $(k+1)$-norm of $B^T B$ minus $1$), prove some properties of $k$-correlation such as convexity and DPI \cite[Section II]{Calmonetal2017}, and demonstrate several applications.

Yet another line of research has focused on the computational aspects of decomposing DTMs. A well-known method of computing SVDs of DTMs is the \textit{alternating conditional expectations} (ACE) algorithm\textemdash{}see \cite{ACEalgorithm} for the original algorithm in the context of non-linear regression, and \cite{ParallelACEAlgorithm} for a variant of the algorithm in the context of feature selection. At its heart, the ACE algorithm employs a power iteration method to estimate singular vectors of the DTM. It turns out that such singular vectors corresponding to larger singular values can be identified as ``more informative'' score functions. This insight has been exploited to perform inference on hidden Markov models in an image processing setting in \cite{EfficientStats}, and has been framed as a means of performing \textit{universal feature selection} in \cite{UniversalFeatureSelectionConf}. 

Having introduced the pertinent contraction coefficients, we now present some well-known properties of contraction coefficients for $f$-divergences.

\begin{proposition}[Properties of Contraction Coefficients of Source-Channel Pairs]
\label{Prop:Contraction Coefficient Properties}
The contraction coefficient for an $f$-divergence satisfies the following properties:
\begin{enumerate}
\item \emph{(Normalization)} For any joint pmf $P_{X,Y}$, we have that $0 \leq \eta_{f}\!\left(P_{X},P_{Y|X}\right) \leq 1$.
\item \emph{(Independence)} Given random variables $X$ and $Y$ with joint pmf $P_{X,Y}$, if $X$ and $Y$ are independent, then $\eta_{f}\!\left(P_{X},P_{Y|X}\right) = 0$. Conversely, if $f$ is strictly convex at unity and $\eta_{f}\!\left(P_{X},P_{Y|X}\right) = 0$, then $X$ and $Y$ are independent.
\item \emph{(Convexity \cite[Proposition III.3]{SDPIandSobolevInequalities})} For fixed $P_X \in \Simplex_{\X}^{\circ}$, the function $\Simplex_{\Y|\X} \ni P_{Y|X} \mapsto \eta_{f}\!\left(P_{X},P_{Y|X}\right)$ is convex in the channel $P_{Y|X}$.
\item \emph{(Tensorization \cite[Theorem III.9]{SDPIandSobolevInequalities})} If the convex function $f:(0,\infty) \rightarrow \R$ that defines our $f$-divergence also induces a sub-additive and homogeneous $f$-entropy,\footnote{For a convex function $f:(0,\infty) \rightarrow \R$, the $f$-entropy of a non-negative random variable $Z$ is defined as $\textsf{\footnotesize Ent}_{f}(Z) \triangleq \E[f(Z)] - f(\E[Z])$, where it is assumed that $\E[f(Z)] < \infty$ (see \cite[Section II]{SDPIandSobolevInequalities} and the references therein).} and $\left\{P_{X_i,Y_i} \! : P_{X_i} \in \Simplex_{\X_i}^{\circ} \text{ and } P_{Y_i} \in \Simplex_{\Y_i}^{\circ}\text{ for } i \! \in \! \{1,\dots,n\}\!\right\}$ are independent joint pmfs, then we have:
$$ \eta_{f}\!\left(P_{X_1^n},P_{Y_1^n|X_1^n}\right) = \max_{1 \leq i \leq n}{\eta_{f}\!\left(P_{X_i},P_{Y_i|X_i}\right)} $$
where $X_1^n = \left(X_1,\dots,X_n\right)$ and $Y_1^n = \left(Y_1,\dots,Y_n\right)$.
\item \emph{(Sub-multiplicativity)} If $U \rightarrow X \rightarrow Y$ are discrete random variables with finite ranges that form a Markov chain, then we have:
$$ \eta_{f}\!\left(P_U,P_{Y|U}\right) \leq \eta_{f}\!\left(P_U,P_{X|U}\right) \eta_{f}\!\left(P_X,P_{Y|X}\right) . $$
Furthermore, for any fixed joint pmf $P_{X,Y}$ such that $X$ is not a constant $a.s.$, we have:
$$ \eta_{f}\!\left(P_X,P_{Y|X}\right) = \sup_{\substack{P_{U|X} : \, U \rightarrow X \rightarrow Y\\\eta_{f}\left(P_U,P_{X|U}\right) > 0}}\,{\frac{\eta_{f}\!\left(P_U,P_{Y|U}\right)}{\eta_{f}\!\left(P_U,P_{X|U}\right)}} $$
where the supremum is over all arbitrary finite ranges $\U$ of $U$, and over all conditional distributions $P_{U|X} \in \Simplex_{\U|\X}$ such that $U \rightarrow X \rightarrow Y$ form a Markov chain. 
\item \emph{(Maximal Correlation Lower Bound \cite[Theorem III.3]{SDPIandSobolevInequalities}, \cite[Theorem 2]{GraphSDPI})} Suppose we have a joint pmf $P_{X,Y}$ such that the marginal pmfs satisfy $P_X \in \Simplex_{\X}^{\circ}$ and $P_Y \in \Simplex_{\Y}^{\circ}$. If the function $f:(0,\infty) \rightarrow \R$ that defines our $f$-divergence is twice differentiable at unity with $f^{\prime \prime}(1) > 0$, then we have:
$$ \etaChi\!\left(P_{X},P_{Y|X}\right) = \rho^2(X;Y) \leq \eta_{f}\!\left(P_{X},P_{Y|X}\right) . $$
\end{enumerate}
\end{proposition}

\begin{proof}
See Appendix \ref{App: Proof of Properties of Contraction Coefficients} for certain proofs, as well as relevant references for specializations of the results.
\end{proof}

We now make some relevant remarks. Firstly, since parts 1 and 2 of Proposition \ref{Prop:Contraction Coefficient Properties} illustrate that contraction coefficients are normalized measures of statistical dependence between random variables, we can perceive the sub-multiplicativity property in part 5 as an SDPI for contraction coefficients in analogy with \eqref{Eq: SDPI for Mutual f-Information}. In fact, part 5 also portrays that the contraction coefficient of the SDPI for $\eta_{f}$ is given by $\eta_{f}$ itself. 

Secondly, the version of the DPI for $\etaKL$ presented in \cite{AhlswedeGacsHypercontraction} (also see \cite[Section II-A]{HypercontractivityBooleanFunc}) holds for general $\eta_{f}$. Indeed, if $U \rightarrow X \rightarrow Y \rightarrow V$ are discrete random variables with finite ranges that form a Markov chain, then a straightforward consequence of parts 1 and 5 of Proposition \ref{Prop:Contraction Coefficient Properties} is the following monotonicity property: 
\begin{equation}
\eta_{f}\!\left(P_{U},P_{V|U}\right) \leq \eta_{f}\!\left(P_{X},P_{Y|X}\right) .
\end{equation}

Thirdly, the maximal correlation lower bound in part 6 of Proposition \ref{Prop:Contraction Coefficient Properties} can be achieved with equality. For instance, let $f(t) = t\log(t)$ and consider a doubly symmetric binary source (DSBS) with parameter $\alpha \in [0,1]$. A DSBS describes a joint distribution of two uniform Bernoulli random variables $(X,Y)$, where $X$ is passed through a binary symmetric channel (BSC) with crossover probability $\alpha$ to produce $Y$. It is proven in \cite{AhlswedeGacsHypercontraction} that for $(X,Y) \sim \textsf{\small DSBS}(\alpha)$, the maximal correlation lower bound holds with equality:
\begin{equation}
\label{Eq: Equality of Contraction Coefficients}
\etaKL\!\left(P_{X},P_{Y|X}\right) = \etaChi\!\left(P_{X},P_{Y|X}\right) = (1-2\alpha)^2
\end{equation}
where $\etaChi\!\left(P_{X},P_{Y|X}\right) = (1-2\alpha)^2$ can be readily computed using the singular value characterization of maximal correlation presented in Proposition \ref{Prop:Singular Value Characterization of Maximal Correlation}. As another example, consider $P_{Y|X}$ defined by an $|\X|$-ary erasure channel $E_{\beta} \in \Simplex_{\Y|\X}$ with erasure probability $\beta \in [0,1]$, which has input alphabet $\X$ and output alphabet $\Y = \X \cup \{\textsf{\small e}\}$, where $\textsf{\small e}$ is the erasure symbol. Recall that given an input $x \in \X$, $E_{\beta}$ erases $x$ and outputs $\textsf{\small e}$ with probability $\beta$, and copies its input $x$ with probability $1 - \beta$. It is straightforward to verify that $D_{f}(R_X E_{\beta}||P_X E_{\beta}) = (1-\beta) D_{f}(R_X||P_X)$ for every $R_X,P_X \in \Simplex_{\X}$. Therefore, for every input pmf $P_X \in \Simplex_{\X}$ and every $f$-divergence, $\eta_{f}\!\left(P_{X},P_{Y|X}\right) = 1-\beta$.

Finally, we note that although we independently proved part 6 of Proposition \ref{Prop:Contraction Coefficient Properties} using the local approximation of $f$-divergence idea from our conference paper \cite[Theorem 5]{BoundsbetweenContractionCoefficients}, the same idea is used by \cite[Theorem III.3]{SDPIandSobolevInequalities} and \cite[Theorem 2]{GraphSDPI} to prove this result. In fact, this idea turns out to stem from the proof of \cite[Theorem 5.4]{CIR93} (which is presented later in part 5 of Proposition \ref{Prop:Contraction Coefficient Properties 2}). 

\subsection{Coefficients of Ergodicity}
\label{Coefficients of Ergodicity}

Before discussing contraction coefficients that depend solely on channels, we briefly introduce the broader notion of coefficients of ergodicity. Coefficients of ergodicity were first studied in the context of understanding ergodicity and convergence rates of finite state-space (time) inhomogeneous Markov chains, cf. \cite[Section 1]{CoefficientsofErgodicity}. We present their definition below.

\begin{definition}[Coefficient of Ergodicity {\cite[Definition 4.6]{NonnegativeMatricesMarkovChains}}]
\label{Def: Coefficient of Ergodicity}
A \textit{coefficient of ergodicity} is a continuous scalar function $\eta:\Simplex_{\Y|\X} \rightarrow [0,1]$ from $\Simplex_{\Y|\X}$ (with fixed dimension) to $[0,1]$.\footnote{The set $\Simplex_{\Y|\X}$ is endowed with the standard topology induced by the Frobenius norm. Furthermore, $\X$ is typically the same as $\Y$ in Markov chain settings.} Such a coefficient is \textit{proper} if for any $W \in \Simplex_{\Y|\X}$, $\eta(W) = 0$ if and only if $W = \1 P_Y$ for some pmf $P_Y \in \Simplex_{\Y}$ (i.e. $W$ is unit rank), where $\1 \in \R^{|\X|}$ denotes a column vector with all entries equal to $1$.
\end{definition}

One useful property of proper coefficients of ergodicity is their connection to weak ergodicity. Consider a sequence of row stochastic matrices $\{W_k \in \Simplex_{\X|\X} : k \in \N\}$ that define an inhomogeneous Markov chain on the state space $\X$. Let the forward product of $r \geq 1$ of these consecutive matrices starting at index $p \in \N$ be:
\begin{equation}
T_{(p,r)} \triangleq \prod_{i = 0}^{r-1}{W_{p+i}} \, .
\end{equation}
The Markov chain $\{W_k \in \Simplex_{\X|\X} : k \in \N\}$ is said to be \textit{weakly ergodic} (in the Kolmogorov sense) if for all $x_1,x_2,x_3 \in \X$ and all $p \in \N$ \cite[Definition 4.4]{NonnegativeMatricesMarkovChains}:
\begin{equation}
\lim_{r \rightarrow \infty}{\left[T_{(p,r)}\right]_{x_1,x_3} - \left[T_{(p,r)}\right]_{x_2,x_3}} = 0 \, .
\end{equation}
This definition captures the intuition that the rows of a forward product should equalize when $r \rightarrow \infty$ for an ergodic Markov chain.\footnote{Note that if the limiting row stochastic matrix $\lim_{r \rightarrow \infty}{T_{(p,r)}}$ exists for all $p \in \N$, then the Markov chain is \textit{strongly ergodic} \cite[Definition 4.5]{NonnegativeMatricesMarkovChains}.} The next proposition conveys that weak ergodicity can be equivalently defined using proper coefficients of ergodicity.

\begin{proposition}[Weak Ergodicity {\cite[Lemma 4.1]{NonnegativeMatricesMarkovChains}}]
\label{Prop: Weak Ergodicity}
Let $\eta:\Simplex_{\X|\X} \rightarrow [0,1]$ be a proper coefficient of ergodicity. Then, the inhomogeneous Markov chain $\{W_k \in \Simplex_{\X|\X} : k \in \N\}$ is weakly ergodic if and only if:
$$ \forall p \in \N, \enspace \lim_{r \rightarrow \infty}{\eta\!\left(T_{(p,r)}\right)} = 0 \, . $$ 
\end{proposition}

To intuitively understand this result, notice that $T_{(p,r)}$ becomes (approximately) unit rank as $r \rightarrow \infty$ for a weakly ergodic Markov chain. So, we also expect $\lim_{r \rightarrow \infty}{\eta(T_{(p,r)})} = 0$, since a proper coefficient of ergodicity is continuous, and equals zero when its input is unit rank. We refer readers to \cite[Lemma 4.1]{NonnegativeMatricesMarkovChains} for a formal proof of Proposition \ref{Prop: Weak Ergodicity}. We also suggest \cite{CoefficientsofErgodicity}, \cite[Chapters 3 and 4]{NonnegativeMatricesMarkovChains}, \cite{IpsenSelee2011}, \cite[Chapter 3]{ErgodicityCoefficientsThesis}, and the references therein for further expositions of such ideas.

One of the earliest and most notable examples of proper coefficients of ergodicity is the Dobrushin contraction coefficient. Given a row stochastic matrix $W \in \Simplex_{\Y|\X}$ corresponding to a channel $P_{Y|X}$, its \textit{Dobrushin contraction coefficient} is defined as the Lipschitz constant of the map $\Simplex_{\X} \ni P_X \mapsto P_X W$ with respect to the $\ell^1$-norm (or TV distance) \cite{DobrushinContraction}:\footnote{Based on the bibliographic discussion in \cite[pp.144-147]{NonnegativeMatricesMarkovChains}, the Dobrushin contraction coefficient (or equivalently, the Dobrushin ergodicity coefficient) may also be attributed (at least partly) to both Doeblin and Markov. In fact, the coefficient has been called the \textit{Doeblin contraction coefficient} or presented as the \textit{Markov contraction lemma} in the literature (see e.g. \cite[p.619]{Kontorovich2012}).}
\begin{align}
\label{Eq: Dobrushin Contraction Definition}
\etaTV\!\left(W\right) & \triangleq \sup_{\substack{R_X,P_X \in \Simplex_{\X} : \\R_X \neq P_X}}{\frac{\left\|R_X W - P_X W\right\|_{\textsf{\tiny TV}}}{\left\|R_X - P_X\right\|_{\textsf{\tiny TV}}}} \\
\label{Eq: Simplification}
& = \max_{\substack{v \in \left(\R^{|\X|}\right)^{\! *} : \\ \left\|v\right\|_{1} = 1, \, v\1 = 0}}{\left\|v W\right\|_{1}} \\
\label{Eq: Peres Book Characterization}
& = \max_{R_X,P_X \in \Simplex_{\X}}{\left\|R_X W - P_X W\right\|_{\textsf{\tiny TV}}} \\
\label{Eq: Two Point Characterization}
& = \max_{x,x^{\prime} \in \X}{\left\|P_{Y|X = x} - P_{Y|X = x^{\prime}}\right\|_{\textsf{\tiny TV}}} \\
\label{Eq: Scrambling Characterization}
& = 1 - \! \min_{x,x^{\prime} \in \X}{\! \sum_{y \in \X}{\! \min\!\left\{\!P_{Y|X}(y|x),P_{Y|X}(y|x^{\prime})\!\right\}}}
\end{align}
where the various equivalent characterizations of \eqref{Eq: Dobrushin Contraction Definition} in \eqref{Eq: Simplification}, \eqref{Eq: Peres Book Characterization}, \eqref{Eq: Two Point Characterization} (Dobrushin's two-point characterization \cite{DobrushinContraction}), and \eqref{Eq: Scrambling Characterization} can be found in (or easily deduced from) \cite[Chapter 4.3]{NonnegativeMatricesMarkovChains}. The characterization in \eqref{Eq: Scrambling Characterization} illustrates that $\etaTV(W) < 1$ if and only if $W$ is a \textit{scrambling matrix} (which means that no two rows of $W$ are orthogonal) \cite[p.82]{NonnegativeMatricesMarkovChains}.\footnote{Thus, $\etaTV(W) < 1$ if and only if the \textit{zero error capacity} of $W$ is $0$ \cite{ZeroErrorCapacity}.} 

In addition to the properties of proper coefficients of ergodicity, $\etaTV$ also exhibits the following properties:
\begin{enumerate}
\item \textit{Lipschitz continuity} \cite[Theorem 3.4, Remark 3.5]{IpsenSelee2011}: For every $V,W \in \Simplex_{\Y|\X}$, $|\etaTV(V) - \etaTV(W)| \leq \left\|V - W\right\|_{\infty}$, where $\left\|\cdot\right\|_{\infty}$ denotes the induced $\ell^{\infty}$-norm, or maximum absolute row sum, when applied to a matrix.
\item \textit{Sub-multiplicativity} \cite[Lemma 4.3]{NonnegativeMatricesMarkovChains}: For every $V \in \Simplex_{\X|\U}$ and $W \in \Simplex_{\Y|\X}$, $\etaTV(V W) \leq \etaTV(V)  \etaTV(W)$. 
\item \textit{Sub-dominant eigenvalue bound} \cite[p.584, (9)]{CoefficientsofErgodicity}: For every $W \in \Simplex_{\X|\X}$, $\etaTV(W) \geq |\lambda|$ for every sub-dominant eigenvalue $\lambda \neq 1$ of $W$.
\end{enumerate}
The last two properties make $\etaTV$ a convenient tool for analyzing inhomogeneous Markov chains. As explained in \cite[Section 1]{IpsenSelee2011}, for a homogeneous Markov chain $W \in \Simplex_{\X|\X}$ with stationary pmf $\pi \in \Simplex_{\X}$, it is well-known that the \textit{second largest eigenvalue modulus} (SLEM) of $W$, denoted $\mu(W)$, controls the rate of convergence to stationarity. Indeed, if $\mu(W) < 1$, then $\mu(W^n) = \mu(W)^n$, and $\lim_{n \rightarrow \infty}{W^n} = \1 \pi$ with rate determined by $\mu(W)$. However, for an inhomogeneous Markov chain $\{W_k \in \Simplex_{\X|\X} : k \in \N\}$, $\mu(T_{(0,n)}) \neq \prod_{i = 0}^{n-1}{\mu(W_i)}$ in general because SLEMs are not multiplicative. The last two properties of $\etaTV$ illustrate that it is a viable replacement for SLEMs in the study of inhomogeneous Markov chains. 

\subsection{Contraction Coefficients of Channels}
\label{Contraction Coefficients of Channels}

Contraction coefficients of channels form a broad class of coefficients of ergodicity. They are defined similarly to \eqref{Eq: Dobrushin Contraction Definition}, but using $f$-divergences in place of TV distance. 

\begin{definition}[Contraction Coefficient of Channel] 
\label{Def:Contraction Coefficient 2}
For any discrete channel $W \in \Simplex_{\Y|\X}$ corresponding to a conditional distribution $P_{Y|X}$, the \textit{contraction coefficient} for a particular $f$-divergence is:
\begin{align*}
\eta_{f}\!\left(P_{Y|X}\right) & \triangleq \sup_{P_X \in \Simplex_{\X}}{\eta_{f}\!\left(P_X,P_{Y|X}\right)} \\
& = \sup_{\substack{R_X,P_X \in \Simplex_{\X}:\\0 < D_{f}(R_X || P_X) < +\infty}} \frac{D_{f}(R_X W||P_X W)}{D_{f}(R_X||P_X)}
\end{align*}
where the supremum is taken over all pmfs $R_X$ and $P_X$ such that $0 < D_{f}(R_X || P_X) < +\infty$. Furthermore, if $Y$ is a constant $a.s.$, we define $\eta_{f}\!\left(P_{Y|X}\right) = 0$. 
\end{definition}

This definition transparently yields SDPIs analogous to \eqref{Eq:SDPI} and \eqref{Eq: SDPI for Mutual f-Information} for contraction coefficients of channels. Furthermore, a version of Proposition \ref{Prop:Contraction Coefficient of Mutual f-Information} also holds for contraction coefficients of channels. Indeed, using Definition \ref{Def:Contraction Coefficient 2} and Proposition \ref{Prop:Contraction Coefficient of Mutual f-Information}, we observe that for any discrete channel $P_{Y|X} \in \Simplex_{\Y|\X}$, and any convex function $f:(0,\infty) \rightarrow \R$ that is differentiable, has uniformly bounded derivative in some neighborhood of unity, and satisfies $f(1) = 0$, we have: 
\begin{equation}
\eta_{f}\!\left(P_{Y|X}\right) = \sup_{\substack{P_{U,X}:\,U \rightarrow X \rightarrow Y \\ 0 < I_{f}(U;X) < +\infty}}{\frac{I_{f}(U;Y)}{I_{f}(U;X)}}
\end{equation}
where the supremum is taken over all joint pmfs $P_{U,X}$ and finite alphabets $\U$ of $U$ such that $U \rightarrow X \rightarrow Y$ form a Markov chain. The specialization of this result for KL divergence can be found in \cite[p.345, Problem 15.12]{CsiszarKorner} (finite alphabet case) and \cite{GraphSDPI} (general alphabet case).

There are two important examples of contraction coefficients of channels: the Dobrushin contraction coefficient for TV distance (defined in \eqref{Eq: Dobrushin Contraction Definition}), and the contraction coefficient for KL divergence. As seen earlier, given a channel $P_{Y|X}$, we use the notation $\etaTV(P_{Y|X})$, $\etaKL(P_{Y|X})$, and $\etaChi(P_{Y|X})$ to represent the contraction coefficient of $P_{Y|X}$ for TV distance, KL divergence, and $\chi^2$-divergence, respectively. It is proved in \cite{AhlswedeGacsHypercontraction} that for any channel $P_{Y|X}$, we have:
\begin{equation}
\label{Eq: KL Contraction and Chi-Squared Contraction}
\etaKL\!\left(P_{Y|X}\right) = \etaChi\!\left(P_{Y|X}\right) .
\end{equation}
Therefore, we do not need to consider $\etaKL$ and $\etaChi$ separately when studying contraction coefficients of channels. We remark that an alternative proof of \eqref{Eq: KL Contraction and Chi-Squared Contraction} (which holds for general measurable spaces) is given in \cite[Theorem 3]{GraphSDPI}. Furthermore, a perhaps lesser known observation is that the proof technique of \cite[Lemma 1, Theorem 1]{EvansSchulman99} (which analytically computes $\etaKL\!\left(P_{Y|X}\right)$ for any binary channel $P_{Y|X}$ with $|\X| = |\Y| = 2$), when appropriately generalized for arbitrary finite alphabet sizes, also yields a proof of \eqref{Eq: KL Contraction and Chi-Squared Contraction}. It is worth mentioning that the main contribution of Evans and Schulman in \cite{EvansSchulman99} is an inductive approach to upper bound $\etaKL$ in directed acyclic graphs. We refer readers to \cite{GraphSDPI} for an insightful distillation of this approach, as well as for proofs of its generalization to TV distance and its connection to site percolation. 

We next present some well-known properties of contraction coefficients of channels.

\begin{proposition}[Properties of Contraction Coefficients of Channels]
\label{Prop:Contraction Coefficient Properties 2}
The contraction coefficient for an $f$-divergence satisfies the following properties:
\begin{enumerate}
\item \emph{(Normalization)} For any discrete channel $P_{Y|X} \in \Simplex_{\Y|\X}$, we have that $0 \leq \eta_{f}\!\left(P_{Y|X}\right) \leq 1$.
\item \emph{(Independence \cite[Section 4]{CIR93})} Given a channel $P_{Y|X} \in \Simplex_{\Y|\X}$, if $X$ and $Y$ are independent, then $\eta_{f}\!\left(P_{Y|X}\right) = 0$. Conversely, if $f$ is strictly convex at unity and $\eta_{f}\!\left(P_{Y|X}\right) = 0$, then $X$ and $Y$ are independent.
\item \emph{(Convexity \cite[Section 4]{CIR93}, \cite[Proposition III.3]{SDPIandSobolevInequalities})} The function $\Simplex_{\Y|\X} \ni P_{Y|X} \mapsto \eta_{f}\!\left(P_{Y|X}\right)$ is convex.
\item \emph{(Sub-multiplicativity \cite[Section 4]{CIR93})} If $U \rightarrow X \rightarrow Y$ are discrete random variables with finite ranges that form a Markov chain, then we have:
$$ \eta_{f}\!\left(P_{Y|U}\right) \leq \eta_{f}\!\left(P_{X|U}\right) \eta_{f}\!\left(P_{Y|X}\right) . $$
\item \emph{($\chi^2$-Divergence Contraction Lower Bound \cite[Theorem 5.4]{CIR93}, \cite[Proposition II.6.15]{StochasticMatricesAndContractionCoefficients})} Given a channel $P_{Y|X} \in \Simplex_{\Y|\X}$, if the function $f:(0,\infty) \rightarrow \R$ that defines our $f$-divergence is twice differentiable at unity with $f^{\prime \prime}(1) > 0$, then we have:
$$ \etaChi\!\left(P_{Y|X}\right) \leq \eta_{f}\!\left(P_{Y|X}\right) . $$
\item \emph{(TV Distance Contraction Upper Bound \cite[Theorem 4.1]{CIR93}, \cite[Proposition II.4.10]{StochasticMatricesAndContractionCoefficients})} For any channel $P_{Y|X} \in \Simplex_{\Y|\X}$, we have:
$$ \eta_{f}\!\left(P_{Y|X}\right) \leq \etaTV\!\left(P_{Y|X}\right) . $$ 
\end{enumerate}
\end{proposition}

We omit proofs of these results, because the proofs are either analogous to the corresponding proofs in Proposition \ref{Prop:Contraction Coefficient Properties}, or are given in the associated references. Parts 1, 2, and 3 of Proposition \ref{Prop:Contraction Coefficient Properties 2} portray that contraction coefficients of channels are often valid proper coefficients of ergodicity.\footnote{The convexity of $P_{Y|X} \mapsto \eta_{f}(P_{Y|X})$ in part 3 of Proposition \ref{Prop:Contraction Coefficient Properties 2} implies that this map is continuous on the interior of $\Simplex_{\Y|\X}$. So, only $\eta_{f}$ that are also continuous on the boundary of $\Simplex_{\Y|\X}$ are valid coefficients of ergodicity.} We also remark that an extremization result analogous to part 5 of Proposition \ref{Prop:Contraction Coefficient Properties}, albeit less meaningful, can be derived in part 4 of Proposition \ref{Prop:Contraction Coefficient Properties 2}.

While \eqref{Eq: KL Contraction and Chi-Squared Contraction} shows that part 5 of Proposition \ref{Prop:Contraction Coefficient Properties 2} can be easily achieved with equality, the inequality in part 6 is often strict. For example, when $P_{Y|X}$ is a binary channel with parameters $a,b \in [0,1]$ and row stochastic transition probability matrix:
\begin{equation}
W = 
\begin{bmatrix}
1-a & a \\
b & 1-b
\end{bmatrix}
\end{equation}
it is straightforward to verify that $\etaKL\!\left(P_{Y|X}\right) \leq \etaTV\!\left(P_{Y|X}\right)$, with the inequality usually strict, since we have:
\begin{align}
\label{Eq: Binary KL Contraction Coefficient}
\etaKL\!\left(P_{Y|X}\right) & = 1 - \left(\sqrt{a(1-b)} + \sqrt{b(1-a)}\right)^2 \\
\etaTV\!\left(P_{Y|X}\right) & = |1-a-b|
\label{Eq: Binary TV Contraction Coefficient}
\end{align}
where \eqref{Eq: Binary KL Contraction Coefficient} is proved in \cite[Theorem 1]{EvansSchulman99}, and \eqref{Eq: Binary TV Contraction Coefficient} is easily computed via \eqref{Eq: Two Point Characterization}. Moreover, in the special case where $P_{Y|X}$ is a BSC with crossover probability $\alpha \in [0,1]$, we get \cite{AhlswedeGacsHypercontraction}:
\begin{equation}
\label{Eq: DSBS TV Bound}
\etaKL\!\left(P_{Y|X}\right) = (1-2\alpha)^2 \leq |1-2\alpha| = \etaTV\!\left(P_{Y|X}\right) .
\end{equation}
On the other hand, as shown towards the end of subsection \ref{Contraction Coefficients of Source-Channel Pairs}, $\eta_{f}\!\left(P_{Y|X}\right) = 1-\beta$ for every $f$-divergence when $P_{Y|X}$ is an $|\X|$-ary erasure channel with erasure probability $\beta \in [0,1]$.

In view of part 5 and \eqref{Eq: KL Contraction and Chi-Squared Contraction}, it is natural to wonder whether there are other $f$-divergences whose contraction coefficients (for channels) also collapse to $\etaChi$. The following result from \cite[Theorem 1]{ChoiRuskaiSeneta1994} generalizes \eqref{Eq: KL Contraction and Chi-Squared Contraction} and addresses this question.

\begin{proposition}[Contraction Coefficients for Operator Convex $f$-Divergences {\cite[Theorem 1]{ChoiRuskaiSeneta1994}, \cite{StochasticMatricesAndContractionCoefficients}}]
\label{Prop: Operator Convex f-Divergence Contraction}
For every non-linear operator convex function $f:(0,\infty) \rightarrow \R$ such that $f(1) = 0$, and every channel $P_{Y|X}$, we have:
$$ \eta_{f}\!\left(P_{Y|X}\right) = \etaChi\!\left(P_{Y|X}\right) . $$
\end{proposition}

The proof of \cite[Theorem 1]{ChoiRuskaiSeneta1994} relies on an elegant integral representation of operator convex functions. Such representations are powerful tools for proving inequalities between contraction coefficients, and we will use them to generalize Proposition \ref{Prop: Operator Convex f-Divergence Contraction} in subsection \ref{Less Noisy Preorder and Operator Convexity}. In fact, part 6 of Proposition \ref{Prop:Contraction Coefficient Properties 2} can also be proved using an integral representation argument, cf. \cite[Theorem III.1]{SDPIandSobolevInequalities}. 

In closing this overview, we also refer readers to \cite[Section 2]{GraphSDPI} for a complementary and comprehensive survey of contraction coefficients, and for references to various applications of these ideas in the literature.

\section{Main Results and Discussion}
\label{Main Results}

We will primarily derive bounds between various contraction coefficients in this paper. In particular, we will address the following leading questions:
\begin{enumerate}
\item \textit{Can we achieve the maximal correlation lower bound in Proposition \ref{Prop:Contraction Coefficient Properties} by adding constraints to the extremal problem that defines contraction coefficients of source-channel pairs?} \\
Yes, we can constrain the input $f$-divergence to be small as shown in Theorem \ref{Thm:Local Approximation of Contraction Coefficients} in subsection \ref{Local Approximation of Contraction Coefficients}. 
\item While we typically lower bound $\etaKL\!\left(P_X,P_{Y|X}\right)$ using $\etaChi\!\left(P_X,P_{Y|X}\right)$ (Proposition \ref{Prop:Contraction Coefficient Properties} part 6), we typically upper bound it using $\etaTV\!\left(P_{Y|X}\right)$ (Proposition \ref{Prop:Contraction Coefficient Properties 2} part 6). \textit{Is there a simple upper bound on $\etaKL\!\left(P_X,P_{Y|X}\right)$ in terms of $\etaChi\!\left(P_X,P_{Y|X}\right)$?} \\
Yes, two such bounds are given in Corollary \ref{Cor:Contraction Coefficient Bound} and Theorem \ref{Thm:Refined Contraction Coefficient Bound} in subsection \ref{Linear Bounds between Contraction Coefficients}.
\item \textit{Can we extend this upper bound for KL divergence to other $f$-divergences?} \\
Yes, a more general bound is presented in Theorem \ref{Thm:General Contraction Coefficient Bound} in subsection \ref{Linear Bounds between Contraction Coefficients}.
\item When $X$ and $Y$ are jointly Gaussian, the mutual information characterization in \eqref{Eq: MI Contraction} can be used to establish that $\etaKL\!\left(P_X,P_{Y|X}\right) = \etaChi\!\left(P_X,P_{Y|X}\right)$ \cite[Theorem 7]{ErkipCover}. \textit{Is there a simple proof of this result that directly uses the definition of $\etaKL$? Does this equality hold when we add a power constraint to the extremization in $\etaKL$?} \\
Yes, we discuss the Gaussian case in subsection \ref{Contraction Coefficients of Gaussian Random Variables}, and prove this equality for $\etaKL$ with a power constraint in Theorem \ref{Thm: Gaussian Contraction Coefficients}. Our proof also establishes the known equality using the KL divergence definition of $\etaKL$.
\item Contraction coefficients of channels are closely related to the less noisy preorder over channels \cite[Section 6]{GraphSDPI}. \textit{Can we generalize the result in Proposition \ref{Prop: Operator Convex f-Divergence Contraction} to say something more about the less noisy preorder?} \\
Yes, we introduce the less noisy preorder in subsection \ref{Less Noisy Preorder and Operator Convexity}, and derive a class of equivalent characterizations for it in Theorem \ref{Thm: Equivalent Characterizations of Less Noisy Preorder}.
\end{enumerate}
The bounds we will derive in response to questions 2, 3, and 4 have the form of the upper bound in:
\begin{equation}
\label{Eq: Linear Bounds}
\etaChi\!\left(P_{X},P_{Y|X}\right) \leq \eta_{f}\!\left(P_{X},P_{Y|X}\right) \leq C \etaChi\!\left(P_{X},P_{Y|X}\right)
\end{equation}
where the first inequality is simply the maximal correlation lower bound from Proposition \ref{Prop:Contraction Coefficient Properties}, and the constant $C$ depends on $P_{X,Y}$ and $f$; note that $C = 1$ in the setting of question 4. We refer to such bounds as \textit{linear bounds} between contraction coefficients of source-channel pairs. We state our main results in the next few subsections.

\subsection{Local Approximation of Contraction Coefficients}
\label{Local Approximation of Contraction Coefficients}

We assume in this subsection and in subsection \ref{Linear Bounds  between Contraction Coefficients} that we are given the random variables $X \in \X$ and $Y \in \Y$ with joint pmf $P_{X,Y}$ such that the marginal pmfs satisfy $P_X \in \Simplex_{\X}^{\circ}$ and $P_Y \in \Simplex_{\Y}^{\circ}$. Our first result portrays that forcing the input $f$-divergence to be small translates general contraction coefficients into the contraction coefficient for $\chi^2$-divergence.

\begin{theorem}[Local Approximation of Contraction Coefficients]
\label{Thm:Local Approximation of Contraction Coefficients} 
Suppose we are given a convex function $f:(0,\infty) \rightarrow \R$ that is strictly convex and twice differentiable at unity with $f(1) = 0$ and $f^{\prime\prime}(1) > 0$. Then, we have:
$$ \etaChi\!\left(P_{X},P_{Y|X}\right) = \lim_{\delta \rightarrow 0^{+}}{\sup_{\substack{{R_X \in \Simplex_{\X}:}\\{0 < D_{f}(R_X||P_X) \leq \delta}}}{\frac{D_{f}(R_X W||P_X W)}{D_{f}(R_X||P_X)}}} $$
where $W \in \Simplex_{\Y|\X}$ is the row stochastic transition probability matrix representing the channel $P_{Y|X}$.
\end{theorem} 

We refer readers to Appendix \ref{App: Proof of Local Approximation of Contraction Coefficients} for the proof, and note that the specialization of Theorem \ref{Thm:Local Approximation of Contraction Coefficients} for KL divergence was presented along with a proof sketch in the conference version of this paper \cite[Theorem 3]{BoundsbetweenContractionCoefficients}. We now make several pertinent remarks. Firstly, notice that the proof of part 6 of Proposition \ref{Prop:Contraction Coefficient Properties} in Appendix \ref{App: Proof of Properties of Contraction Coefficients} (or the independent proofs in \cite[Theorem III.2]{SDPIandSobolevInequalities} and \cite[Theorem 2]{GraphSDPI}) already captures the intuition that performing the optimization of $\eta_{f}\!\left(P_X,P_{Y|X}\right)$ over local perturbations of $P_X$ yields $\etaChi\!\left(P_X,P_{Y|X}\right)$ due to \eqref{Def:Spherical Local f-Divergence} and \eqref{Eq: Contraction Coefficient as Singular Value}. However, this proof (with minor modifications) only demonstrates that $\etaChi\!\left(P_X,P_{Y|X}\right)$ is upper bounded by the right hand side of Theorem \ref{Thm:Local Approximation of Contraction Coefficients}. Although it is intuitively clear that this upper bound is met with equality, the formal proof contains a few technical details as shown in Appendix \ref{App: Proof of Local Approximation of Contraction Coefficients}. 

Secondly, Theorem \ref{Thm:Local Approximation of Contraction Coefficients} transparently portrays that the maximal correlation lower bound in part 6 of Proposition \ref{Prop:Contraction Coefficient Properties} can be achieved when the optimization of $\eta_{f}\!\left(P_X,P_{Y|X}\right)$ imposes an additional constraint that the input $f$-divergence is small. (Hence, Theorem \ref{Thm:Local Approximation of Contraction Coefficients} implies the maximal correlation bound.) This insight has proved useful in comparing $\etaChi\!\left(P_X,P_{Y|X}\right)$ and $\etaKL\!\left(P_X,P_{Y|X}\right)$ in statistical contexts \cite[p.5]{Kimetal2017}.

Thirdly, Theorem \ref{Thm:Local Approximation of Contraction Coefficients} can be construed as a minimax characterization of $\etaChi\!\left(P_{X},P_{Y|X}\right)$ since the supremum of the ratio of $f$-divergences is a non-increasing function of $\delta$ and the limit (as $\delta \rightarrow 0^+$) can therefore be replaced by an infimum (over all $\delta > 0$). 

Fourthly, when the conditions of Proposition \ref{Prop:Contraction Coefficient of Mutual f-Information} and Theorem \ref{Thm:Local Approximation of Contraction Coefficients} hold, it is straightforward to verify that:
\begin{equation}
\label{Eq: Local Mutual f-Information Contraction}
\eta_{f}\!\left(P_X,P_{Y|X}\right) = \lim_{\delta \rightarrow 0^+}{\sup_{\substack{P_{U|X} :\, U \rightarrow X \rightarrow Y\\ 0 < I_{f}(U;X) \leq \delta}}{\frac{I_{f}(U;Y)}{I_{f}(U;X)}}}
\end{equation}
where the supremum is taken over all conditional distributions $P_{U|X} \in \Simplex_{\U|\X}$ such that $\U = \{0,1\}$, $U \sim \textsf{\small Bernoulli}\!\left(\frac{1}{2}\right)$, and $U \rightarrow X \rightarrow Y$ form a Markov chain. Thus, the small input $f$-divergence constraint in the $f$-divergence formulation of $\eta_{f}\!\left(P_{X},P_{Y|X}\right)$ corresponds to the small $I_{f}(U;X)$ and $U \sim \textsf{\small Bernoulli}\!\left(\frac{1}{2}\right)$ constraints in \eqref{Eq: Local Mutual f-Information Contraction}. 

Lastly, consider the trajectory of input pmfs $R_{X}^{(\epsilon)} = P_X + \epsilon \, K_X^{*} \, \textsf{\small diag}\!\left(\sqrt{P_X}\right)$, where $\epsilon > 0$ is sufficiently small, and $K_{X}^{*} \in \left(\R^{|\X|}\right)^{\! *}$ is the left singular vector corresponding to the second largest singular value of the DTM $B$ (see \eqref{Eq: Contraction Coefficient as Singular Value}). As the proof in Appendix \ref{App: Proof of Local Approximation of Contraction Coefficients} illustrates, this trajectory satisfies $\lim_{\epsilon \rightarrow 0}{D_{f}(R_X^{(\epsilon)}||P_X)} = 0$ and achieves $\etaChi\!\left(P_{X},P_{Y|X}\right)$ in Theorem \ref{Thm:Local Approximation of Contraction Coefficients}:
\begin{equation}
\lim_{\epsilon \rightarrow 0}{\frac{D_{f}(R_X^{(\epsilon)} W||P_X W)}{D_{f}(R_X^{(\epsilon)}||P_X)}} = \etaChi\!\left(P_{X},P_{Y|X}\right) .
\end{equation}
The corresponding trajectory of conditional distributions for \eqref{Eq: Local Mutual f-Information Contraction} is $\{P_{X|U = u}^{(\epsilon)} = P_X + (2u - 1) \, \epsilon \, K_X^{*} \,\textsf{\small diag}\!\left(\sqrt{P_X}\right):u \in \{0,1\}\}$, where $\epsilon > 0$ is sufficient small. This trajectory satisfies $\lim_{\epsilon \rightarrow 0}{I_{f}(P_U,P_{X|U}^{(\epsilon)})} = 0$, produces $P_X$ after $(P_U,P_{X|U}^{(\epsilon)})$ is marginalized, and achieves $\etaChi\!\left(P_{X},P_{Y|X}\right)$ in \eqref{Eq: Local Mutual f-Information Contraction}:
\begin{equation}
\lim_{\epsilon \rightarrow 0}{\frac{I_{f}(P_U,P_{Y|U}^{(\epsilon)})}{I_{f}(P_U,P_{X|U}^{(\epsilon)})}} = \etaChi\!\left(P_{X},P_{Y|X}\right)
\end{equation}
where $U \sim \textsf{\small Bernoulli}\!\left(\frac{1}{2}\right)$, and $P_{Y|U}^{(\epsilon)} = P_{X|U}^{(\epsilon)} P_{Y|X}$ as row stochastic matrices.

\subsection{Linear Bounds between Contraction Coefficients}
\label{Linear Bounds between Contraction Coefficients}

For any joint pmf $P_{X,Y}$ with $P_X \in \Simplex_{\X}^{\circ}$ and $P_Y \in \Simplex_{\Y}^{\circ}$, our next result provides a linear upper bound on $\eta_{f}\!\left(P_{X},P_{Y|X}\right)$ using $\etaChi\!\left(P_{X},P_{Y|X}\right)$ for a certain class of $f$-divergences.  

\begin{theorem}[Contraction Coefficient Bound] 
\label{Thm:General Contraction Coefficient Bound}
Suppose we are given a continuous convex function $f:[0,\infty) \rightarrow \R$ that is thrice differentiable at unity with $f(1) = 0$ and $f^{\prime \prime}(1) > 0$, and satisfies \eqref{Eq:General Pinsker Condition} for every $t \in (0,\infty)$ (see subsection \ref{Bounds on f-Divergences using Chi-Squared Divergence}). Suppose further that the difference quotient $g:(0,\infty) \rightarrow \R$, defined as $g(x) = \frac{f(x) - f(0)}{x}$, is concave. Then, we have:
$$ \eta_{f}\!\left(P_{X},P_{Y|X}\right) \leq \frac{f^{\prime}(1) + f(0)}{\displaystyle{f^{\prime \prime}(1) \min_{x \in \X}{P_X(x)}}} \, \etaChi\!\left(P_{X},P_{Y|X}\right) . $$
\end{theorem}

Theorem \ref{Thm:General Contraction Coefficient Bound} is proved in subsection \ref{Upper Bounds on Contraction Coefficients}. The conditions on $f$ ensure that the resulting $f$-divergence exhibits the properties of KL divergence required by the proof of Theorem \ref{Thm:Refined Contraction Coefficient Bound} (see below). So, a similar proof technique also works for Theorem \ref{Thm:General Contraction Coefficient Bound}. A straightforward specialization of this theorem for KL divergence (which we first proved in the conference version of this paper \cite[Theorem 10]{BoundsbetweenContractionCoefficients}) is presented next.

\begin{corollary}[KL Contraction Coefficient Bound] 
\label{Cor:Contraction Coefficient Bound}
$$ \etaKL\!\left(P_{X},P_{Y|X}\right) \leq \frac{\etaChi\!\left(P_{X},P_{Y|X}\right)}{\displaystyle{\min_{x \in \X}{P_X(x)}}} \, . $$
\end{corollary}

\begin{proof}
This can be recovered from Theorem \ref{Thm:General Contraction Coefficient Bound} by verifying that $f(t) = t \log(t)$ satisfies the conditions of Theorem \ref{Thm:General Contraction Coefficient Bound}, cf. \cite{fDivergencePinsker}. See Appendix \ref{App: Proof of Contraction Coefficient Bound from General Contraction Coefficient Bound} for details.
\end{proof}

The constant in this upper bound on $\etaKL\!\left(P_{X},P_{Y|X}\right)$ can be improved, and the ensuing theorem presents this improvement.

\begin{theorem}[Refined KL Contraction Coefficient Bound] 
\label{Thm:Refined Contraction Coefficient Bound}
$$ \etaKL\!\left(P_{X},P_{Y|X}\right) \leq \frac{2 \, \etaChi\!\left(P_{X},P_{Y|X}\right)}{\displaystyle{\phi\!\left(\max_{A \subseteq \X}\,{\pi(A)}\right) \min_{x \in \X}{P_X(x)}}} $$
where $\pi(A) \triangleq \min\!\left\{P_{X}(A),1-P_{X}(A)\right\}$ for any $A \subseteq \X$, and the function $\phi:\left[0,\frac{1}{2}\right] \rightarrow \R$ is defined in \eqref{Eq: phi Function} (see subsection \ref{Bounds on f-Divergences using Chi-Squared Divergence}). 
\end{theorem}

Theorem \ref{Thm:Refined Contraction Coefficient Bound} is also proved in subsection \ref{Upper Bounds on Contraction Coefficients}, and it is tighter than the bound in Corollary \ref{Cor:Contraction Coefficient Bound} due to \eqref{Eq: phi Bound} in subsection \ref{Bounds on f-Divergences using Chi-Squared Divergence}. We now make some pertinent remarks about Corollary \ref{Cor:Contraction Coefficient Bound} and Theorems \ref{Thm:General Contraction Coefficient Bound} and \ref{Thm:Refined Contraction Coefficient Bound}. 

\begin{figure}[!t]
\centering

\subfloat[Plots of $\etaChi(P_{X},P_{Y|X})$ (blue mesh), $\etaKL(P_{X},P_{Y|X})$ (red mesh), and linear upper bounds on $\etaKL(P_{X},P_{Y|X})$. The green mesh denotes the upper bound from Corollary \ref{Cor:Contraction Coefficient Bound}, and the yellow mesh denotes the tighter upper bound from Theorem \ref{Thm:Refined Contraction Coefficient Bound}.]{%
\includegraphics[trim = 0mm 0mm 0mm 0mm, clip, width=\linewidth]{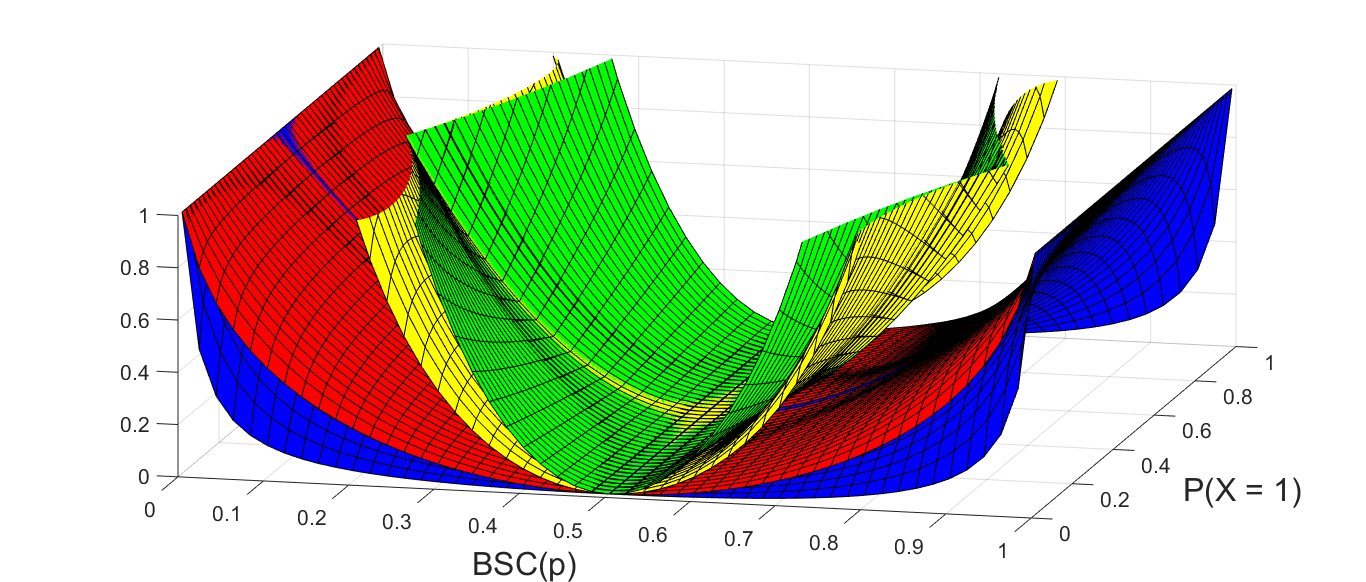}}

\subfloat[Plots of upper bounds on the ratio $\etaKL(P_{X},P_{Y|X})/\etaChi(P_{X},P_{Y|X})$, denoted by the red mesh. The bound $1/\min_{x \in \X}{P_X(x)}$ from Corollary \ref{Cor:Contraction Coefficient Bound} is the green mesh, and the bound $2/(\phi(\max_{A \subseteq \X}{\pi(A)})\min_{x \in \X}\!{P_X(x)})$ from Theorem \ref{Thm:Refined Contraction Coefficient Bound} is the blue mesh.]{%
\includegraphics[trim = 0mm 0mm 0mm 0mm, clip, width=\linewidth]{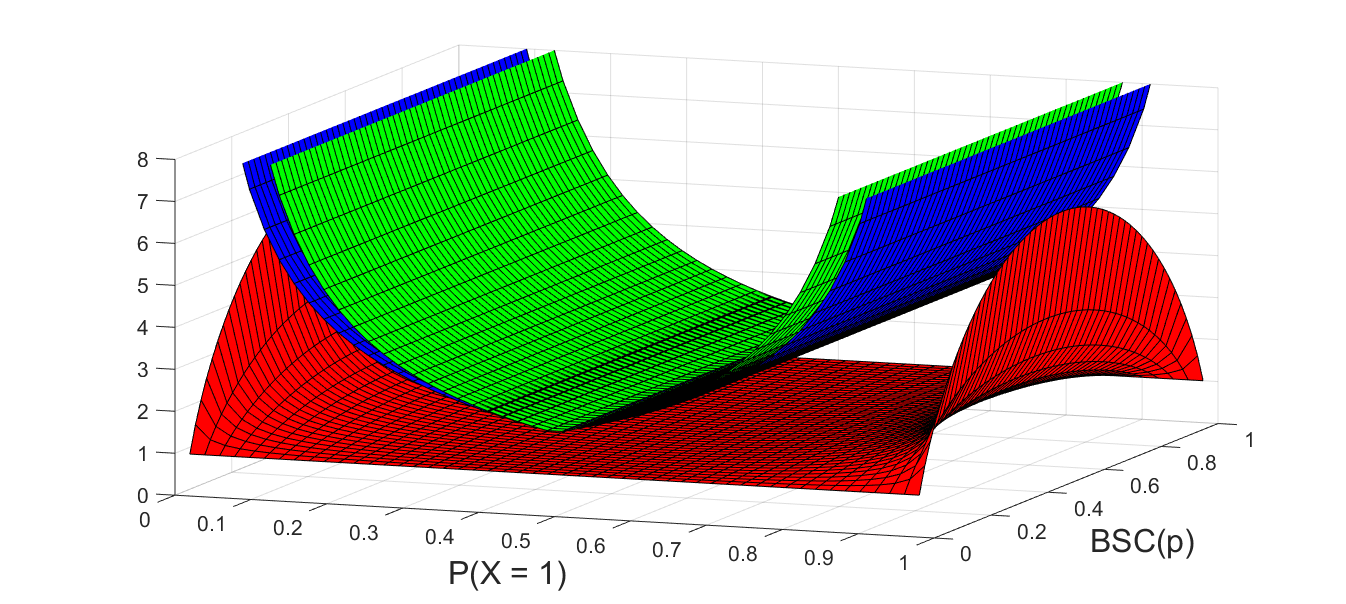}}

\caption{Plots of the contraction coefficient bounds in Corollary \ref{Cor:Contraction Coefficient Bound} and Theorem \ref{Thm:Refined Contraction Coefficient Bound} for a BSC, $P_{Y|X}$, with crossover probability $p \in [0,1]$, and input random variable $X \sim \textsf{{\scriptsize Bernoulli}}(\P(X = 1))$.}
\label{Fig:Contraction Coefficient Bounds}
\end{figure}

Firstly, as shown in Figure \ref{Fig:Contraction Coefficient Bounds}(a), the upper bounds in these results can be strictly less than the trivial bound of unity. For example, when $(X,Y) \sim \textsf{\small DSBS}(p)$ for some $p \in [0,1]$ (which is a slice along $\P(X = 1) = \frac{1}{2}$ in Figure \ref{Fig:Contraction Coefficient Bounds}(a)), the upper bounds in Corollary \ref{Cor:Contraction Coefficient Bound} and Theorem \ref{Thm:Refined Contraction Coefficient Bound} are both equal to:
\begin{equation}
\label{Eq: Binary Upper Bound}
\frac{2 \, \etaChi\!\left(P_{X},P_{Y|X}\right)}{\displaystyle{\phi\!\left(\max_{A \subseteq \X}\,{\pi(A)}\right) \min_{x \in \X}{P_X(x)}}} = \frac{\etaChi\!\left(P_X,P_{Y|X}\right)}{\displaystyle{\min_{x \in \X}{P_X(x)}}} = 2(1-2p)^2
\end{equation}
using \eqref{Eq: Equality of Contraction Coefficients} and the fact that $\max_{A \subseteq \X}{\pi(A)} = \frac{1}{2}$. This upper bound is tighter than the trivial bound of unity when:
\begin{equation}
\label{Eq: Greater than 1}
2(1-2p)^2 < 1 \quad \Leftrightarrow \quad \frac{2-\sqrt{2}}{4} < p < \frac{2+\sqrt{2}}{4} \, .
\end{equation}
We also note that this upper bound is not achieved with equality in this scenario since $\etaKL\!\left(P_X,P_{Y|X}\right) = \etaChi\!\left(P_X,P_{Y|X}\right) = (1-2p)^2$, as shown in \eqref{Eq: Equality of Contraction Coefficients}.  

Secondly, our proofs of Theorems \ref{Thm:General Contraction Coefficient Bound} and \ref{Thm:Refined Contraction Coefficient Bound} will rely on extensions of the well-known \textit{Pinsker's inequality} (or Csisz\'{a}r-Kemperman-Kullback-Pinsker inequality, cf. \cite[Section V]{Verdu2014}) which upper bound TV distance using KL and other $f$-divergences. So, it is natural to ask: Are these bounds tighter than the TV distance contraction bound in part 6 of Proposition \ref{Prop:Contraction Coefficient Properties 2}? As the ensuing example illustrates, our bounds are tighter in certain regimes. Let $(X,Y) \sim \textsf{\small DSBS}(p)$ for some $p \in [0,1]$. Then, \eqref{Eq: Binary Upper Bound} presents the upper bounds in Corollary \ref{Cor:Contraction Coefficient Bound} and Theorem \ref{Thm:Refined Contraction Coefficient Bound}, and the TV distance contraction bound is:
\begin{equation}
\etaKL\!\left(P_{X},P_{Y|X}\right) \leq \etaKL\!\left(P_{Y|X}\right) \leq \etaTV\!\left(P_{Y|X}\right) = |1-2p| 
\end{equation}
using Definition \ref{Def:Contraction Coefficient 2}, part 6 of Proposition \ref{Prop:Contraction Coefficient Properties 2}, and \eqref{Eq: DSBS TV Bound}. Hence, our bound in \eqref{Eq: Binary Upper Bound} is tighter than the $\etaTV$ bound when:
\begin{equation}
2 (1-2p)^2 < |1-2p| \enspace \Leftrightarrow \enspace \frac{1}{4} < p < \frac{1}{2} \text{ or } \frac{1}{2} < p < \frac{3}{4} \, .
\end{equation}
Since our upper bounds can be greater than $1$ (see \eqref{Eq: Greater than 1}), we cannot hope to beat the $\etaTV$ ($\leq 1$) bound in all regimes. On the other hand, one advantage of our upper bounds is that they ``match'' the $\etaChi$ lower bound in part 6 of Proposition \ref{Prop:Contraction Coefficient Properties}; we will illustrate a useful application of this in subsection \ref{Markov Convergence}.

Thirdly, we intuitively expect a bound between contraction coefficients to depend on the cardinalities $|\X|$ or $|\Y|$. Since the minimum probability term in all our upper bounds satisfies $1/\!\min_{x \in \X}{P_X(x)} \geq |\X|$, we can superficially construe it as ``modeling'' $|\X|$. Unfortunately, this intuition is quite misleading. Simulations for the binary case, depicted in Figure \ref{Fig:Contraction Coefficient Bounds}(b), illustrate that the ratio $\etaKL\!\left(P_{X},P_{Y|X}\right)\!/\etaChi\!\left(P_{X},P_{Y|X}\right)$ increases significantly near the boundary of $\Simplex_{\X}$ when any of the probability masses of $P_X$ is close to $0$. This effect, while unsurprising given the skewed nature of probability simplices at their boundaries with respect to KL divergence as the distance measure, is correctly captured by the upper bounds in Corollary \ref{Cor:Contraction Coefficient Bound} and Theorem \ref{Thm:Refined Contraction Coefficient Bound} because $1/\!\min_{x \in \X}{P_X(x)}$ increases when any of the input probability masses tends to $0$ (see Figure \ref{Fig:Contraction Coefficient Bounds}(b)). Clearly, linear upper bounds on $\eta_{f}\!\left(P_{X},P_{Y|X}\right)$ that are purely in terms of $|\X|$ or $|\Y|$ cannot capture this effect. This gives credence to the existence of the minimum probability term in our linear bounds. 

Finally, we note that the inequality $1/\!\min_{x \in \X}{P_X(x)} \geq |\X|$ does not preclude the possibility of $1/\!\min_{x \in \X}{P_X(x)}$ being much larger than $|\X|$. So, our bounds can become loose when $|\X|$ is large (see the example in subsection \ref{Tensorization of Bounds between Contraction Coefficients}). As a result, the bounds in Theorem \ref{Thm:General Contraction Coefficient Bound}, Corollary \ref{Cor:Contraction Coefficient Bound}, and Theorem \ref{Thm:Refined Contraction Coefficient Bound} are usually of interest in the following settings:
\begin{enumerate}
\item $|\X|$ and $|\Y|$ are small: Figure \ref{Fig:Contraction Coefficient Bounds} portrays that our bounds can be quite tight when $|\X| = |\Y| = 2$.
\item Weak dependence i.e. $I(X;Y)$ is small: This situation naturally arises in the analysis of ergodicity of Markov chains\textemdash{}see subsection \ref{Markov Convergence}.
\item Product Distributions: If the underlying joint pmf is a product pmf, we can exploit tensorization of contraction coefficients (Proposition \ref{Prop:Contraction Coefficient Properties} part 4)\textemdash{}see subsection \ref{Tensorization of Bounds between Contraction Coefficients}.
\end{enumerate}

\subsection{Contraction Coefficients of Gaussian Random Variables}
\label{Contraction Coefficients of Gaussian Random Variables}

In this subsection, we consider contraction coefficients for KL and $\chi^2$-divergences corresponding to Gaussian source-channel pairs. Suppose $X$ and $Y$ are jointly Gaussian random variables. Their joint distribution has three possible forms:
\begin{enumerate}
\item $X$ or $Y$ are constants $a.s.$, and we define the contraction coefficients to be $\etaKL\!\!\left(P_X,P_{Y|X}\right) = \etaChi\!\!\left(P_X,P_{Y|X}\right) = 0$.
\item $a X + b Y = c \enspace a.s.$ for some constants $a,b,c \in \R$ such that $a \neq 0$ and $b \neq 0$. Here, it is straightforward to verify that $\rho(X;Y)$ = 1, which implies that $\etaKL\!\left(P_X,P_{Y|X}\right) = \etaChi\!\left(P_X,P_{Y|X}\right) = 1$.\footnote{Note that Definition \ref{Def:Maximal Correlation} holds for general random variables, and \eqref{Eq: Maximal Correlation as Contraction Coefficient} and part 6 of Proposition \ref{Prop:Contraction Coefficient Properties} (which also hold generally\textemdash{}see \cite[Equations (9) and (13)]{GraphSDPI}) can be used to conclude $\etaKL\!\left(P_X,P_{Y|X}\right) = \etaChi\!\left(P_X,P_{Y|X}\right) = 1$.}
\item The joint probability density function (pdf) $P_{X,Y}$ exists with respect to the Lebesgue measure.
\end{enumerate}
The final non-degenerate case is our regime of interest. For simplicity, we will assume that $X$ and $Y$ are zero-mean, and analyze the classical \textit{additive white Gaussian noise} (AWGN) channel model \cite[Chapter 9]{CoverThomas}:
\begin{equation}
\label{Eq: AWGN}
Y = X + W, \quad X \indep W
\end{equation}
where the input is $X \sim \Gauss(0,\sigma_X^2)$ with $\sigma_X^2 > 0$ (i.e. $X$ has a Gaussian pdf with mean $0$ and variance $\sigma_X^2$), the Gaussian noise is $W \sim \Gauss(0,\sigma_W^2)$ with $\sigma_W^2 > 0$, and $X$ is independent of $W$. This relation also defines the channel conditional pdfs $\{P_{Y|X = x} = \Gauss(x,\sigma_W^2) : x \in \R\}$. 

For the jointly Gaussian pdf $P_{X,Y}$ define above, the contraction coefficients for KL and $\chi^2$-divergences are given by (cf. \eqref{Eq:Contraction Coefficient for KL Divergence} and \eqref{Eq:Contraction Coefficient for Chi Squared Divergence}): 
\begin{align}
\label{Eq: Gaussian KL Contraction}
\etaKL\!\left(P_{X},P_{Y|X}\right) & = \sup_{\substack{R_X:\\0 < D(R_X || P_X) < +\infty}}{\frac{D(R_Y||P_Y)}{D(R_X||P_X)}} \\
\etaChi\!\left(P_{X},P_{Y|X}\right) & = \sup_{\substack{R_X:\\0 < \chi^2(R_X || P_X) < +\infty}}{\frac{\chi^2(R_Y||P_Y)}{\chi^2(R_X||P_X)}}
\label{Eq: Gaussian Chi Squared Contraction} 
\end{align}
where the suprema are over pdfs $R_X$ (which differ from $P_X$ on a set with non-zero Lebesgue measure),\footnote{When $P_X$ is a general probability measure and $P_{Y|X}$ is a Markov kernel between two measurable spaces, the contraction coefficients for KL and $\chi^2$-divergences are defined exactly as in \eqref{Eq:Contraction Coefficient for KL Divergence} and \eqref{Eq:Contraction Coefficient for Chi Squared Divergence} using the measure theoretic definitions of KL and $\chi^2$-divergences \cite[Section 2]{GraphSDPI}. In \eqref{Eq: Gaussian KL Contraction}, when we optimize over all probability measures $R_X$ on $\R$ (with its Borel $\sigma$-algebra), the constraint $D(R_X||P_X) < +\infty$ implies that $R_X$ must be absolutely continuous with respect to the Gaussian distribution $P_X$, cf. \cite[Section 1.6]{InfoTheoryNotes}. Hence, the supremum in \eqref{Eq: Gaussian KL Contraction} can be taken over all pdfs $R_X$ such that $0 < D(R_X||P_X) < +\infty$. A similar argument applies for \eqref{Eq: Gaussian Chi Squared Contraction}. (Note that KL and $\chi^2$-divergences for pdfs are defined just as in \eqref{Eq:KL Divergence} and \eqref{Eq:Chi-Squared Divergence} with Lebesgue integrals replacing summations.)} and $R_Y$ denotes the marginal pdf of $Y$ after passing $R_X$ through the AWGN channel $P_{Y|X}$. In particular, $R_Y = R_X \star \Gauss(0,\sigma_W^2)$, where $\star$ denotes the \textit{convolution} operation. Furthermore, we define the \textit{contraction coefficient for KL divergence with average power constraint} $p \geq \sigma_X^2$ as:
\begin{equation}
\label{Eq: Contraction Coefficient with Power Constraint}
\etaKLp\!\!\left(P_{X},P_{Y|X}\right) \triangleq \sup_{\substack{R_X: \, \E_{R_X}\!\left[X^2\right] \leq p\\0 < D(R_X||P_X) < +\infty}}{\frac{D(R_Y||P_Y)}{D(R_X||P_X)}} 
\end{equation}
where the supremum is over pdfs $R_X$ satisfying the average power constraint $\E\!\left[X^2\right] \leq p$. Note that setting $p = +\infty$ yields the standard contraction coefficient in \eqref{Eq: Gaussian KL Contraction}.

It is well-known in the literature that $\etaKL\!\left(P_{X},P_{Y|X}\right) = \etaChi\!\left(P_{X},P_{Y|X}\right)$ for the jointly Gaussian pdf $P_{X,Y}$ in \eqref{Eq: AWGN}. For example, \cite[Theorem 7]{ErkipCover} proves this result in the context of investment portfolio theory, \cite[p.2]{Nair2014} proves a generalization of it in the context of Gaussian hypercontractivity, and \cite[Section 5.2, part 5]{Kimetal2017} proves it in an effort to axiomatize $\etaKL$. While the proofs in \cite[Theorems 6 and 7]{ErkipCover} and \cite[Section 5.2, part 5]{Kimetal2017} use the mutual information characterization of $\etaKL$ in \eqref{Eq: MI Contraction} (cf. \cite[Theorem 4]{GraphSDPI}), we provide an alternative proof of this result in section \ref{Proof of Gaussian Contraction Coefficients} that directly uses the KL divergence definition of $\etaKL$ in \eqref{Eq: Gaussian KL Contraction}. Furthermore, our proof also establishes that $\etaKLp\!\!\left(P_{X},P_{Y|X}\right)$ equals $\etaChi\!\left(P_{X},P_{Y|X}\right)$ for every $p \in \left[\sigma_X^2,\infty\right]$. Although this latter result follows easily from our proof, it has not explicitly appeared in the literature to our knowledge. The ensuing theorem states these results formally.

\begin{theorem}[Gaussian Contraction Coefficients] 
\label{Thm: Gaussian Contraction Coefficients}
Given the jointly Gaussian pdf $P_{X,Y}$, defined via \eqref{Eq: AWGN} with source $P_X = \Gauss(0,\sigma_X^2)$ and channel $\{P_{Y|X = x} = \Gauss(x,\sigma_W^2) : x \in \R\}$ such that $\sigma_X^2,\sigma_W^2 > 0$, the following quantities are equivalent:
\begin{align*}
\etaKL\!\left(P_{X},P_{Y|X}\right) & = \etaKLp\!\!\left(P_{X},P_{Y|X}\right) \\
& = \etaChi\!\left(P_{X},P_{Y|X}\right) = \frac{\sigma_X^2}{\sigma_X^2 + \sigma_W^2} 
\end{align*}
where the average power constraint $p \geq \sigma_X^2$.
\end{theorem}

As mentioned earlier, we prove this in section \ref{Proof of Gaussian Contraction Coefficients}. In contrast to Theorem \ref{Thm: Gaussian Contraction Coefficients}, where $\etaKL\!\left(P_{X},P_{Y|X}\right)$ and $\etaKLp\!\!\left(P_{X},P_{Y|X}\right)$ can both be strictly less than $1$, we note that the contraction coefficients for KL divergence of channels (i.e. the setting of Definition \ref{Def:Contraction Coefficient 2}) are equal to $1$ regardless of whether we impose power constraints, cf. \cite[Section 1.2]{ContractionCoefficients} and \cite[Section 1]{CalmonPolyanskiyWu2018}. 

\subsection{Less Noisy Preorder and Operator Convexity}
\label{Less Noisy Preorder and Operator Convexity}

Our last main result presents an equivalent characterization of the less noisy preorder over channels that generalizes the result in Proposition \ref{Prop: Operator Convex f-Divergence Contraction}. We begin by defining the less noisy preorder. Within the finite alphabet setting of subsection \ref{f-Divergence}, consider an input random variable $X \in \X$, and two output random variables $Y \in \Y$ and $Z \in \Z$, where $\Z \triangleq \{1,\dots,|\Z|\}$ such that $2 \leq |\Z| < +\infty$. Let $P_{Y|X}$ and $P_{Z|X}$ be any two channels with the same input alphabet $\X$, and corresponding row stochastic transition probability matrices $W \in \Simplex_{\Y|\X}$ and $V \in \Simplex_{\Z|\X}$, respectively. We say that $P_{Y|X}$ is \textit{less noisy} than $P_{Z|X}$, denoted $P_{Y|X} \succeq_{\textsf{\tiny ln}} P_{Z|X}$, if and only if:
\begin{equation}
\label{Eq: Less Noisy Definition}
D(R_X W || P_X W) \geq D(R_X V || P_X V)
\end{equation}
for every pair of pmfs $R_X,P_X \in \Simplex_{\X}$ \cite{ChannelPartialOrders}. It is straightforward to verify that \eqref{Eq: Less Noisy Definition} defines a preorder over channel matrices. Moreover, this definition conveys that the pair of pmfs $R_X W$ and $P_X W$ is ``more distinguishable'' than the pair $R_X V$ and $P_X V$, which indeed intuitively corresponds to $W$ being ``less noisy'' than $V$. There are several other equivalent characterizations of $\succeq_{\textsf{\tiny ln}}$; for example, via channel coding \cite[Definition B, Proposition 2]{ChannelPartialOrders}, mutual information \cite[Proposition 2]{ChannelPartialOrders}, and the van Dijk functional \cite[Theorem 2]{vanDijk1997}. We refer readers to \cite[Sections I-B, I-D, II-A, IV]{SymmetricChannelDomination} and the references therein for further details on the less noisy preorder.

The authors of \cite[Section 6]{GraphSDPI} illustrate that less noisy domination of a given channel by an erasure channel is closely related to the contraction coefficient for KL divergence of the channel. Recall that $E_{1-\beta}$ denotes an $|\X|$-ary erasure channel with erasure probability $1-\beta \in [0,1]$, input alphabet $\X$, and output alphabet $\X \cup \{\textsf{\small e}\}$ (as defined towards the end of subsection \ref{Contraction Coefficients of Source-Channel Pairs}). It can be deduced from \cite[Proposition 15]{GraphSDPI} that for any channel $P_{Y|X}$:\footnote{Note that when $\beta = 1$, $E_{0}$ is the identity channel and $E_{0} \succeq_{\textsf{\tiny ln}} P_{Y|X}$.}
\begin{equation}
\label{Eq: Less Noisy Characterization}
\etaKL\!\left(P_{Y|X}\right) = \inf\!\left\{\beta \in [0,1] : E_{1-\beta} \succeq_{\textsf{\tiny ln}} P_{Y|X}\right\} .
\end{equation}
In \cite[Section IV-A]{SymmetricChannelDomination}, the authors note that while \eqref{Eq: Less Noisy Characterization} conveys that $\etaKL$ characterizes less noisy domination by erasure channels, \eqref{Eq: KL Contraction and Chi-Squared Contraction} portrays that $\etaChi$ also characterizes this domination. This begs the question: Does $\chi^2$-divergence characterize the less noisy preorder in general? To answer this question, \cite[Theorem 1]{SymmetricChannelDomination} (which we present later as Lemma \ref{Lemma: Chi-Squared Divergence Characterization of Less Noisy}) characterizes $\succeq_{\textsf{\tiny ln}}$ using $\chi^2$-divergence, thereby generalizing \eqref{Eq: KL Contraction and Chi-Squared Contraction}.

Inspired by these results, we consider Proposition \ref{Prop: Operator Convex f-Divergence Contraction}, which shows that $\etaKL\!\left(P_{Y|X}\right) = \eta_{f}\!\left(P_{Y|X}\right)$ for all non-linear operator convex functions $f$ (defined in subsection \ref{Operator Convex Functions}). The ensuing theorem generalizes both Proposition \ref{Prop: Operator Convex f-Divergence Contraction} and \cite[Theorem 1]{SymmetricChannelDomination}, and portrays that non-linear operator convex $f$-divergences also characterize the less noisy preorder. 

\begin{theorem}[Equivalent Characterizations of $\succeq_{\textsf{\tiny ln}}$]
\label{Thm: Equivalent Characterizations of Less Noisy Preorder}
Consider any non-linear operator convex function $f:(0,\infty) \rightarrow \R$ such that $f(1) = 0$. Then, for any two channels $P_{Y|X}$ and $P_{Z|X}$ on the same input alphabet $\X$ with row stochastic transition probability matrices $W \in \Simplex_{\Y|\X}$ and $V \in \Simplex_{\Z|\X}$ respectively, $P_{Y|X} \succeq_{\textsf{\tiny ln}} P_{Z|X}$ if and only if:
$$ D_{f}(R_X W||P_X W) \geq D_{f}(R_X V||P_X V) $$
for every pair of input pmfs $R_X,P_X \in \Simplex_{\X}$.
\end{theorem} 

Theorem \ref{Thm: Equivalent Characterizations of Less Noisy Preorder} is proved in subsection \ref{Proof of Theorem: Equivalent Characterizations of Less Noisy Preorder} using techniques from \cite{ChoiRuskaiSeneta1994}. It is well-known that $f(t) = t \log(t)$ and $f(t) = \frac{t^{\alpha} - 1}{\alpha - 1}$ for any $\alpha \in (0,1) \cup (1,2]$ are operator convex functions (see \cite[Theorems V.2.5 and V.2.10, Exercises V.2.11 and V.2.13]{Bhatia1997}, and apply the affine transformation property in subsection \ref{Operator Convex Functions} appropriately). Hence, one class of $f$-divergences that satisfy the conditions of the theorem are the Hellinger divergences of order $\alpha \in (0,2]$, where the cases $\alpha = 1$ and $\alpha = 2$ correspond to KL and $\chi^2$-divergences, respectively.

\section{Proofs of Linear Bounds between Contraction Coefficients}
\label{Proofs of Linear Bounds between Contraction Coefficients}

In this section, we will prove Theorems \ref{Thm:General Contraction Coefficient Bound} and \ref{Thm:Refined Contraction Coefficient Bound}. The central idea to establish these results entails upper and lower bounding the $f$-divergences in the numerator and denominator of Definition \ref{Def:Contraction Coefficient} respectively, using $\chi^2$-divergences. To this end, we will illustrate some simple bounds between $f$-divergences and $\chi^2$-divergence in the next subsection, and prove the main results in subsection \ref{Upper Bounds on Contraction Coefficients}.

\subsection{Bounds on $f$-Divergences using $\chi^2$-Divergence}
\label{Bounds on f-Divergences using Chi-Squared Divergence}

We first present bounds between KL divergence and $\chi^2$-divergence. To derive our lower bound on KL divergence, we will require the following ``distribution dependent refinement of Pinsker's inequality'' proved in \cite{TighterPinskerInequality}.

\begin{lemma}[Distribution Dependent Pinsker's Inequality {\cite[Theorem 2.1]{TighterPinskerInequality}}]
\label{Lemma: Distribution Dependent Pinsker's Inequality}
For any two pmfs $R_X,P_X \in \Simplex_{\X}$, we have:
$$ D(R_X||P_X) \geq \phi\!\left(\max_{A \subseteq \X}{\pi(A)}\right) \left\|R_X - P_X\right\|_{\textsf{\tiny TV}}^2 $$
where $\pi(A) = \min\!\left\{P_{X}(A),1-P_{X}(A)\right\}$ for any $A \subseteq \X$, and the function $\phi:\left[0,\frac{1}{2}\right] \rightarrow \R$ is defined as:
\begin{equation}
\label{Eq: phi Function}
\phi(p) \triangleq 
  \left\{ 
    \begin{array}{cl}
      \frac{1}{1-2p}\log\left(\frac{1-p}{p}\right) &, \enspace p \in \left[0,\frac{1}{2}\right) \\
      2 &, \enspace p = \frac{1}{2}
    \end{array}
	\right. .
\end{equation}
Moreover, this inequality uses the optimal distribution dependent constant in the sense that for any fixed $P_X \in \Simplex_{\X}$:
$$ \inf_{R_X \in \Simplex_{\X}\backslash\{P_X\}}{\frac{D(R_X||P_X)}{\left\|R_X - P_X\right\|_{\textsf{\tiny TV}}^2}} = \phi\!\left(\max_{A \subseteq \X}{\pi(A)}\right) . $$
\end{lemma}

Recall that Pinsker's inequality states that for any $R_X,P_X \in \Simplex_{\X}$ (see e.g. \cite[Lemma 11.6.1]{CoverThomas}):
\begin{equation}
\label{Eq:Pinsker's Inequality}
D(R_X||P_X) \geq 2 \left\|R_X - P_X\right\|_{\textsf{\tiny TV}}^2 .
\end{equation}
Hence, Lemma \ref{Lemma: Distribution Dependent Pinsker's Inequality} is tighter than Pinsker's inequality, because $0 \leq \max_{A \subseteq \X}{\pi(A)} \leq \frac{1}{2}$, and hence:
\begin{equation} 
\label{Eq: phi Bound}
\phi\!\left(\max_{A \subseteq \X}{\pi(A)}\right) \geq 2 
\end{equation}
with equality if and only if $\max_{A \subseteq \X}{\pi(A)} = \frac{1}{2}$, cf. \cite[Section III]{TighterPinskerInequality}. The ensuing lemma uses Lemma \ref{Lemma: Distribution Dependent Pinsker's Inequality} to lower bound KL divergence using $\chi^2$-divergence. 

\begin{lemma}[KL Divergence Lower Bound] 
\label{Lemma:Distribution Dependent KL Divergence Lower Bound}
Given any two pmfs $R_X,P_X \in \Simplex_{\X}$, we have:\footnote{Throughout this paper, when $\min_{x \in \X}{P_X(x)} = 0$ and $\chi^2(R_X||P_X) = +\infty$, we assume that $\min_{x \in \X}{P_X(x)} \, \chi^2(R_X||P_X) = 0$.}
$$ D(R_X||P_X) \geq \frac{\displaystyle{\phi\!\left(\max_{A \subseteq \X}{\pi(A)}\right) \min_{x \in \X}{P_X(x)}}}{2} \, \chi^2(R_X||P_X) $$
where $\pi(\cdot)$ and $\phi:\left[0,\frac{1}{2}\right] \rightarrow \R$ are defined in Lemma \ref{Lemma: Distribution Dependent Pinsker's Inequality}.
\end{lemma}

\begin{proof}
Observe that if $R_X = P_X$ or $\min_{x \in \X}{P_X(x)} = 0$, then the inequality is trivially satisfied. So, we assume without loss of generality that $R_X \neq P_X$ and $P_X \in \Simplex_{\X}^{\circ}$. 

Since $\chi^2$-divergence resembles a weighted $\ell^2$-norm, we first use Lemma \ref{Lemma: Distribution Dependent Pinsker's Inequality} to get the lower bound:
\begin{equation}
\label{Eq: Intermediate Pinsker}
D(R_X||P_X) \geq \phi\!\left(\max_{A \subseteq \X}{\pi(A)}\right) \frac{\left\|R_X - P_X\right\|_{1}^2}{4}
\end{equation}
where we use the $\ell^1$-norm characterization of TV distance given in \eqref{Eq: Total Variation Distance}. We next notice using \eqref{Eq:Chi-Squared Divergence} that:
\begin{align*}
\chi^2(R_X||P_X) & = \sum_{x \in \X}{\left|R_X(x) - P_X(x)\right| \left|\frac{R_X(x) - P_X(x)}{P_X(x)}\right|} \nonumber \\
& \leq \frac{\left\|R_X - P_X\right\|_{\infty}}{\displaystyle{\min_{x \in \X}{P_X(x)}}} \left\|R_X - P_X\right\|_{1} .
\end{align*}
This implies that:
\begin{align}
\frac{\left\|R_X - P_X\right\|_{1}^2}{\displaystyle{\min_{x \in \X}{P_X(x)}}} & \geq \chi^2(R_X||P_X) \frac{\left\|R_X - P_X\right\|_{1}}{\left\|R_X - P_X\right\|_{\infty}} \nonumber \\
& \geq \chi^2(R_X||P_X) \min_{\substack{S_X,Q_X \in \Simplex_{\X}\\S_X \neq Q_X}}{\frac{\left\|S_X - Q_X\right\|_{1}}{\left\|S_X - Q_X\right\|_{\infty}}} \nonumber \\
& = 2 \, \chi^2(R_X||P_X)
\label{Eq:Common Step} 
\end{align}
where we use the fact that:
\begin{equation}
\label{Eq: Helper min}
\min_{\substack{S_X,Q_X \in \Simplex_{\X}\\S_X \neq Q_X}}{\frac{\left\|S_X - Q_X\right\|_{1}}{\left\|S_X - Q_X\right\|_{\infty}}} = 2 \, . 
\end{equation}

To prove \eqref{Eq: Helper min}, note that for every $S_X,Q_X \in \Simplex_{\X}$ (see e.g. \cite[Lemma 1]{fDivergenceBoundsSason}):
$$ \left\|S_X - Q_X\right\|_{\infty} \leq \frac{1}{2} \left\|S_X - Q_X\right\|_{1} $$
because $(S_X - Q_X) \1 = 0$, and this inequality can in fact be tight. For example, choose any pmf $Q_X \in \Simplex_{\X}^{\circ}$ and let $x_0 = \argmin_{x \in \X}{Q_X(x)}$. Then, select $S_X \in \Simplex_{\X}$ such that $S_X(x_0) = Q_X(x_{0}) + \delta$ for some sufficiently small $\delta > 0$, $S_X(x_1) = Q_X(x_1) -\delta$ for some $x_1 \in \X\backslash\!\{x_0\}$, and $S_X(x) = Q_X(x)$ for every other $x \in \X\backslash\!\{x_0,x_1\}$. These choices of $S_X$ and $Q_X$ yield $\left\|S_X - Q_X\right\|_{\infty} = \delta = \frac{1}{2} \left\|S_X - Q_X\right\|_{1}$. 

Finally, combining \eqref{Eq: Intermediate Pinsker} and \eqref{Eq:Common Step}, we get:
$$ D(R_X||P_X) \geq \frac{\displaystyle{\phi\!\left(\max_{A \subseteq \X}{\pi(A)}\right) \min_{x \in \X}{P_X(x)}}}{2} \, \chi^2(R_X||P_X) $$
which completes the proof.
\end{proof}

We remark that if we apply \eqref{Eq: phi Bound} to Lemma \ref{Lemma:Distribution Dependent KL Divergence Lower Bound}, or equivalently, if we use the standard Pinsker's inequality \eqref{Eq:Pinsker's Inequality} instead of Lemma \ref{Lemma: Distribution Dependent Pinsker's Inequality} in the proof of Lemma \ref{Lemma:Distribution Dependent KL Divergence Lower Bound}, then we obtain the well-known weaker inequality (see e.g. \cite[Equation (338)]{fDivergenceBounds}):
\begin{equation}
\label{Eq:KL Divergence Lower Bound}
D\left(R_X||P_X\right) \geq \min_{x \in \X}{P_X(x)} \, \chi^2(R_X||P_X) 
\end{equation}
for every $R_X,P_X \in \Simplex_{\X}$.

It is worth mentioning that a systematic method of deriving optimal bounds between any pair of $f$-divergences is given by the \textit{Harremo\"{e}s-Vajda joint range} \cite{HarremoesVajda2011}. However, we cannot use this technique to derive lower bounds on KL divergence using $\chi^2$-divergence since no such general lower bound exists (when both input distributions vary) \cite[Section 7.3]{InfoTheoryNotes}. On the other hand, distribution dependent bounds can be easily found using ad hoc techniques.\footnote{A ``distribution dependent'' bound between two $f$-divergences is a bound that contains terms that are not either of the $f$-divergences, but depend on the input distributions.} Our proof of Lemma \ref{Lemma:Distribution Dependent KL Divergence Lower Bound} demonstrates one such ad hoc approach based on Pinsker's inequality.

It is tempting to try and improve Lemma \ref{Lemma:Distribution Dependent KL Divergence Lower Bound} by using better lower bounds on KL divergence in terms of TV distance. For example, the best possible lower bound on KL divergence via TV distance is the lower boundary of their \textit{Harremo\"{e}s-Vajda joint range}, cf. \cite[Figure 1]{HarremoesVajda2011}. This lower boundary, known as \textit{Vajda's tight lower bound}, gives the minimum possible KL divergence for each value of TV distance, and is completely characterized using a parametric formula in \cite[Theorem 1]{FedotovHarremoesTopsoe2003} (also see \cite[Section 7.2.2]{InfoTheoryNotes}). Although Vajda's tight lower bound yields a non-linear lower bound on KL divergence using $\chi^2$-divergence, this lower bound is difficult to apply in conjunction with Lemma \ref{Lemma:KL Divergence Upper Bound} (shown below) to obtain a non-linear upper bound on a ratio of KL divergences using a ratio of $\chi^2$-divergences (see the proof of Theorem \ref{Thm:Refined Contraction Coefficient Bound} in subsection \ref{Upper Bounds on Contraction Coefficients}). For this reason, we resort to using simple linear bounds between KL and $\chi^2$-divergence, which yields a linear bound in Theorem \ref{Thm:Refined Contraction Coefficient Bound}.

Another subtler reason for proving a linear lower bound on KL divergence using $\chi^2$-divergence is to exploit Lemma \ref{Lemma: Distribution Dependent Pinsker's Inequality}. Although Pinsker's inequality is the best lower bound on KL divergence using squared TV distance over all pairs of input pmfs (see e.g. \cite[Equation (9)]{FedotovHarremoesTopsoe2003}), the contraction coefficients in subsection \ref{Linear Bounds between Contraction Coefficients} have a fixed source pmf $P_X$. Therefore, we can use the distribution dependent improvement of Pinsker's inequality in Lemma \ref{Lemma: Distribution Dependent Pinsker's Inequality} to obtain a tighter bound than \eqref{Eq:KL Divergence Lower Bound}. 

We next present an upper bound on KL divergence using $\chi^2$-divergence which trivially follows from Jensen's inequality.
This bound was derived in the context of studying ergodicity of Markov chains in \cite{KLUpperBoundViaChiSquared}, and has been re-derived in the study of inequalities related to $f$-divergences, cf. \cite{KLChiSquaredInequalities,fDivergence}.

\begin{lemma}[KL Divergence Upper Bound {\cite{KLUpperBoundViaChiSquared}}] 
\label{Lemma:KL Divergence Upper Bound}
Given any two pmfs $P_X,R_X \in \Simplex_{\X}$, we have:
$$ D(R_X||P_X) \leq \log\!\left(1 + \chi^2(R_X||P_X)\right) \leq \chi^2(R_X||P_X) \, . $$
\end{lemma}

\begin{proof}
We provide a proof for completeness, cf. \cite{KLChiSquaredInequalities}. Assume without loss of generality that there does not exist $x \in \X$ such that $R_X(x) > P_X(x) = 0$. (If this is not the case, then $\chi^2(R_X||P_X) = +\infty$ and the inequalities are trivially true.) So, restricting $\X$ to be the support of $P_X$, we assume that $P_X \in \Simplex_{\X}^{\circ}$ (which ensures that none of the ensuing quantities are infinity). Since $x \mapsto \log(x)$ is a concave function, using Jensen's inequality, we have:
\begin{align*}
D(R_X||P_X) & = \sum_{x \in \X}{R_X(x)\log\!\left(\frac{R_X(x)}{P_X(x)}\right)} \\
& \leq \log\!\left(\sum_{x \in \X}{\frac{R_X(x)^2}{P_X(x)}}\right) \\
& = \log\!\left(1 + \chi^2(R_X||P_X)\right) \\
& \leq \chi^2(R_X||P_X)
\end{align*} 
where the third equality follows from \eqref{Eq:Chi-Squared Divergence} after some algebra, and the final inequality follows from the well-known inequality: $\log(1+x) \leq x$ for all $x>-1$. 
\end{proof}

We remark that the first non-linear bound in Lemma \ref{Lemma:KL Divergence Upper Bound} turns out to capture the Harremo\"{e}s-Vajda joint range \cite[Section 7.3]{InfoTheoryNotes}. Although it is tighter than the second linear bound, we will use the latter to prove Theorem \ref{Thm:Refined Contraction Coefficient Bound} (as explained earlier). 

We now present bounds between general $f$-divergences and $\chi^2$-divergence. To derive our lower bound on $f$-divergences, we first state a generalization of Pinsker's inequality for $f$-divergences that is proved in \cite{fDivergencePinsker}.

\begin{lemma}[Generalized Pinsker's Inequality for $f$-Divergence {\cite[Theorem 3]{fDivergencePinsker}}]
\label{Lemma:Pinsker's Inequality for f-Divergence} 
Suppose we are given a convex function $f:(0,\infty) \rightarrow \R$ that is thrice differentiable at unity with $f(1) = 0$ and $f^{\prime \prime}(1) > 0$, and satisfies:
\begin{equation}
\label{Eq:General Pinsker Condition}
\left(f(t) - f^{\prime}(1) (t-1)\right)\!\!\left(1 - \frac{f^{\prime \prime \prime}(1)}{3 f^{\prime \prime}(1)}(t-1)\right) \geq \frac{f^{\prime \prime}(1)}{2}(t-1)^2 
\end{equation}
for every $t \in (0,\infty)$. Then, we have for every $R_X,P_X \in \Simplex_{\X}$:
$$ D_{f}(R_X||P_X) \geq 2 \, f^{\prime \prime}(1) \left\|R_X - P_X\right\|_{\textsf{\tiny TV}}^2 \, . $$
Moreover, this inequality uses the optimal constant in the sense that:
$$ \inf_{\substack{R_X,P_X \in \Simplex_{\X}:\\R_X \neq P_X}}{\frac{D_{f}(R_X||P_X)}{\left\|R_X - P_X\right\|_{\textsf{\tiny TV}}^2}} = 2 \, f^{\prime \prime}(1) \, . $$
\end{lemma}

We remark that $f(t) = t \log(t)$ satisfies the conditions of Lemma \ref{Lemma:Pinsker's Inequality for f-Divergence} with $f^{\prime \prime}(1) = 1$ as shown in Appendix \ref{App: Proof of Contraction Coefficient Bound from General Contraction Coefficient Bound}; this yields the standard Pinsker's inequality presented in \eqref{Eq:Pinsker's Inequality}. Since \eqref{Eq:General Pinsker Condition} can be difficult to check for other $f$-divergences, the author of \cite{fDivergencePinsker} provides sufficient conditions for \eqref{Eq:General Pinsker Condition} in \cite[Corollary 4]{fDivergencePinsker}. (These conditions can be verified to yield a variant of Pinsker's inequality for \textit{R\'{e}nyi divergences} of order $\alpha \in (0,1)$ \cite[Corollary 6]{fDivergencePinsker}.) The ensuing lemma uses Lemma \ref{Lemma:Pinsker's Inequality for f-Divergence} to establish a lower bound on certain $f$-divergences using $\chi^2$-divergence which parallels Lemma \ref{Lemma:Distribution Dependent KL Divergence Lower Bound} (or more precisely, \eqref{Eq:KL Divergence Lower Bound}, since it follows from the standard Pinsker's inequality). 

\begin{lemma}[$f$-Divergence Lower Bound] 
\label{Lemma:f-Divergence Lower Bound}
Suppose we are given a convex function $f:(0,\infty) \rightarrow \R$ that is thrice differentiable at unity with $f(1) = 0$ and $f^{\prime \prime}(1) > 0$, and satisfies \eqref{Eq:General Pinsker Condition} for every $t \in (0,\infty)$. Then, for any two pmfs $R_X \in \Simplex_{\X}$ and $P_X \in \Simplex_{\X}^{\circ}$, we have:
$$ D_{f}(R_X||P_X) \geq f^{\prime \prime}(1) \min_{x \in \X}{P_X(x)} \, \chi^2(R_X||P_X) \, . $$
\end{lemma}

\begin{proof}
We follow the proof of Lemma \ref{Lemma:Distribution Dependent KL Divergence Lower Bound} mutatis mutandis. Assume without loss of generality that $R_X \neq P_X$. The generalized Pinsker's inequality for $f$-divergences in Lemma \ref{Lemma:Pinsker's Inequality for f-Divergence} yields:
$$ D_{f}(R_X||P_X) \geq \frac{f^{\prime \prime}(1)}{2} \left\|R_X - P_X\right\|_1^2 $$
using the $\ell^1$-norm characterization of TV distance in \eqref{Eq: Total Variation Distance}. Applying \eqref{Eq:Common Step} to this inequality produces the desired result. 
\end{proof}

Note that setting $f(t) = t \log(t)$ in Lemma \ref{Lemma:f-Divergence Lower Bound} gives \eqref{Eq:KL Divergence Lower Bound}.

Finally, we present an upper bound on certain $f$-divergences using $\chi^2$-divergence which is analogous to Lemma \ref{Lemma:KL Divergence Upper Bound}. This upper bound was proved in \cite[Lemma A.2]{SDPIandSobolevInequalities} with the assumption that $f$ is differentiable, but we only need to check differentiability at unity as seen below. (It is instructive for readers to revisit the proof of Lemma \ref{Lemma:KL Divergence Upper Bound} to see how the ensuing proof generalizes it for $f$-divergences.)

\begin{lemma}[$f$-Divergence Upper Bound {\cite[Lemma A.2]{SDPIandSobolevInequalities}}] 
\label{Lemma:f-Divergence Upper Bound}
Suppose we are given a continuous convex function $f:[0,\infty) \rightarrow \R$ that is differentiable at unity with $f(1) = 0$ such that the difference quotient $g:(0,\infty) \rightarrow \R$, defined as $g(x) = \frac{f(x) - f(0)}{x}$, is concave.\footnote{Since $f$ is convex, it is clearly continuous on $(0,\infty)$. So, the continuity assumption on $f$ only asserts that $f(0) = \lim_{t \rightarrow 0^+}{f(t)}$ (see Definition \ref{Def: f-Divergence}).} Then, for any two pmfs $R_X,P_X \in \Simplex_{\X}$, we have:
$$ D_{f}(R_X||P_X) \leq \left(f^{\prime}(1) + f(0)\right) \chi^2(R_X||P_X) \, . $$
\end{lemma}

\begin{proof}
We provide the proof in \cite{SDPIandSobolevInequalities} for completeness. As in the proof of Lemma \ref{Lemma:KL Divergence Upper Bound}, we may assume without loss of generality that $P_X \in \Simplex_{\X}^{\circ}$ so that none of the ensuing quantities are infinity. We then have the following sequence of equalities and inequalities:
\begin{align}
D_{f}(R_X||P_X) & = \sum_{x \in \X}{P_X(x) \, f\!\left(\frac{R_X(x)}{P_X(x)}\right)} \nonumber \\
& = f(0) + \sum_{x \in \X}{R_X(x) \, g\!\left(\frac{R_X(x)}{P_X(x)}\right)} \nonumber \\
& \leq f(0) + g\!\left( \sum_{x \in \X}{\frac{R_X(x)^2}{P_X(x)}}\right) \nonumber \\
\label{Eq: Non-linear f-Divergence Upper Bound}
& = f(0) + g\!\left( 1 + \chi^2(R_X||P_X)\right) \\
& \leq f(0) + g(1) + g^{\prime}(1) \chi^2(R_X||P_X) \nonumber \\
& = \left(f^{\prime}(1) + f(0)\right) \chi^2(R_X||P_X) \nonumber
\end{align}
where the second equality uses the convention $0\,g(0) = 0$, the first inequality follows from Jensen's inequality since $g:(0,\infty) \rightarrow \R$ is concave, the second inequality is also a consequence of the concavity of $g:(0,\infty) \rightarrow \R$ as shown in \cite[Section 3.1.3]{ConvexOptimization}, and the final equality holds because $g(1) = -f(0)$ (as $f(1) = 0$) and:
\begin{align*}
g^{\prime}(1) & = \lim_{\delta \rightarrow 0}{\frac{g(1 + \delta) + f(0)}{\delta}} \\
& = \lim_{\delta \rightarrow 0}{\frac{f(1 + \delta) + \delta f(0)}{\delta(1 + \delta)}} \\
& = \left(\lim_{\delta \rightarrow 0}{\frac{1}{1+\delta}}\right) \!\! \left(f(0) + \lim_{\delta \rightarrow 0}{\frac{f(1 + \delta)}{\delta}}\right) \\
& = f^{\prime}(1) + f(0) \, . 
\end{align*}
This completes the proof. 
\end{proof}

We note that \eqref{Eq: Non-linear f-Divergence Upper Bound} is the analog of the tighter (non-linear) bound in Lemma \ref{Lemma:KL Divergence Upper Bound}. Furthermore, we remark that $g(x) = \frac{f(x)}{x}$ (when it is assumed to be concave) is a valid definition for the function in Lemma \ref{Lemma:f-Divergence Upper Bound} instead of the difference quotient. The proof carries through with a constant of $f^{\prime}(1)$ instead of $f^{\prime}(1) + f(0)$. However, we choose the difference quotient to prove Lemma \ref{Lemma:f-Divergence Upper Bound} in view of the affine invariance property of $f$-divergences (cf. subsection \ref{f-Divergence}). It is easy to verify that the quantity $f^{\prime}(1) + f(0)$ is invariant to appropriate affine shifts, but the quantity $f^{\prime}(1)$ is not. We also remark that the constant $f^{\prime \prime}(1)$ in Lemma \ref{Lemma:f-Divergence Lower Bound} is invariant to appropriate affine shifts. 

\subsection{Proofs of Theorems \ref{Thm:General Contraction Coefficient Bound} and \ref{Thm:Refined Contraction Coefficient Bound}}
\label{Upper Bounds on Contraction Coefficients}

Recall from the outset of subsection \ref{Local Approximation of Contraction Coefficients} that we are given a joint pmf $P_{X,Y}$ such that $P_X \in \Simplex_{\X}^{\circ}$ and $P_Y \in \Simplex_{\Y}^{\circ}$. Moreover, we let $W \in \Simplex_{\Y|\X}$ denote the row stochastic transition probability matrix of the channel $P_{Y|X}$. Using Lemmata \ref{Lemma:Distribution Dependent KL Divergence Lower Bound} and \ref{Lemma:KL Divergence Upper Bound} from subsection \ref{Bounds on f-Divergences using Chi-Squared Divergence}, we can now prove Theorem \ref{Thm:Refined Contraction Coefficient Bound}.

\renewcommand{\proofname}{\bfseries \emph{Proof of Theorem \ref{Thm:Refined Contraction Coefficient Bound}}}

\begin{proof}
For every pmf $R_X \in \Simplex_{\X}$ such that $R_X \neq P_X$, we have:
$$ \frac{D(R_X W||P_Y)}{D(R_X||P_X)} \leq \frac{2 \, \chi^2(R_X W||P_Y)}{\displaystyle{\phi\!\left(\max_{A \subseteq \X}{\pi(A)}\right) \min_{x \in \X}{P_X(x)} \, \chi^2(R_X||P_X)}} $$
using Lemmata \ref{Lemma:Distribution Dependent KL Divergence Lower Bound} and \ref{Lemma:KL Divergence Upper Bound}, where $P_Y = P_X W$. Taking the supremum over all $R_X \neq P_X$ on both sides produces:
$$ \etaKL\!\left(P_{X},P_{Y|X}\right) \leq \frac{2 \, \etaChi\!\left(P_{X},P_{Y|X}\right)}{\displaystyle{\phi\!\left(\max_{A \subseteq \X}{\pi(A)}\right) \min_{x \in \X}{P_X(x)}}}  $$
using \eqref{Eq:Contraction Coefficient for KL Divergence} and \eqref{Eq:Contraction Coefficient for Chi Squared Divergence}. This completes the proof.
\end{proof}

We now make a few pertinent remarks. Firstly, applying \eqref{Eq:KL Divergence Lower Bound} instead of Lemma \ref{Lemma:Distribution Dependent KL Divergence Lower Bound} in the preceding proof yields Corollary \ref{Cor:Contraction Coefficient Bound}. 

Secondly, while the conference version of this paper proves Corollary \ref{Cor:Contraction Coefficient Bound} (see \cite[Theorem 10]{BoundsbetweenContractionCoefficients}), it also proves the following weaker upper bound on $\etaKL\!\left(P_X,P_{Y|X}\right)$ \cite[Theorem 9]{BoundsbetweenContractionCoefficients}:
\begin{equation}
\label{Eq:Looser Contraction Coefficient Bound}
\etaKL\!\left(P_{X},P_{Y|X}\right) \leq \frac{2}{\displaystyle{\min_{x \in \X}{P_X(x)}}} \, \etaChi\!\left(P_{X},P_{Y|X}\right)
\end{equation}
which is independently derived in \cite[Equation III.19]{SDPIandSobolevInequalities}. Our proof of \eqref{Eq:Looser Contraction Coefficient Bound} in \cite[Theorem 9]{BoundsbetweenContractionCoefficients} uses the ensuing variant of \eqref{Eq:KL Divergence Lower Bound} that is looser by a factor of $2$, cf. \cite[Lemma 6]{BoundsbetweenContractionCoefficients}:
\begin{equation}
\label{Eq:Looser Lower Bound on KL Divergence}
D(S_X||Q_X) \geq \frac{\displaystyle{\min_{x \in \X}{Q_X(x)}}}{2} \, \chi^2(S_X||Q_X)
\end{equation}
for all $S_X,Q_X \in \Simplex_{\X}$. This follows from executing the proof of Lemma \ref{Lemma:Distribution Dependent KL Divergence Lower Bound} using the bound $\left\|S_X - Q_X\right\|_{\infty} \leq \left\|S_X - Q_X\right\|_{1}$ (which neglects the information that $(S_X - Q_X) \1 = 0$) instead of \eqref{Eq: Helper min}, and then applying \eqref{Eq: phi Bound} to the resulting lower bound on KL divergence. Alternatively, we provide a proof of \eqref{Eq:Looser Lower Bound on KL Divergence} via \textit{Bregman divergences} in Appendix \ref{App: Proof of Looser Lower Bound on KL Divergence} for completeness, cf. \cite[Lemma 6]{BoundsbetweenContractionCoefficients}. The improvement by a factor of $2$ from \eqref{Eq:Looser Lower Bound on KL Divergence} to \eqref{Eq:KL Divergence Lower Bound} is also observed in \cite[Remark 33]{fDivergenceBounds}, where the authors mention that our result \cite[Theorem 9]{BoundsbetweenContractionCoefficients} (see \eqref{Eq:Looser Contraction Coefficient Bound}) in the conference version of this paper can be improved by a factor of $2$ by using \eqref{Eq:KL Divergence Lower Bound} instead \eqref{Eq:Looser Lower Bound on KL Divergence}. We believe the authors of \cite{fDivergenceBounds} may have missed our result \cite[Theorem 10]{BoundsbetweenContractionCoefficients} (see Corollary \ref{Cor:Contraction Coefficient Bound}) in the conference paper, which presents precisely this improvement by a factor of $2$.

Lastly, we remark that \cite[Section III-D]{SDPIandSobolevInequalities} also presents upper bounds on $\etaKL\!\left(P_{X},P_{Y|X}\right)$ that use the function $\phi:\left[0,\frac{1}{2}\right] \rightarrow \R$, which stems from the refined Pinsker's inequality in \cite{TighterPinskerInequality}. However, these upper bounds are not in terms of $\etaChi\!\left(P_{X},P_{Y|X}\right)$. 

We next prove Theorem \ref{Thm:General Contraction Coefficient Bound} by combining Lemmata \ref{Lemma:f-Divergence Lower Bound} and \ref{Lemma:f-Divergence Upper Bound} from subsection \ref{Bounds on f-Divergences using Chi-Squared Divergence}.

\renewcommand{\proofname}{\bfseries \emph{Proof of Theorem \ref{Thm:General Contraction Coefficient Bound}}}

\begin{proof}
The conditions of Theorem \ref{Thm:General Contraction Coefficient Bound} encapsulate all the conditions of Lemmata \ref{Lemma:f-Divergence Lower Bound} and \ref{Lemma:f-Divergence Upper Bound}. Hence, using Lemmata \ref{Lemma:f-Divergence Lower Bound} and \ref{Lemma:f-Divergence Upper Bound}, for every pmf $R_X \in \Simplex_{\X}$ such that $R_X \neq P_X$, we have:
$$ \frac{D_{f}(R_X W||P_Y)}{D_{f}(R_X||P_X)} \leq \frac{\left(f^{\prime}(1) + f(0)\right) \chi^2(R_X W||P_Y)}{\displaystyle{f^{\prime \prime}(1) \min_{x \in \X}{P_X(x)} \, \chi^2(R_X||P_X)}} $$
where $P_Y = P_X W$. Taking the supremum over all $R_X \neq P_X$ on both sides produces: 
$$ \eta_{f}\!\left(P_{X},P_{Y|X}\right) \leq \frac{f^{\prime}(1) + f(0)}{\displaystyle{f^{\prime \prime}(1) \min_{x \in \X}{P_X(x)}}} \, \etaChi\!\left(P_{X},P_{Y|X}\right) $$
where we use Definition \ref{Def:Contraction Coefficient} and \eqref{Eq:Contraction Coefficient for Chi Squared Divergence}. This completes the proof.
\end{proof}

\renewcommand{\proofname}{\bfseries \emph{Proof}}

We remark that \cite[Theorem III.4]{SDPIandSobolevInequalities} presents an alternative linear upper bound on $\eta_{f}\!\left(P_{X},P_{Y|X}\right)$ using $\etaChi\!\left(P_{X},P_{Y|X}\right)$. Suppose $f:[0,\infty) \rightarrow \R$ is a twice differentiable convex function that has $f(1) = 0$, is strictly convex at unity, and has non-increasing second derivative. If we further assume that the difference quotient $x \mapsto \frac{f(x) - f(0)}{x}$ is concave, then the following bound holds \cite[Theorem III.4]{SDPIandSobolevInequalities}:
\begin{equation}
\label{Eq:Alternative Bound}
\eta_{f}\!\left(P_{X},P_{Y|X}\right) \leq \frac{2 \left(f^{\prime}(1) + f(0)\right)}{f^{\prime \prime}\!\left(1/p_{\star}\right)} \, \etaChi\!\left(P_{X},P_{Y|X}\right)
\end{equation}
where $p_{\star} = \min_{x \in \X}{P_X(x)}$. Hence, when $f$ is additionally thrice differentiable at unity, has $f^{\prime \prime}(1) > 0$, and satisfies \eqref{Eq:General Pinsker Condition} for every $t \in (0,\infty)$, we can improve the upper bound in Theorem \ref{Thm:General Contraction Coefficient Bound} to:
\begin{equation}
\label{Eq:Best f-Divergence Bound}
\eta_{f} \leq \min\!\left\{\frac{f^{\prime}(1) + f(0)}{f^{\prime \prime}(1) \, p_{\star}},\frac{2 \left(f^{\prime}(1) + f(0)\right)}{f^{\prime \prime}\!\left(1/p_{\star}\right)}\right\} \etaChi
\end{equation}
where $\eta_{f} = \eta_{f}\!\left(P_{X},P_{Y|X}\right)$ and $\etaChi = \etaChi\!\left(P_{X},P_{Y|X}\right)$. 

Observe that our bound in Theorem \ref{Thm:General Contraction Coefficient Bound} is tighter than that in \eqref{Eq:Alternative Bound} if and only if:
\begin{align}
\frac{2 \left(f^{\prime}(1) + f(0)\right)}{f^{\prime \prime}\!\left(1/p_{\star}\right)} & \geq \frac{f^{\prime}(1) + f(0)}{f^{\prime \prime}(1) \, p_{\star}} \\
\Leftrightarrow \quad 2 \, f^{\prime \prime}(1) \, p_{\star} & \geq f^{\prime \prime}\!\left(1/p_{\star}\right) .
\label{Eq:Tightness Condition}
\end{align}
One function that satisfies the conditions of Theorem \ref{Thm:General Contraction Coefficient Bound} and \eqref{Eq:Alternative Bound} as well as \eqref{Eq:Tightness Condition} is $f(t) = t \log(t)$. This engenders the improvement that Corollary \ref{Cor:Contraction Coefficient Bound} (which can be recovered from Theorem \ref{Thm:General Contraction Coefficient Bound}) offers over \eqref{Eq:Looser Contraction Coefficient Bound} (which can be recovered from \cite[Theorem III.4]{SDPIandSobolevInequalities}). 

As another example, consider the function $f(t) = \frac{t^{\alpha} - 1}{\alpha-1}$ for $\alpha \in (0,2]\backslash\!\{1\}$, which defines the Hellinger divergence of order $\alpha$ (see subsection \ref{f-Divergence}). It is straightforward to verify that this function satisfies the conditions of Theorem \ref{Thm:General Contraction Coefficient Bound} and \eqref{Eq:Alternative Bound}, cf. \cite[Corollary 6]{fDivergencePinsker}, \cite[Section III-B, p.3362]{SDPIandSobolevInequalities}. In this case, our bound in Theorem \ref{Thm:General Contraction Coefficient Bound} is tighter than \eqref{Eq:Alternative Bound} for all Hellinger divergences of order $\alpha$ satisfying \eqref{Eq:Tightness Condition}, i.e. $2 f^{\prime \prime}(1) \, p_{\star} = 2 \alpha p_{\star} \geq \alpha (1/p_{\star})^{\alpha-2} = f^{\prime \prime}(1/p_{\star}) \, \Leftrightarrow \, p_{\star}^{\alpha - 1} \geq \frac{1}{2}$, or equivalently, $0 < \alpha \leq 1 + (\log(2)/\!\log(1/p_{\star}))$ (where $\alpha = 1$ corresponds to KL divergence\textemdash{}see subsection \ref{f-Divergence}).

\subsection{Ergodicity of Markov Chains}
\label{Markov Convergence}

In this subsection, we derive a corollary of Corollary \ref{Cor:Contraction Coefficient Bound} that illustrates one use of upper bounds on contraction coefficients of source-channel pairs via $\etaChi\!\left(P_X,P_{Y|X}\right)$. Consider a Markov kernel $W \in \Simplex_{\X|\X}$ on a state space $\X$ that defines an \textit{irreducible} and \textit{aperiodic} (time homogeneous) discrete-time Markov chain with unique stationary pmf $P_X \in \Simplex_{\X}^{\circ}$ such that $P_X W = P_X$, cf. \cite[Section 1.3]{MarkovMixing}.\footnote{The matrix $W$ is \textit{primitive} since the chain is irreducible and aperiodic.} For simplicity, suppose further that $W$ is \textit{reversible} (i.e. the detailed balance equations, $P_X(x) [W]_{x,y} = P_X(y) [W]_{y,x}$ for all $x,y \in \X$, hold \cite[Section 1.6]{MarkovMixing}). This means that $W$ is self-adjoint with respect to the weighted inner product defined by $P_X$, and has all real eigenvalues $1 = \lambda_1(W) > \lambda_2(W) \geq \dots \geq \lambda_{|\X|}(W) > -1$. Let $\mu(W) \triangleq \max\{|\lambda_2(W)|,|\lambda_{|\X|}(W)|\} \in [0,1)$ denote the SLEM of $W$ (see subsection \ref{Coefficients of Ergodicity}). 

Since this Markov chain is \textit{ergodic}, $\lim_{n \rightarrow \infty}{R_X W^n} = P_X$ for all $R_X \in \Simplex_{\X}$ \cite[Theorem 4.9]{MarkovMixing}. This implies that $\lim_{n \rightarrow \infty}{D(R_X W^n||P_X)} = 0$ by the continuity of KL divergence \cite[Proposition 3.1]{InfoTheoryNotes}. Let us estimate the rate at which this ``distance to stationarity'' (measured by KL divergence) vanishes. A na\"{i}ve approach is to apply the SDPI \eqref{Eq:SDPI} for KL divergence recursively to get:
\begin{equation}
D(R_X W^n||P_X) \leq \etaKL\!\left(P_X,W\right)^n D(R_X||P_X)
\end{equation}
for all $R_X \in \Simplex_{\X}$. Using \eqref{Eq:Contraction Coefficient for KL Divergence}, this implies that:
\begin{equation}
\lim_{n \rightarrow \infty}{\etaKL\!\left(P_X,W^n\right)^{\frac{1}{n}}} \leq \etaKL\!\left(P_X,W\right)
\end{equation}
which turns out to be a loose bound on the rate in general. 

When $n$ is large, since $R_X W^n$ is close to $P_X$, we intuitively expect $D(R_X W^n||P_X)$ to resemble a $\chi^2$-divergence (see \eqref{Def:Local f-Divergence} in subsection \ref{f-Divergence}), which suggests that $\etaKL\!\left(P_X,W^n\right)$ should scale like $\etaChi\!\left(P_X,W\right)^n$. This intuition is rigorously executed in \cite[Section 6]{CIR93}. Indeed, when $\mu(W)$ is strictly greater than the third largest eigenvalue modulus of $W$, a consequence of \cite[Corollary 6.2]{CIR93} is:
\begin{equation}
\lim_{n \rightarrow \infty}{\frac{D(R_X W^{n}||P_X)}{D(R_X W^{n-1}||P_X)}} \leq \mu(W)^2
\end{equation}
for all $R_X \in \Simplex_{\X}$ such that the denominator is always positive. (This limit is either $0$ or $\mu(W)^2$.) Hence, after employing a Ces\`{a}ro convergence argument and telescoping, we get:
\begin{equation}
\lim_{n \rightarrow \infty}{\frac{D(R_X W^{n}||P_X)}{D(R_X||P_X)}} \leq \mu(W)^2 
\end{equation} 
which suggests that $\lim_{n \rightarrow \infty}{\etaKL\!\left(P_X,W^n\right)^{\frac{1}{n}}} \leq \mu(W)^2$. The next result proves that this inequality is in fact tight. 

\begin{corollary}[Rate of Convergence]
\label{Cor: Rate of Convergence}
For every irreducible, aperiodic, and reversible Markov chain with transition kernel $W \in \Simplex_{\X|\X}$ and stationary pmf $P_X \in \Simplex_{\X}^{\circ}$, we have:
$$ \lim_{n \rightarrow \infty}{\etaKL\!\left(P_X,W^n\right)^{\frac{1}{n}}} = \etaChi\!\left(P_X,W\right) = \mu(W)^2 \, . $$
\end{corollary} 

\begin{proof}
Since $W$ is reversible and $P_X$ is its stationary pmf, the DTM $B = \textsf{\small diag}\!\left(\sqrt{P_X}\right) W \textsf{\small diag}\!\left(\sqrt{P_X}\right)^{-1}$ is symmetric and similar to $W$ (see definition \eqref{Eq: DTM}). Hence, $W$ and $B$ share the same eigenvalues, and $\mu(W)$ is the second largest singular value of $B$. Using Proposition \ref{Prop:Singular Value Characterization of Maximal Correlation} and \eqref{Eq: Maximal Correlation as Contraction Coefficient}, we have $\etaChi\!\left(P_X,W\right) = \mu(W)^2$, which proves the second equality.  

Likewise, $\etaChi\!\left(P_X,W^n\right) = \mu(W^n)^2$ since $W^n$ is reversible for any $n \geq 1$. This yields:
\begin{equation}
\label{Eq:Reversible SLEM}
\etaChi\!\left(P_X,W^n\right) = \mu(W^n)^2 = \mu(W)^{2n} = \etaChi\!\left(P_X,W\right)^n 
\end{equation} 
where the second equality holds because eigenvalues of $W^n$ are $n$th powers of eigenvalues of $W$. Using \eqref{Eq:Reversible SLEM}, part 6 of Proposition \ref{Prop:Contraction Coefficient Properties}, and Corollary \ref{Cor:Contraction Coefficient Bound}, we get:
$$ \etaChi\!\left(P_X,W\right)^n \leq \etaKL\!\left(P_X,W^n\right) \leq \frac{\etaChi\!\left(P_X,W\right)^n}{\displaystyle{\min_{x \in \X}{P_X(x)}}} \, . $$
Taking $n$th roots and letting $n \rightarrow \infty$ yields the desired result.
\end{proof}

Corollary \ref{Cor: Rate of Convergence} portrays the well-understood phenomenon that $D(R_X W^n||P_X)$ vanishes with rate determined by $\mu(W)^2 = \etaChi\!\left(P_X,W\right)$. 
More generally, it illustrates that the bounds in Corollary \ref{Cor:Contraction Coefficient Bound} and Theorems \ref{Thm:General Contraction Coefficient Bound} and \ref{Thm:Refined Contraction Coefficient Bound} are useful in the regime where the random variables $X$ and $Y$ are weakly dependent (e.g. $X$ is the initial state of an ergodic reversible Markov chain, and $Y$ is the state after a large number of time steps). In this regime, these bounds are quite tight, and beat the $\etaTV$ bound in part 6 of Proposition \ref{Prop:Contraction Coefficient Properties 2}.

\subsection{Tensorization of Bounds between Contraction Coefficients}
\label{Tensorization of Bounds between Contraction Coefficients}

In the absence of weak dependence, the upper bounds in Corollary \ref{Cor:Contraction Coefficient Bound} and Theorems \ref{Thm:General Contraction Coefficient Bound} and \ref{Thm:Refined Contraction Coefficient Bound} can be loose. In fact, they can be rendered arbitrarily loose because the constants in these bounds do not tensorize, while contraction coefficients do (as shown in part 4 of Proposition \ref{Prop:Contraction Coefficient Properties}). For instance, if we are given $P_{X,Y}$ with $X \sim \textsf{\small Bernoulli}\!\left(\frac{1}{2}\right)$, then the constant in the upper bound of Corollary \ref{Cor:Contraction Coefficient Bound} is $1/\min_{x \in \{0,1\}}{P_X(x)} = 2$. If we instead consider a sequence of pairs $\left(X_1,Y_1\right),\dots,\left(X_n,Y_n\right)$ that are independent and identically distributed (i.i.d.) according to $P_{X,Y}$, then the constant in the upper bound of Corollary \ref{Cor:Contraction Coefficient Bound} is $1/\min_{x_1^n \in \{0,1\}^n}{P_{X_1^n}(x_1^n)} = 2^{n}$. However, since $\etaKL\!\left(P_{X_1^n},P_{Y_1^n|X_1^n}\right) = \etaKL\!\left(P_{X},P_{Y|X}\right)$ and $\etaChi\!\left(P_{X_1^n},P_{Y_1^n|X_1^n}\right) = \etaChi\!\left(P_{X},P_{Y|X}\right)$ by the tensorization property in part 4 of Proposition \ref{Prop:Contraction Coefficient Properties}, the constant $2^n$ becomes arbitrarily loose as $n$ grows. The next corollary presents a partial remedy for this i.i.d. slackening attack for Corollary \ref{Cor:Contraction Coefficient Bound}.

\begin{corollary}[Tensorized KL Contraction Coefficient Bound] 
\label{Cor:Tensorized Contraction Coefficient Bound}
If $(X_1,Y_1),\dots,(X_n,Y_n)$ are i.i.d. with joint pmf $P_{X,Y}$ such that $P_X \in \Simplex_{\X}^{\circ}$ and $P_Y \in \Simplex_{\Y}^{\circ}$, then:
$$ \etaKL\!\left(P_{X_1^n},P_{Y_1^n|X_1^n}\right) \leq \frac{\etaChi\!\left(P_{X_1^n},P_{Y_1^n|X_1^n}\right)}{\displaystyle{\min_{x \in \X}{P_X(x)}}} \, . $$
\end{corollary}

\begin{proof}
This follows trivially from Corollary \ref{Cor:Contraction Coefficient Bound} and the tensorization property in part 4 of Proposition \ref{Prop:Contraction Coefficient Properties}.
\end{proof}

In the product distribution context, this corollary permits us to use the tighter factor $1/\!\min_{x \in \X}{P_X(x)}$ in the upper bound of Corollary \ref{Cor:Contraction Coefficient Bound} instead of $1/\!\min_{x_1^n \in \X^n}{P_{X_1^n}\left(x_1^n\right)} = \left(1/\!\min_{x \in \X}{P_X(x)}\right)^n$. Similar adjustments can be made for the constants in Theorems \ref{Thm:General Contraction Coefficient Bound} and \ref{Thm:Refined Contraction Coefficient Bound} in this context as well. Thus, tensorization can improve the upper bounds in Corollary \ref{Cor:Contraction Coefficient Bound} and Theorems \ref{Thm:General Contraction Coefficient Bound} and \ref{Thm:Refined Contraction Coefficient Bound}.

\section{Proof of Equivalence between Gaussian Contraction Coefficients}
\label{Proof of Gaussian Contraction Coefficients}

We prove Theorem \ref{Thm: Gaussian Contraction Coefficients} in this section. Recall from subsection \ref{Contraction Coefficients of Gaussian Random Variables} that we are given the jointly Gaussian pdf $P_{X,Y}$ defined via \eqref{Eq: AWGN}, with source pdf $P_X = \Gauss(0,\sigma_X^2)$ and channel conditional pdfs $\{P_{Y|X = x} = \Gauss(x,\sigma_W^2) : x \in \R\}$ such that $\sigma_X^2,\sigma_W^2 > 0$. Let $\mathcal{T}$ be the set of all pdfs with bounded support. Thus, a pdf $R_X \in \mathcal{T}$ if and only if there exists $C \in \R_{+}$ such that $R_X = R_X \I_{[-C,C]}$ almost everywhere with respect to the Lebesgue measure.\footnote{Note that $\I_{[-C,C]}:\R \rightarrow \{0,1\}$ denotes the indicator function on $[-C,C]$, i.e. $\I_{[C,C]}(x) = 1$ if $x \in [C,C]$, and $\I_{[C,C]}(x) = 0$ otherwise.} We first derive the following useful lemma.

\begin{lemma}[Bounded Support Characterization of $\etaKL$]
\label{Lemma: Bounded Support Characterization}
The supremum in \eqref{Eq: Gaussian KL Contraction} can be restricted to pdfs in $\mathcal{T}$:
$$ \etaKL\!\left(P_{X},P_{Y|X}\right) = \sup_{\substack{R_X \in \mathcal{T}:\\D(R_X||P_X) < +\infty}}{\frac{D(R_Y||P_Y)}{D(R_X||P_X)}}$$
where $R_Y = R_X \star P_W$ for each $R_X$, and $P_W = \Gauss(0,\sigma_W^2)$.
\end{lemma} 

\begin{proof}
Consider any pdf $R_X$ such that $0 < D(R_X||P_X) < +\infty$, and define a corresponding the sequence of pdfs $R_X^{(n)} = R_X \I_{[-n,n]}/C_n \in \mathcal{T}$, where $C_n = \E_{R_X}[\I_{[-n,n]}(X)]$, the indices $n \in \N$ are sufficiently large so that $C_n > 0$, and $\lim_{n \rightarrow \infty}{C_n} = 1$. Observe that:
\begin{align*}
D(R_X^{(n)}||P_X) & = \frac{1}{C_n}\E_{R_X}\!\left[\I_{[-n,n]}(X) \log\!\left(\frac{R_X(X)}{P_X(X)}\right)\right] \\ 
& \quad \, - \log(C_n) \, .
\end{align*}
Clearly, $\lim_{n \rightarrow\infty}{\I_{[-n,n]} \log(R_X/P_X)} = \log(R_X/P_X)$ pointwise $R_X-a.s.$ and $|\I_{[-n,n]} \log(R_X/P_X)| \leq |\log(R_X/P_X)|$ pointwise $R_X-a.s.$ such that $\E_{R_X}\!\left[|\log(R_X(X)/P_X(X))|\right] < +\infty$ (where the finiteness follows from $D(R_X||P_X) < +\infty$). Hence, the dominated convergence theorem (DCT) yields:
\begin{equation}
\label{Eq: Denominator convergence}
\lim_{n \rightarrow \infty}{D(R_X^{(n)}||P_X)} = D(R_X||P_X) \, . 
\end{equation}

Furthermore, let $R_Y^{(n)} = R_X^{(n)} \star P_W$ so that for every $y \in \R$:
$$ R_Y(y) - C_n R_Y^{(n)}(y) = \E_{R_X}\!\left[\I_{\R\backslash[-n,n]}(X) P_W(y - X)\right] . $$
Since for all $x,y \in \R$, $\lim_{n \rightarrow\infty}{\I_{\R\backslash[-n,n]}(x) P_W(y-x)} = 0$ and $0 \leq \I_{\R\backslash[-n,n]}(x) P_W(y-x) \leq P_W(y-x)$ such that $\E_{R_X}\!\left[P_W(y-X)\right] = R_Y(y) < +\infty$, applying the DCT shows the pointwise convergence of the pdfs $\{R_Y^{(n)}\}$:
$$ \forall y \in \R, \enspace \lim_{n \rightarrow \infty}{C_n R_Y^{(n)}(y)} = \lim_{n \rightarrow \infty}{R_Y^{(n)}(y)} = R_Y(y) \, . $$
This implies that $R_Y^{(n)}$ converges weakly to $R_Y$ as $n \rightarrow \infty$ by Scheff\'{e}'s lemma. Hence, by the weak lower semi-continuity of KL divergence \cite[Theorem 3.6, Section 3.5]{InfoTheoryNotes}:
\begin{equation}
\label{Eq: Numerator convergence}
\liminf_{n \rightarrow \infty}{D(R_Y^{(n)}||P_Y)} \geq D(R_Y||P_Y) \, . 
\end{equation} 
Combining \eqref{Eq: Denominator convergence} and \eqref{Eq: Numerator convergence}, we get: 
\begin{equation}
\label{Eq: Bounded support approximation}
\liminf_{n \rightarrow \infty}{\frac{D(R_Y^{(n)}||P_Y)}{D(R_X^{(n)}||P_X)}} \geq \frac{D(R_Y||P_Y)}{D(R_X||P_X)} \, . 
\end{equation}

To complete the proof, we use a ``diagonalization argument.'' Suppose $\{R_{X,m} : m \in \N\}$ is a sequence of pdfs that satisfies $0 < D(R_{X,m}||P_X) < +\infty$ for all $m \in \N$ and achieves the supremum in \eqref{Eq: Gaussian KL Contraction}:
$$ \lim_{m \rightarrow \infty}{\frac{D(R_{Y,m}||P_Y)}{D(R_{X,m}||P_X)}} = \etaKL\!\left(P_X,P_{Y|X}\right) $$
where $R_{Y,m} = R_{X,m} \star P_W$. Then, since \eqref{Eq: Bounded support approximation} is true, we can construct a sequence $\{R_{X,m}^{(n(m))} \in \mathcal{T}: m \in \N\}$, where each $n(m)$ is chosen such that for every $m \in \N$:
$$ \frac{D(R_{Y,m}^{(n(m))}||P_Y)}{D(R_{X,m}^{(n(m))}||P_X)} \geq \frac{D(R_{Y,m}||P_Y)}{D(R_{X,m}||P_X)} - \frac{1}{2^m} $$ 
where $R_{Y,m}^{(n(m))} = R_{X,m}^{(n(m))} \star P_W$. Letting $m \rightarrow \infty$, we have:
$$ \liminf_{m \rightarrow \infty}{\frac{D(R_{Y,m}^{(n(m))}||P_Y)}{D(R_{X,m}^{(n(m))}||P_X)}} \geq \etaKL\!\left(P_X,P_{Y|X}\right) \, . $$
Since the supremum in \eqref{Eq: Gaussian KL Contraction} is over all pdfs (which certainly includes all pdfs in $\mathcal{T}$), this inequality is actually an equality. This completes the proof. (Also note that for any $R_X \in \mathcal{T}$, the constraint $D(R_X||P_X) > 0$ is automatically true since $P_X = \Gauss(0,\sigma_X^2)$. So, the supremum in the lemma statement does not include this constraint.)
\end{proof}

We next prove Theorem \ref{Thm: Gaussian Contraction Coefficients} using Lemma \ref{Lemma: Bounded Support Characterization}, which ensures that all differential entropy terms in the ensuing argument are well-defined and finite. 

\renewcommand{\proofname}{\bfseries \emph{Proof of Theorem \ref{Thm: Gaussian Contraction Coefficients}}}

\begin{proof} 
First note that:
\begin{align*}
\etaChi\!\left(P_{X},P_{Y|X}\right) & = \rho^2(X;Y) \\
& = \frac{\mathbb{C}\mathbb{O}\mathbb{V}(X,Y)^2}{\VAR(X)\VAR(Y)} = \frac{\sigma_X^2}{\sigma_X^2 + \sigma_W^2} 
\end{align*}
where the first equality is precisely \eqref{Eq: Maximal Correlation as Contraction Coefficient} (which holds for general random variables \cite{ChiSquaredContractionMaximalCorrelation}), the second equality follows from to R\'{e}nyi's seventh axiom that $\rho(X;Y)$ is the absolute value of the Pearson correlation coefficient of $X$ and $Y$ \cite{RenyiCorrelation}, and the final equality follows from direct computation. 

We next prove that for any $p \geq \sigma_X^2$:
$$ \etaKL\!\left(P_{X},P_{Y|X}\right) \geq \etaKLp\!\!\left(P_{X},P_{Y|X}\right) \geq \etaChi\!\left(P_{X},P_{Y|X}\right) . $$
The first inequality is obvious from \eqref{Eq: Gaussian KL Contraction} and \eqref{Eq: Contraction Coefficient with Power Constraint}. For the second inequality, let $R_X = \Gauss(\sqrt{\delta},\sigma_X^2 - \delta)$ and $R_Y = R_X \star P_W = \Gauss(\sqrt{\delta},\sigma_X^2 + \sigma_W^2 - \delta)$ for any $\delta > 0$. Then, we get:
$$ \lim_{\delta \rightarrow 0^+}{\frac{D(R_Y||P_Y)}{D(R_X||P_X)}} = \lim_{\delta \rightarrow 0^+}{\frac{\log\!\left(\frac{\sigma_X^2 + \sigma_W^2}{\sigma_X^2 + \sigma_W^2 - \delta}\right)}{\log\!\left(\frac{\sigma_X^2}{\sigma_X^2 - \delta}\right)}} = \frac{\sigma_X^2}{\sigma_X^2 + \sigma_W^2} $$
where the second equality follows from l'H\^{o}pital's rule. Since $\E_{R_X}\!\left[X^2\right] = \sigma_X^2$ for every $\delta > 0$, we have $\etaKLp\!\!\left(P_{X},P_{Y|X}\right) \geq \sigma_X^2/(\sigma_X^2 + \sigma_W^2)$ for any $p \geq \sigma_X^2$.

Therefore, it suffices to prove that $\etaKL\!\left(P_{X},P_{Y|X}\right) \leq \sigma_X^2/$ $(\sigma_X^2 + \sigma_W^2)$. Using Lemma \ref{Lemma: Bounded Support Characterization}, we can equivalently show that:
\begin{equation}
\label{Eq: Sufficient Condition 1}
\frac{D(R_Y||P_Y)}{D(R_X||P_X)} \leq \frac{\sigma_X^2}{\sigma_X^2 + \sigma_W^2} 
\end{equation}
for every pdf $R_X \in \mathcal{T}$ with $D(R_X||P_X) < +\infty$. 

For any pdf $R_X$, we define the \textit{differential entropy} of $R_X$ as $h(R_X) \triangleq -\E_{R_X}\!\left[\log\!\left(R_X(X)\right)\right]$. To check that such differential entropy terms are well-defined and finite for $R_X \in \mathcal{T}$, we employ the argument in \cite[Lemma 8.3.1, Theorem 8.3.3]{InfoTheoryAsh}. Observe that for all $x \in \esssupp(R_X)$:\footnote{For any Borel measurable function $f:\R \rightarrow \R$, $\esssupp(f)$ denotes the \textit{essential support} of $f$ (with respect to the Lebesgue measure).}
$$ \log(R_X(x)) = \log\!\left(\frac{R_X(x)}{P_X(x)}\right) - \frac{1}{2}\log\!\left(2\pi\sigma_X^2\right) - \frac{x^2}{2 \sigma_X^2} \, . $$
Since $D(R_X||P_X)$ must be finite in \eqref{Eq: Gaussian KL Contraction} and $X^2 \geq 0$, we can take expectations with respect to $R_X$ to get:
\begin{equation}
\label{Eq:KL Divergence with Gaussian}
-h(R_X) = D(R_X||P_X) - \frac{1}{2}\log\!\left(2\pi \sigma_X^2\right) - \frac{\E_{R_X}\!\left[X^2\right]}{2\sigma_X^2} 
\end{equation}
which shows that $h(R_X)$ always exists, $h(R_X)$ is finite when $\E_{R_X}\!\left[X^2\right] < +\infty$, and $h(R_X) = +\infty$ when $\E_{R_X}\!\left[X^2\right] = +\infty$. Furthermore, if the pdf $R_X \in \mathcal{T}$ has bounded support, $\E_{R_X}\!\left[X^2\right] < +\infty$ and $h(R_X)$ is well-defined and finite.

Let $R_X \in \mathcal{T}$ and $R_Y = R_X \star P_W$ have second moments $\E_{R_X}\!\left[X^2\right] = \sigma_X^2 + q > 0$ and $\E_{R_Y}\!\left[Y^2\right] = \sigma_X^2 + \sigma_W^2 + q > 0$ for some $q > -\sigma_X^2$. Using \eqref{Eq:KL Divergence with Gaussian}, we have:
\begin{align*}
D(R_X||P_X) & = \frac{1}{2}\log\!\left(2\pi \sigma_X^2\right) + \frac{\sigma_X^2 + q}{2\sigma_X^2} - h(R_X) \\
& = h(P_X) - h(R_X) + \frac{q}{2\sigma_X^2} , \\
D(R_Y||P_Y) & = h(P_Y) - h(R_Y) + \frac{q}{2\left(\sigma_X^2 + \sigma_W^2\right)} 
\end{align*}
where $h(R_Y)$ exists and is finite because $\E_{R_Y}\!\left[Y^2\right]$ is finite (as argued earlier using \eqref{Eq:KL Divergence with Gaussian}). Hence, it suffices to prove that:
\begin{equation}
\label{Eq: Sufficient Condition 2}
h(P_Y) - h(R_Y) \leq \frac{\sigma_X^2}{\sigma_X^2 + \sigma_W^2} \left(h(P_X) - h(R_X)\right) 
\end{equation}
which is equivalent to \eqref{Eq: Sufficient Condition 1}. We can recast \eqref{Eq: Sufficient Condition 2} as:
\begin{align*}
\left(e^{2h(P_Y) - 2h(R_Y)}\right)^{\sigma_X^2 + \sigma_W^2} & \leq \left(e^{2h(P_X) - 2h(R_X)}\right)^{\sigma_X^2} \\
\left(\frac{\frac{1}{2\pi e}e^{2h(P_Y)}}{\frac{1}{2\pi e}e^{2h(R_Y)}}\!\right)^{\sigma_X^2 + \sigma_W^2} & \leq \left(\frac{\frac{1}{2\pi e}e^{2h(P_X)}}{\frac{1}{2\pi e}e^{2h(R_X)}}\!\right)^{\sigma_X^2} \\
\left(\frac{N(P_Y)}{N(R_Y)}\right)^{\sigma_X^2 + \sigma_W^2} & \leq \left(\frac{N(P_X)}{N(R_X)}\right)^{\sigma_X^2} 
\end{align*}
where for any pdf $Q_X$ such that $h(Q_X)$ exists, we define the \textit{entropy power} of $Q_X$ as $N(Q_X) \triangleq e^{2h(Q_X)}\!/(2\pi e)$ \cite[Section III-A]{DemboCoverThomas1991}. For $P_X = \Gauss(0,\sigma_X^2)$, $P_W = \Gauss(0,\sigma_W^2)$, and $P_Y = P_X \star P_W = \Gauss(0,\sigma_X^2 + \sigma_W^2)$, the entropy powers are $N(P_X) = \sigma_X^2$, $N(P_W) = \sigma_W^2$, and $N(P_Y) = \sigma_X^2 + \sigma_W^2$, respectively. Applying the \textit{entropy power inequality} to $R_X$, $P_W$, and $R_Y = R_X \star P_W$ \cite[Theorem 4]{DemboCoverThomas1991}, we have: 
$$ N(R_Y) \geq N(R_X) + N(P_W) = N(R_X) + \sigma_W^2 \, . $$
Hence, it is sufficient to prove that:
$$ \left(\frac{\sigma_X^2 + \sigma_W^2}{N(R_X) + \sigma_W^2}\right)^{\sigma_X^2 + \sigma_W^2} \leq \left(\frac{\sigma_X^2}{N(R_X)}\right)^{\sigma_X^2} . $$
Let $a = \sigma_X^2 + \sigma_W^2$, $b = \sigma_X^2 - N(R_X)$, and $c = \sigma_X^2$. Then, we have $a > c > 0$ and $c > b$ (which follows from the finiteness of $h(R_X)$), and it is sufficient to prove that:
$$ \left(\frac{a}{a-b}\right)^a \leq \left(\frac{c}{c-b}\right)^c $$
which is equivalent to proving:
$$ a > c > 0 \text{  and  } c > b \enspace \Rightarrow \enspace \left(1-\frac{b}{c}\right)^c \leq \left(1-\frac{b}{a}\right)^a . $$
This statement is a variant of Bernoulli's inequality proved in \cite[Theorem 3.1, parts $\left(r^{\prime}_7\right)$ and $\left(r^{\prime\prime}_7\right)$]{BernoulliInequality}. This completes the proof.
\end{proof}

\renewcommand{\proofname}{\bfseries \emph{Proof}}

\section{Proof of Equivalent Characterizations of the Less Noisy Preorder}
\label{Proof of Equivalent Characterizations of the Less Noisy Preorder}

We finally turn to deriving the equivalent characterizations of $\succeq_{\textsf{\tiny ln}}$ using operator convexity. The next subsection presents some preliminaries on operator convex functions, and subsection \ref{Proof of Theorem: Equivalent Characterizations of Less Noisy Preorder} proves Theorem \ref{Thm: Equivalent Characterizations of Less Noisy Preorder}.

\subsection{Operator Convex Functions} 
\label{Operator Convex Functions}

For any interval $I \subseteq \R$,\footnote{The interval $I$ can be finite or infinite, and closed or open.} let $\C^{n \times n}_{\textsf{\tiny Herm}}(I)$ denote the set of all $n \times n$ (complex) Hermitian matrices with all eigenvalues in $I$, where $\C^{n \times n}_{\textsf{\tiny Herm}} = \C^{n \times n}_{\textsf{\tiny Herm}}(\R)$ is the set of all Hermitian matrices. Given a function $f:I \rightarrow \R$, we can extend it to a function $f : \C^{n \times n}_{\textsf{\tiny Herm}}(I) \rightarrow \C^{n \times n}_{\textsf{\tiny Herm}}$ as follows \cite[Chapter V.1]{Bhatia1997}:
\begin{equation}
\forall A \in \C^{n \times n}_{\textsf{\tiny Herm}}(I), \enspace f(A) \triangleq U \, \textsf{\small diag}(f(\lambda_1),\dots,f(\lambda_n)) \, U^H
\end{equation} 
where $A = U \, \textsf{\small diag}(\lambda_1,\dots,\lambda_n) \, U^H$ is the spectral decomposition of $A$ with real eigenvalues $\lambda_1,\dots,\lambda_n \in I$, $U \in \C^{n \times n}$ is a unitary matrix, and $U^H$ is its Hermitian transpose. We say that $f$ is \textit{operator convex} if for every $n \geq 1$, every pair of matrices $A,B \in \C^{n \times n}_{\textsf{\tiny Herm}}(I)$, and every $\lambda \in [0,1]$:\footnote{It is straightforward to verify that $A,B \in \C^{n \times n}_{\textsf{\tiny Herm}}(I)$ implies that $\lambda A + (1-\lambda) B \in \C^{n \times n}_{\textsf{\tiny Herm}}(I)$ for all $\lambda \in [0,1]$.}
\begin{equation}
\lambda f(A) + (1-\lambda) f(B) \succeq_{\textsf{\tiny PSD}} f(\lambda A + (1-\lambda) B)
\end{equation}
where $\succeq_{\textsf{\tiny PSD}}$ denotes the \textit{L\"{o}wner partial order} (i.e. for any two matrices $A,B \in \C^{n \times n}_{\textsf{\tiny Herm}}$, $A \succeq_{\textsf{\tiny PSD}} B$ if and only if $A - B$ is positive semidefinite) \cite[Chapter V.1]{Bhatia1997}. Note that an operator convex function $f:I \rightarrow \R$ is clearly convex, and its translated affine transformations $g:c + I \rightarrow \R$, $g(t) = a f(t - c) + b$ are also operator convex for every $a \geq 0$, $b \in \R$, and $c \in \R$.\footnote{Note that $c + I = \{c + x : x \in I\}$ is the interval $I$ translated by $c$.}

A more striking property of operator convex functions is that they are characterized by certain \textit{integral representations}\textemdash{}see \textit{L\"{o}wner's theorems} in \cite[Chapter V.4, Problem V.5.5]{Bhatia1997}. In particular, for every operator convex function $f:(0,\infty) \rightarrow \R$ with $f(1) = 0$, there exist constants $a \in \R$ and $b \geq 0$, and a finite positive measure $\mu$ on $(1,\infty)$ (with its Borel $\sigma$-algebra) such that:
\begin{equation}
\label{Eq: Integral Representation}
f(t) = a(t-1) + b(t-1)^2 + \int_{(1,\infty)}{\!\frac{(t-1)(\omega t - \omega - 1)}{t + \omega - 1} \, d\mu(\omega)}
\end{equation}
which follows from \cite[Equation (7)]{ChoiRuskaiSeneta1994} and the associated references (note that our $f$ is related to $g$ in \cite{ChoiRuskaiSeneta1994} by $f(t) = g(t-1)$). As noted in \cite{ChoiRuskaiSeneta1994}, the converse also holds, i.e. functions of the form \eqref{Eq: Integral Representation} are operator convex. The ensuing lemma is a direct consequence of \eqref{Eq: Integral Representation}, which we distill from \cite{ChoiRuskaiSeneta1994} and present in a more transparent form.

\begin{lemma}[Integral Representation {\cite[p.33]{ChoiRuskaiSeneta1994}}]
\label{Lemma: Integral Representation}
Consider any $f$-divergence such that $f:(0,\infty) \rightarrow \R$ is operator convex and satisfies $f(1) = 0$. Then, there exists a constant $b \geq 0$ and a positive measure $\tau$ on $(0,1)$ (with its Borel $\sigma$-algebra) such that for every $R_X,P_X \in \Simplex_{\X}$:
\begin{align*}
D_{f}(R_X||P_X) & = b \, \chi^2(R_X||P_X) \\
& \quad \, + \int_{(0,1)}{\frac{1+\lambda^2}{\lambda(1-\lambda)} \textsf{\small LC}_{\lambda}(R_X||P_X) \, d\tau(\lambda)} \, .
\end{align*}
\end{lemma}
  
\begin{proof}
Fix any two pmfs $R_X,P_X \in \Simplex_{\X}$, and suppose the random variable $X$ has pmf $P_X$. Then, since $f$ exhibits the integral representation \eqref{Eq: Integral Representation}, we substitute $t = R_X(X)/P_X(X)$ into \eqref{Eq: Integral Representation} and take expectations to get:
\begin{align*}
\E\!\left[f\!\left(\frac{R_X(X)}{P_X(X)}\right)\right] & = b \, \E\!\left[\left(\frac{R_X(X)}{P_X(X)}-1\right)^2\right] \\
+ \bigintss_{\!(1,\infty)} \!\!\! \E & \!\left[\frac{\left(\frac{R_X(X)}{P_X(X)}-1\right)\!\!\left(\omega \frac{R_X(X)}{P_X(X)} - \omega - 1\right)}{\frac{R_X(X)}{P_X(X)} + \omega - 1}\right] \! d\mu(\omega) 
\end{align*}
where the first term on the right hand side of \eqref{Eq: Integral Representation} vanishes after taking expectations (see the affine invariance property in subsection \ref{f-Divergence}). This implies that:
\begin{align*}
D_{f}(R_X||P_X) & = b \, \chi^2(R_X||P_X) \\
& \quad \, + \bigintss_{\!(1,\infty)}{\!\!\!\!\E\!\left[\frac{\left(1+\omega^2\right) \! \left(\frac{R_X(X)}{P_X(X)}
-1\right)^2}{\omega \left(\frac{R_X(X)}{P_X(X)} + \omega - 1\right)}\right] \! d\mu(\omega)}
\end{align*}
where the left hand side follows from Definition \ref{Def: f-Divergence}, the $\chi^2$-divergence term follows from the definition in subsection \ref{f-Divergence}, and the last term follows from the affine invariance property in subsection \ref{f-Divergence} and the relation:
$$ \frac{(t-1)(\omega t - \omega - 1)}{t + \omega - 1} = \frac{(1+\omega^2) (t-1)^2}{\omega (t + \omega - 1)} - \frac{t-1}{\omega} $$
for any $t > 0$ and $\omega > 0$. Next, observe that the change of variables $\omega = \frac{1}{\lambda}$ yields:
\begin{align*}
D_{f}(R_X||P_X) & = b \, \chi^2(R_X||P_X) \\
& \quad \, + \bigintss_{\!(0,1)}{\!\!\!\E\!\left[\frac{\left(1 + \lambda^2\right) \! \left(\frac{R_X(X)}{P_X(X)}
-1\right)^2}{\left(\lambda \frac{R_X(X)}{P_X(X)} + 1 - \lambda\right)}\right] \! d\tau(\lambda)}
\end{align*}
for some positive measure $\tau$ on $(0,1)$. Finally, recognizing that the integrand on the right hand side is a scaled Vincze-Le Cam divergence (see subsection \ref{f-Divergence}), some straightforward algebra produces the desired integral representation.
\end{proof}

Lemma \ref{Lemma: Integral Representation} is used in \cite[p.33]{ChoiRuskaiSeneta1994} (in a slightly different form) to prove Proposition \ref{Prop: Operator Convex f-Divergence Contraction}, cf. \cite[Theorem 1]{ChoiRuskaiSeneta1994}. Furthermore, \cite[p.3363]{SDPIandSobolevInequalities} also distills the key idea in \cite[p.33]{ChoiRuskaiSeneta1994} and presents an alternative integral representation (in terms of Vincze-Le Cam and $\chi^2$-divergences) analogous to Lemma \ref{Lemma: Integral Representation}. However, the representation in \cite[p.3363]{SDPIandSobolevInequalities} only holds for operator convex functions $f$ where $f(0)$ is finite, while Lemma \ref{Lemma: Integral Representation} holds for infinite $f(0)$ as well.

\subsection{Proof of Theorem \ref{Thm: Equivalent Characterizations of Less Noisy Preorder}}
\label{Proof of Theorem: Equivalent Characterizations of Less Noisy Preorder}

We first recall the result in \cite[Theorem 1]{SymmetricChannelDomination}.

\begin{lemma}[$\chi^2$-Divergence Characterization of $\succeq_{\textsf{\tiny ln}}$ {\cite[Theorem 1]{SymmetricChannelDomination}}]
\label{Lemma: Chi-Squared Divergence Characterization of Less Noisy}
For any two channels $P_{Y|X}$ and $P_{Z|X}$ on the same input alphabet $\X$ with row stochastic transition probability matrices $W \in \Simplex_{\Y|\X}$ and $V \in \Simplex_{\Z|\X}$ respectively, $P_{Y|X} \succeq_{\textsf{\tiny ln}} P_{Z|X}$ if and only if:
$$ \chi^2(R_X W||P_X W) \geq \chi^2(R_X V||P_X V) $$
for every pair of input pmfs $R_X,P_X \in \Simplex_{\X}$. 
\end{lemma}

This lemma is proved in \cite[Section IV-A]{SymmetricChannelDomination}. We next use it along with Lemma \ref{Lemma: Integral Representation} to prove Theorem \ref{Thm: Equivalent Characterizations of Less Noisy Preorder}. 

\renewcommand{\proofname}{\bfseries \emph{Proof of Theorem \ref{Thm: Equivalent Characterizations of Less Noisy Preorder}}}

\begin{proof}
Fix any non-linear operator convex function $f:(0,\infty) \rightarrow \R$ such that $f(1) = 0$, where the non-linearity ensures that the corresponding $f$-divergence is not identically zero (see the affine invariance property in subsection \ref{f-Divergence}). We follow the proof outline in \cite[Section IV-A]{SymmetricChannelDomination} (also see \cite{ChoiRuskaiSeneta1994} and \cite[Section III-C]{SDPIandSobolevInequalities}). 

To prove the forward direction, we first use Lemma \ref{Lemma: Integral Representation} and the equivalent form of Vincze-Le Cam divergences in \eqref{Eq: Chi-Squared Characterization of Le Cam Divergence} to obtain the following integral representation of our $f$-divergence in terms of $\chi^2$-divergence (cf. \cite[p.33]{ChoiRuskaiSeneta1994}):
\begin{equation}
\label{Eq: Chi-Squared Integral Representation}
\begin{aligned}
D_{f}(R_X||P_X) & = b \, \chi^2(R_X||P_X) \\
& + \int_{(0,1)}{\frac{1+\lambda^2}{(1-\lambda)^{2}} \chi^2(R_X||\lambda R_X + \bar{\lambda} P_X) \, d\tau(\lambda)}
\end{aligned}
\end{equation}
for all $R_X,P_X \in \Simplex_{\X}$, where $\bar{\lambda} = 1 - \lambda$. Since $P_{Y|X} \succeq_{\textsf{\tiny ln}} P_{Z|X}$, Lemma \ref{Lemma: Chi-Squared Divergence Characterization of Less Noisy} yields $\chi^2(R_X W||P_X W) \geq \chi^2(R_X V||P_X V)$ and $\chi^2(R_X W||(\lambda R_X + \bar{\lambda} P_X) W) \geq \chi^2(R_X V||(\lambda R_X + \bar{\lambda} P_X) V)$ for every $R_X,P_X \in \Simplex_{\X}$ and every $\lambda \in (0,1)$. Using these inequalities and the integral representation in \eqref{Eq: Chi-Squared Integral Representation}, we get:
$$ D_{f}(R_X W || P_X W) \geq D_{f}(R_X V || P_X V) $$
for every $R_X,P_X \in \Simplex_{\X}$, as desired.

To prove the converse direction, observe that the integral representation in \eqref{Eq: Integral Representation} ensures that $f$ is infinitely differentiable and $f^{\prime \prime}(1) > 0$. For any $R_X,P_X \in \Simplex_{\X}$ and any $\lambda \in (0,1)$, we are given that:
$$ D_{f}((\bar{\lambda} P_X + \lambda R_X) W || P_X W) \geq D_{f}((\bar{\lambda} P_X + \lambda R_X) V || P_X V) . $$
So, when $P_X \in \Simplex_{\X}^{\circ}$, we can scale both sides by $\frac{2}{f^{\prime \prime}(1) \lambda^2} > 0$ and let $\lambda \rightarrow 0$ to obtain:
$$ \chi^2(R_X W || P_X W) \geq \chi^2(R_X V || P_X V) $$
using the local approximation of $f$-divergences in \eqref{Def:Local f-Divergence}. Although our version of \eqref{Def:Local f-Divergence} requires the $P_X \in \Simplex_{\X}^{\circ}$ assumption, the above inequality also holds for $P_X \in \Simplex_{\X}\backslash\Simplex_{\X}^{\circ}$ due to the continuity of $\chi^2$-divergence in its second argument with fixed first argument. Hence, Lemma \ref{Lemma: Chi-Squared Divergence Characterization of Less Noisy} yields that $P_{Y|X} \succeq_{\textsf{\tiny ln}} P_{Z|X}$. This completes the proof.
\end{proof}

We remark that even without using Lemma \ref{Lemma: Chi-Squared Divergence Characterization of Less Noisy}, this proof illustrates that all channel preorders that are defined using non-linear operator convex $f$-divergences (in a manner analogous to $\succeq_{\textsf{\tiny ln}}$) are equivalent. Indeed, the proof shows that they are all characterized by $\chi^2$-divergence. Since $\succeq_{\textsf{\tiny ln}}$ (corresponding to KL divergence) is one of these preorders, we can conclude that all the preorders are equivalent to $\succeq_{\textsf{\tiny ln}}$. 

Finally, we derive Proposition \ref{Prop: Operator Convex f-Divergence Contraction} to illustrate that it is a straightforward corollary of Theorem \ref{Thm: Equivalent Characterizations of Less Noisy Preorder}.

\renewcommand{\proofname}{\bfseries \emph{Proof of Proposition \ref{Prop: Operator Convex f-Divergence Contraction}}}

\begin{proof}
Fix any non-linear operator convex function $f$ with $f(1) = 0$, and any channel $P_{Y|X}$ with row stochastic transition probability matrix $W \in \Simplex_{\Y|\X}$. Using Theorem \ref{Thm: Equivalent Characterizations of Less Noisy Preorder}, the $|\X|$-ary erasure channel $E_{1-\beta} \succeq_{\textsf{\tiny ln}} P_{Y|X}$ if and only if for every pair of input pmfs $R_X,P_X \in \Simplex_{\X}$:
\begin{align*}
D_{f}(R_X W||P_X W) & \leq D_{f}(R_X E_{1-\beta}||P_X E_{1-\beta}) \\
& = \beta \, D_{f}(R_X||P_X)
\end{align*}
where the equality is shown near the end of subsection \ref{Contraction Coefficients of Source-Channel Pairs}. This equivalence yields the following generalization of \eqref{Eq: Less Noisy Characterization}:
\begin{equation}
\eta_{f}\!\left(P_{Y|X}\right) = \inf\!\left\{\beta \in [0,1] : E_{1-\beta} \succeq_{\textsf{\tiny ln}} P_{Y|X}\right\} .
\end{equation}
Therefore, the contraction coefficients $\eta_{f}\!\left(P_{Y|X}\right)$ for all non-linear operator convex $f$ are equal, and in particular, they all equal $\etaChi\!\left(P_{Y|X}\right)$ (since $f(t) = t^2 - 1$ is operator convex).
\end{proof}

\renewcommand{\proofname}{\bfseries \emph{Proof}}

\section{Conclusion}
\label{Conclusion}

In closing, we briefly recapitulate our main contributions and then propose some directions for future research. We first illustrated in Theorem \ref{Thm:Local Approximation of Contraction Coefficients} that if the optimization problem defining $\eta_{f}\!\left(P_X,P_{Y|X}\right)$ is subjected to an additional ``local approximation'' constraint that forces the input $f$-divergence to vanish, then the resulting optimum value is $\etaChi\!\left(P_X,P_{Y|X}\right)$. This transparently captures the intuition behind the maximal correlation lower bound in part 6 of Proposition \ref{Prop:Contraction Coefficient Properties}. We then derived a linear upper bound on $\eta_{f}\!\left(P_X,P_{Y|X}\right)$ in terms of $\etaChi\!\left(P_X,P_{Y|X}\right)$ for a class of $f$-divergences in Theorem \ref{Thm:General Contraction Coefficient Bound}, and improved this bound for the salient special case of $\etaKL\!\left(P_X,P_{Y|X}\right)$ in Theorem \ref{Thm:Refined Contraction Coefficient Bound}. Such bounds are useful in weak dependence regimes such as in the analysis of ergodicity of Markov chains (as shown in Corollary \ref{Cor: Rate of Convergence}). In the spirit of comparing contraction coefficients of source-channel pairs, we also gave an alternative proof of the equivalence, $\etaKL\!\left(P_X,P_{Y|X}\right) = \etaChi\!\left(P_X,P_{Y|X}\right)$, for jointly Gaussian distributions $P_{X,Y}$ defined by AWGN channels in Theorem \ref{Thm: Gaussian Contraction Coefficients} and section \ref{Proof of Gaussian Contraction Coefficients}. This proof showed that adding a large enough power constraint to the extremization in $\etaKL\!\left(P_X,P_{Y|X}\right)$ does not change its value. Finally, in the realm of contraction coefficients of channels, we generalized Proposition \ref{Prop: Operator Convex f-Divergence Contraction} in Theorem \ref{Thm: Equivalent Characterizations of Less Noisy Preorder}, and established that the less noisy preorder over channels is completely characterized by any non-linear operator convex $f$-divergence.   

As discussed in subsection \ref{Tensorization of Bounds between Contraction Coefficients}, the constants in the linear bounds in Theorems \ref{Thm:General Contraction Coefficient Bound} and \ref{Thm:Refined Contraction Coefficient Bound} vary ``blindly'' with the dimension of a product distribution. While results like Corollary \ref{Cor:Tensorized Contraction Coefficient Bound} partially remedy this tensorization issue, one compelling direction of future work is to discover linear bounds whose constants gracefully tensorize. Another, perhaps more concrete, avenue of future work is to derive the optimal distribution dependent refinement of Lemma \ref{Lemma:Pinsker's Inequality for f-Divergence} (as suggested in \cite[Remark, p.5380]{fDivergencePinsker}). Such a refinement could be used to tighten Theorem \ref{Thm:General Contraction Coefficient Bound} so that it specializes to Theorem \ref{Thm:Refined Contraction Coefficient Bound} instead of Corollary \ref{Cor:Contraction Coefficient Bound}. However, such a refinement cannot circumvent the more critical tensorization issue that ails these bounds.

\appendices

\section{Proof of Proposition \ref{Prop:Singular Value Characterization of Maximal Correlation}}
\label{App: Proof of Singular Value Characterization of Maximal Correlation}

\begin{proof}
This proof is outlined in \cite{RenyiCorrHypercontractivity}, and presented in \cite[Theorem 3.2.4]{MastersThesis} for the $P_X \in \Simplex_{\X}^{\circ}$ and $P_Y \in \Simplex_{\Y}^{\circ}$ case. We provide it here for completeness. 

Suppose the marginal pmfs of $X$ and $Y$ satisfy $P_X \in \Simplex_{\X}^{\circ}$ and $P_Y \in \Simplex_{\Y}^{\circ}$. We first show that the largest singular value of the DTM $B$ is unity. Consider the matrix:
\begin{align*}
M & = \textsf{\small diag}\!\left(\sqrt{P_Y}\right)^{-1}B^{T}B\,\textsf{\small diag}\!\left(\sqrt{P_Y}\right) \\
& = \textsf{\small diag}\!\left(P_Y\right)^{-1} W^T \textsf{\small diag}\!\left(P_X\right) W \\
& = V W
\end{align*}
where $V = \textsf{\small diag}\!\left(P_Y\right)^{-1} W^T \textsf{\small diag}\!\left(P_X\right) \in \Simplex_{\X|\Y}$ is the row stochastic reverse transition probability matrix of conditional pmfs $P_{X|Y}$. Observe that $M$ has the same set of eigenvalues as the Gramian of the DTM $B^T B$, because we are simply using a similarity transformation to define it. As $B^T B$ is positive semidefinite, the eigenvalues of $M$ and $B^T B$ are non-negative real numbers by the \textit{spectral theorem} (see \cite[Section 2.5]{BasicMatrixAnalysis}). Moreover, since $V$ and $W$ are both row stochastic, their product $M = VW$ is also row stochastic. Hence, the largest eigenvalue of $M$ and $B^T B$ is unity by the \textit{Perron-Frobenius theorem} (see \cite[Chapter 8]{BasicMatrixAnalysis}). It follows that the largest singular value of $B$ is also unity. Notice further that $\sqrt{P_X}$ and $\sqrt{P_Y}$ are the left and right singular vectors of $B$, respectively, corresponding to the singular value of unity. Indeed, we have:
\begin{align*}
\sqrt{P_X} B & = \sqrt{P_X} \, \textsf{diag}\!\left(\sqrt{P_X}\right) \! W \textsf{diag}\!\left(\sqrt{P_Y}\right)^{-1} = \sqrt{P_Y} , \\
B \sqrt{P_Y}^T & = \textsf{diag}\!\left(\sqrt{P_X}\right) \! W \textsf{diag}\!\left(\sqrt{P_Y}\right)^{-1} \sqrt{P_Y} = \sqrt{P_X}^T \! . 
\end{align*}

Next, starting from Definition \ref{Def:Maximal Correlation}, let $f \in \R^{|\X|}$ and $g \in \R^{|\Y|}$ be the column vectors representing the range of the functions $f:\X \rightarrow \R$ and $g:\Y \rightarrow \R$, respectively. Note that we can express the expectations in Definition \ref{Def:Maximal Correlation} in terms of $B$, $P_X$, $P_Y$, $f$, and $g$:
\begin{align*}
 \E\left[f(X)g(Y)\right] & = \left(\textsf{\small diag}\!\left(\sqrt{P_X}\right)f\right)^T B \left(\textsf{\small diag}\!\left(\sqrt{P_Y}\right)g\right) , \\
 \E\left[f(X)\right] & = \sqrt{P_X} \left(\textsf{\small diag}\!\left(\sqrt{P_X}\right)f\right) , \\
 \E\left[g(Y)\right] & = \sqrt{P_Y} \left(\textsf{\small diag}\!\left(\sqrt{P_Y}\right)g\right) , \\
 \E\left[f^2(X)\right] & = \left\|\textsf{\small diag}\!\left(\sqrt{P_X}\right)f\right\|_2^2 , \\
 \E\left[g^2(Y)\right] & = \left\|\textsf{\small diag}\!\left(\sqrt{P_Y}\right)g\right\|_2^2 .
\end{align*}
Letting $a = \textsf{\small diag}\!\left(\sqrt{P_X}\right)f$ and $b = \textsf{\small diag}\!\left(\sqrt{P_Y}\right)g$, we have from Definition \ref{Def:Maximal Correlation}:
$$ \rho(X;Y) = \max_{\substack{{a \in \R^{|\X|}, \, b \in \R^{|\Y|}\,:} \\ {\sqrt{P_X} a = \sqrt{P_Y} b = 0} \\ {\|a\|_2^2 = \|b\|_2^2 = 1}}}{a^T B b} $$
where the optimization is over all $a \in \R^{|\X|}$ and $b \in \R^{|\Y|}$ because $P_X \in \Simplex_{\X}^{\circ}$ and $P_Y \in \Simplex_{\Y}^{\circ}$. Since $a$ and $b$ are orthogonal to the left and right singular vectors corresponding to the maximum singular value of unity of $B$, respectively, this maximization produces the second largest singular value of $B$ using an alternative version (see \cite[Lemma 2]{SubtleVariationalCharacterizationofSingularValue}) of the \textit{Courant-Fischer min-max theorem} (see \cite[Theorems 4.2.6 and 7.3.8]{BasicMatrixAnalysis}). This proves that $\rho(X;Y)$ is the second largest singular value of the DTM when $P_X \in \Simplex_{\X}^{\circ}$ and $P_Y \in \Simplex_{\Y}^{\circ}$. 

We finally argue that one can assume $P_X \in \Simplex_{\X}^{\circ}$ and $P_Y \in \Simplex_{\Y}^{\circ}$ without loss of generality. When $P_X$ or $P_Y$ have zero entries, $X$ and $Y$ only take values in the support sets $\textsf{\small supp}\!\left(P_X\right) = \left\{x \in \X:P_X(x) > 0\right\} \subseteq \X$ and $\textsf{\small supp}\!\left(P_Y\right) = \left\{y \in \Y:P_Y(y) > 0\right\} \subseteq \Y$ respectively, which means that $P_X \in \Simplex_{\textsf{\scriptsize supp}(P_X)}^{\circ}$ and $P_Y \in \Simplex_{\textsf{\scriptsize supp}(P_Y)}^{\circ}$. Let $B$ denote the ``true'' DTM of dimension $|\X| \times |\Y|$ corresponding to the pmf $P_{X,Y}$ on $\X \times \Y$, and $B_{\textsf{\scriptsize supp}}$ denote the ``support'' DTM of dimension $|\textsf{\small supp}\!\left(P_X\right)\!| \times |\textsf{\small supp}\!\left(P_Y\right)\!|$ corresponding to the pmf $P_{X,Y}$ on $\textsf{\small supp}\!\left(P_X\right) \times \textsf{\small supp}\!\left(P_Y\right)$. Clearly, $B$ can be constructed from $B_{\textsf{\scriptsize supp}}$ by inserting zero vectors into the rows and columns associated with the zero probability letters in $\X$ and $\Y$, respectively. Hence, $B$ and $B_{\textsf{\scriptsize supp}}$ have the same non-zero singular values (counting multiplicity), which implies that they have the same second largest singular value. This completes the proof. 
\end{proof}

\section{Proof of Proposition \ref{Prop:Contraction Coefficient Properties}}
\label{App: Proof of Properties of Contraction Coefficients}

\begin{proof} ~\newline
\textbf{Part 1:} The normalization of contraction coefficients is evident from the non-negativity of $f$-divergences and their DPIs \eqref{Eq:DPI for f-Divergence}. We remark that in the case of $\etaChi\!\left(P_X,P_{Y|X}\right) = \rho^2(X;Y)$ (where we use \eqref{Eq: Maximal Correlation as Contraction Coefficient}), $0 \leq \rho(X;Y) \leq 1$ is R\'{e}nyi's third axiom in defining maximal correlation \cite{RenyiCorrelation}. \\
\textbf{Part 2:} We provide a simple proof of this well-known property. Assume without loss of generality that $P_X \in \Simplex_{\X}^{\circ}$ by ignoring any zero probability letters of $\X$. If the resulting $|\X| = 1$, then $X$ is a constant $a.s.$, and the result follows trivially. So, we may also assume that $|\X| \geq 2$. Let $W \in \Simplex_{\Y|\X}$ denote the row stochastic transition probability matrix of the channel $P_{Y|X}$. Since $W$ is unit rank (with all its columns equal to $P_Y$) if and only if $X$ and $Y$ are independent (which means $P_{Y|X = x} = P_Y$ for every $x \in \X$), it suffices to show that $W$ is unit rank if and only if $\eta_{f}\!\left(P_X,P_{Y|X}\right) = 0$. 

To prove the forward direction, note that if $W$ is unit rank, all its rows are equal to $P_Y$ and we have $R_X W = P_Y$ for all $R_X \in \Simplex_{\X}$. Hence, $\eta_{f}\!\left(P_X,P_{Y|X}\right) = 0$ using Definition \ref{Def:Contraction Coefficient}, because $D_{f}(R_X W || P_X W) = 0$ for all input pmfs $R_X \in \Simplex_{\X}$. 

To prove the converse direction, we employ the argument in \cite{MastersThesis} that was used to prove the $\etaKL\!\left(P_{X},P_{Y|X}\right)$ case. Let $R_X = \delta_x$ for any $x \in \X$, where $\delta_x$ is the Kronecker delta pmf such that $\delta_x(x) = 1$ and $\delta_x(i) = 0$ for $i \in \X\backslash\!\{x\}$. (Note that such $R_X \neq P_X$ as $P_X \in \Simplex_{\X}^{\circ}$.) Then, for every $x \in \X$, the $x$th row of $W$ is $P_{Y|X = x} = \delta_x W = P_X W = P_Y$, where $\delta_x W = P_X W$ because $D_{f}(\delta_x W||P_X W) = 0$ and $f$ is strictly convex at unity. Hence, $W$ has unit rank.

We also note that in the $\etaChi\!\left(P_X,P_{Y|X}\right)$ case, this property of maximal correlation is R\'{e}nyi's fourth axiom in \cite{RenyiCorrelation}.
\\
\textbf{Part 3:} This is proven in \cite[Proposition III.3]{SDPIandSobolevInequalities}. \\
\textbf{Part 4:} This is proven in \cite[Theorem III.9]{SDPIandSobolevInequalities}. We also note that two proofs of the tensorization property of $\etaKL$ can be found in \cite{RenyiCorrHypercontractivity}, and a proof of the tensorization property of $\etaChi$ can be found in \cite{TensorizationRenyiCorr}. \\
\textbf{Part 5:} To prove the first part, let $P_{U,X,Y}$ denote the joint pmf of $(U,X,Y)$, and $S \in \Simplex_{\X|\U}$ and $W \in \Simplex_{\Y|\X}$ denote the row stochastic transition probability matrices corresponding to the channels $P_{X|U}$ and $P_{Y|X}$, respectively. Then, $SW \in \Simplex_{\Y|\U}$ is the row stochastic transition probability matrix corresponding to the channel $P_{Y|U}$ using the Markov property. Observe that for every pmf $R_U \in \Simplex_{\U}\backslash\!\left\{P_U\right\}$:
\begin{align*}
D_{f}\!\left(R_U S W ||P_U S W \right) \leq \, & \,\eta_{f}\!\left(P_X,P_{Y|X}\right) \eta_{f}\!\left(P_U,P_{X|U}\right) \cdot \\
& \, D_{f}\!\left(R_U||P_U\right)
\end{align*}
where $P_Y = P_U S W$, $P_X = P_U S$, and we use the SDPI \eqref{Eq:SDPI} twice. Hence, we have:
$$ \eta_{f}\!\left(P_U,P_{Y|U}\right) \leq \eta_{f}\!\left(P_U,P_{X|U}\right) \eta_{f}\!\left(P_X,P_{Y|X}\right) $$ 
using Definition \ref{Def:Contraction Coefficient}. 

The $\etaChi$ specialization of this result corresponds to the sub-multiplicativity property of the second largest singular value of the DTM. Such a sub-multiplicativity property also holds for the $i$th largest singular value of the DTM, cf. \cite[Theorem 2]{KangUlukus2011}, and is useful for distributed source and channel coding applications \cite{KangUlukus2011}. Moreover, the result in \cite[Theorem 2]{KangUlukus2011} is also proved in \cite[Theorem 3]{Calmonetal2017}, where the relation to principal inertia components and maximal correlation is expounded.

To prove the second part, observe that for fixed $P_{X,Y}$, and every $P_{U|X}$ such that $U \rightarrow X \rightarrow Y$ form a Markov chain and $\eta_{f}\!\left(P_U,P_{X|U}\right) > 0$ (which requires that $X$ is not a constant $a.s.$), we have:
\begin{equation}
\label{Eq:Sub-Multiplicativity 2}
\frac{\eta_{f}\!\left(P_U,P_{Y|U}\right)}{\eta_{f}\!\left(P_U,P_{X|U}\right)} \leq \eta_{f}\!\left(P_X,P_{Y|X}\right) 
\end{equation}
using the sub-multiplicativity property established above. Let $U = X \, a.s.$ so that $P_{U|X} \in \Simplex_{\X|\X}$ is the identity matrix. Then, $\eta_{f}\!\left(P_U,P_{X|U}\right) = 1$ and $\eta_{f}\!\left(P_U,P_{Y|U}\right) = \eta_{f}\!\left(P_X,P_{Y|X}\right)$ using Definition \ref{Def:Contraction Coefficient}. Therefore, equality can be achieved in \eqref{Eq:Sub-Multiplicativity 2}, and the proof is complete.

We remark that the $\etaChi$ case of this result is presented in \cite[Lemma 4]{Asoodehetal2016}, where the authors also prove that the optimal channel $P_{U|X}$ can be taken as $P_{Y|X}$ (so that $U$ is a copy of $Y$) instead of the identity matrix (where $U = X \, a.s.$). \\
\textbf{Part 6:} Following the remark after \cite[Theorem 5]{BoundsbetweenContractionCoefficients} in the conference version of this paper, we prove this result via the technique used to prove the $\etaKL$ case in \cite[Theorem 5]{BoundsbetweenContractionCoefficients}.

Let $W \in \Simplex_{\Y|\X}$ denote the row stochastic matrix of the channel $P_{Y|X}$, and $B$ denote the DTM of the joint pmf $P_{X,Y}$. Let us define a trajectory of spherically perturbed pmfs of the form \eqref{Eq: Spherically perturbed pmf}:
$$ R_X^{(\epsilon)} = P_X + \epsilon \, K_X \, \textsf{\small diag}\!\left(\sqrt{P_X}\right) $$
where $K_X \in \mathcal{S} \triangleq \{x \in \left(\R^{|\X|}\right)^{\! *}: \sqrt{P_X} x^T = 0, \, \|x\|_2 = 1\}$ is a spherical perturbation vector. When these pmfs pass through the channel $W$, we get the output trajectory:
\begin{equation}
\label{Eq:Output Trajectory}
R_X^{(\epsilon)} W = P_Y + \epsilon \, K_X B \, \textsf{\small diag}\!\left(\sqrt{P_Y}\right)
\end{equation}
where $B$ maps input spherical perturbations to output spherical perturbations \cite{InfoCouplingConf}. Now, starting from Definition \ref{Def:Contraction Coefficient}, we have:
\begin{align*}
\eta_{f}\!\left(P_X, P_{Y|X}\right) & = \sup_{\substack{R_X \in \Simplex_{\X}:\\0 < D_{f}\left(R_X || P_X\right) < +\infty}}{\frac{D_{f}(R_X W||P_X W)}{D_{f}(R_X||P_X)}} \\
& \geq \liminf_{\epsilon \rightarrow 0}{\sup_{K_X \in \mathcal{S}}{\frac{\left\|K_X B\right\|_2^2 + o(1)}{\left\|K_X\right\|_2^2 + o(1)}}} \\
& \geq \sup_{K_X \in \mathcal{S}}{\liminf_{\epsilon \rightarrow 0}{\frac{\left\|K_X B\right\|_2^2 + o(1)}{1 + o(1)}}} \\
& = \etaChi\!\left(P_X, P_{Y|X}\right) = \rho^2(X;Y)
\end{align*}
where the second inequality follows from \eqref{Def:Spherical Local f-Divergence} after restricting the supremum over all pmfs of the form \eqref{Eq: Spherically perturbed pmf} (where $\epsilon \neq 0$ is some sufficiently small fixed value) and then letting $\epsilon \rightarrow 0$, the third inequality follows from the minimax inequality, and the final equalities follow from \eqref{Eq: Contraction Coefficient as Singular Value} and \eqref{Eq: Maximal Correlation as Contraction Coefficient}, respectively. This completes the proof. 

We remark that the $P_X \in \Simplex_{\X}^{\circ}$ and $P_Y \in \Simplex_{\Y}^{\circ}$ assumptions, while useful for defining the aforementioned trajectories of pmfs, are not essential for this result. For the special case of $\etaKL$, this result was first proved in \cite{AhlswedeGacsHypercontraction}, and then again in \cite{NoninteractiveSimulations} and \cite[Theorem 5]{BoundsbetweenContractionCoefficients}\textemdash{}the latter two proofs both use perturbation arguments with different flavors.
\end{proof}

\section{Proof of Theorem \ref{Thm:Local Approximation of Contraction Coefficients}}
\label{App: Proof of Local Approximation of Contraction Coefficients}

\begin{proof}
We begin by defining the function $\tau:(0,\infty) \rightarrow [0,1]$:
$$ \tau(\delta) \triangleq \sup_{\substack{{R_X \in \Simplex_{\X}:}\\{0 < D_{f}(R_X||P_X) \leq \delta}}}{\frac{D_{f}(W R_X||W P_X)}{D_{f}(R_X||P_X)}} $$
so that what we seek to prove is:
$$ \lim_{n \rightarrow \infty}{\tau(\delta_n)} = \etaChi\!\left(P_X,P_{Y|X}\right) $$
for any decreasing sequence $\{\delta_n > 0:n \in \N\}$ such that $\lim_{n \rightarrow \infty}{\delta_n} = 0$. Note that the limit on the left hand side exists because as $\delta_n \rightarrow 0$, the supremum in $\tau(\delta_n)$ is non-increasing and bounded below by $0$.

We first prove that $\lim_{n \rightarrow \infty}{\tau(\delta_n)} \geq \etaChi\!\left(P_X,P_{Y|X}\right)$. To this end, consider a trajectory of spherically perturbed pmfs of the form \eqref{Eq: Spherically perturbed pmf}:
$$ R_X^{(n)} = P_X + \epsilon_n \, K_X \, \textsf{\small diag}\!\left(\sqrt{P_X}\right) $$
where $K_X \in \mathcal{S} = \{x \in \left(\R^{|\X|}\right)^{\! *}: \sqrt{P_X} x^T = 0, \, \|x\|_2 = 1\}$ is a spherical perturbation vector. The associated trajectory of output pmfs after passing through $W$ is given by \eqref{Eq:Output Trajectory}:
$$ R_X^{(n)} W = P_Y + \epsilon_n \, K_X B \, \textsf{\small diag}\!\left(\sqrt{P_Y}\right) $$
where $B$ denotes the DTM corresponding to $P_{X,Y}$. We ensure that the scalars $\{\epsilon_n \neq 0 : n \in \N\}$ that define our trajectory satisfy $\lim_{n \rightarrow \infty}{\epsilon_n} = 0$ and are sufficiently small such that:
$$ D_{f}(R_X^{(n)}||P_X) = \frac{f^{\prime \prime}(1)}{2}\epsilon_n^2 \left\|K_X\right\|_2^2 + o\!\left(\epsilon_n^2\right) \leq \delta_n $$
where we use \eqref{Def:Spherical Local f-Divergence} (and the fact that $f^{\prime \prime}(1)$ exists and is strictly positive), and the Bachmann-Landau asymptotic little-$o$ notation.\footnote{Given two functions $g(n)$ and $h(n)$ such that $h(n)$ is non-zero, we write $g(n) = o(h(n))$ if and only if $\lim_{n \rightarrow \infty}{g(n)/h(n)} = 0$.} By definition of $\tau$, we have:
\begin{align*}
\sup_{K_X \in \mathcal{S}}{\frac{D_{f}(R_X^{(n)} W||P_X W)}{D_{f}(R_X^{(n)}||P_X)}} & \leq \tau\!\left(\delta_n\right) \\
\lim_{n \rightarrow \infty}{\sup_{K_X \in \mathcal{S}}{\frac{\frac{f^{\prime \prime}(1)}{2}\epsilon_n^2 \left\|K_X B\right\|_2^2 + o\!\left(\epsilon_n^2\right)}{\frac{f^{\prime \prime}(1)}{2}\epsilon_n^2 \left\|K_X\right\|_2^2 + o\!\left(\epsilon_n^2\right)}}} & \leq \lim_{n \rightarrow \infty}{\tau\!\left(\delta_n\right)} \\
\lim_{n \rightarrow \infty}{\sup_{K_X \in \mathcal{S}}{\frac{\left\|K_X B\right\|_2^2 + o(1)}{1 + o(1)}}} & \leq \lim_{n \rightarrow \infty}{\tau\!\left(\delta_n\right)} \\
\etaChi\!\left(P_X,P_{Y|X}\right) & \leq \lim_{n \rightarrow \infty}{\tau\!\left(\delta_n\right)}
\end{align*}
where the second inequality uses \eqref{Def:Spherical Local f-Divergence} for both the numerator and denominator, and the final inequality uses the singular value characterization of $\etaChi\!\left(P_X,P_{Y|X}\right)$ in \eqref{Eq: Contraction Coefficient as Singular Value}.

We next prove that $\lim_{n \rightarrow \infty}{\tau(\delta_n)} \leq \etaChi\!\left(P_X,P_{Y|X}\right)$. Observe that for each $n \in \N$, there exists a pmf $R_X^{(n)} \in \Simplex_{\X}$ satisfying two properties:
\begin{enumerate}
\item $0 < D_{f}(R_X^{(n)}||P_X) \leq \delta_n$
\item $\displaystyle{0 \leq \tau(\delta_n) - \frac{D_{f}(R_X^{(n)} W||P_X W)}{D_{f}(R_X^{(n)}||P_X)} \leq \frac{1}{2^n}}$
\end{enumerate}
where the first property holds because $R_X \mapsto D_{f}(R_X||P_X)$ is a continuous map for fixed $P_X \in \Simplex_{\X}^{\circ}$ (which follows from the convexity of $f$), and the second property holds because $\tau(\delta_n)$ is defined as a supremum. Since $\tau(\delta_n)$ converges as $n \rightarrow \infty$, we have:\footnote{We use the fact that if two sequences $\{a_n \in \R:n \in \N\}$ and $\{b_n \in \R:n \in \N\}$ satisfy $\lim_{n \rightarrow \infty}{|a_n - b_n|} = 0$ and $\lim_{n \rightarrow \infty}{b_n} = b \in \R$, then $\lim_{n \rightarrow \infty}{a_n} = b$.} 
\begin{equation}
\label{Eq: Same limit relation}
\lim_{n \rightarrow \infty}{\frac{D_{f}(R_X^{(n)} W||P_X W)}{D_{f}(R_X^{(n)}||P_X)}} = \lim_{n \rightarrow \infty}{\tau\left(\delta_n\right)} \, .
\end{equation}
Using the sequential compactness of $\Simplex_{\X}$, we can assume that $R_X^{(n)}$ converges as $n \rightarrow \infty$ (in the $\ell^{2}$-norm sense) by passing to a subsequence if necessary. Since $\lim_{n \rightarrow \infty}{D_{f}(R_X^{(n)}||P_X)} = 0$, we have that $\lim_{n \rightarrow \infty}{R_X^{(n)}} = P_X$ due to the continuity of $R_X \mapsto D_{f}(R_X||P_X)$ for fixed $P_X \in \Simplex_{\X}^{\circ}$ and the fact that an $f$-divergence (where $f$ is strictly convex at unity) is zero if and only if its input pmfs are equal. Let us define the spherical perturbation vectors $\{K_X^{(n)} \in \mathcal{S}: n \in \N\}$ using the relation:  
$$ R_X^{(n)} = P_X + \epsilon_n \,K_X^{(n)} \textsf{\small diag}\!\left(\sqrt{P_X}\right) $$
where $\{\epsilon_n \neq 0 : n \in \N\}$ provide the appropriate scalings, and $\lim_{n \rightarrow \infty}{\epsilon_n} = 0$ (since $\lim_{n \rightarrow \infty}{R_X^{(n)}} = P_X$). The corresponding output pmfs are of the form \eqref{Eq:Output Trajectory} mutatis mutandis, and we can approximate the ratio between output and input $f$-divergences as before using \eqref{Def:Spherical Local f-Divergence}:
\begin{align*}
\frac{D_{f}(R_X^{(n)} W||P_X W)}{D_{f}(R_X^{(n)}||P_X)} & = \frac{\frac{f^{\prime \prime}(1)}{2}\epsilon_n^2 \left\|K_X^{(n)} B\right\|_2^2 + o\!\left(\epsilon_n^2\right)}{\frac{f^{\prime \prime}(1)}{2}\epsilon_n^2 \left\|K_X^{(n)}\right\|_2^2 + o\!\left(\epsilon_n^2\right)} \\
& = \frac{\left\|K_X^{(n)} B\right\|_2^2 + o(1)}{1 + o(1)} \, .
\end{align*}
Using the sequential compactness of $\mathcal{S}$, we may assume that $\lim_{n \rightarrow \infty}{K_X^{(n)}} = K_X^{\star} \in \mathcal{S}$ by passing to a subsequence if necessary. Hence, letting $n \rightarrow \infty$, we get:
$$ \lim_{n \rightarrow \infty}{\tau(\delta_n)} = \left\|K_X^{\star} B\right\|_2^2 \leq \etaChi\!\left(P_X,P_{Y|X}\right) $$
where the equality follows from \eqref{Eq: Same limit relation} and the continuity of the map $\left(\R^{|\X|}\right)^{\! *} \ni x \mapsto \left\|x B\right\|_2^2$, and the inequality follows from \eqref{Eq: Contraction Coefficient as Singular Value}. This completes the proof.
\end{proof}

\section{Proof of Corollary \ref{Cor:Contraction Coefficient Bound}}
\label{App: Proof of Contraction Coefficient Bound from General Contraction Coefficient Bound}

\begin{proof}
The convex function $f:(0,\infty) \rightarrow \R, \, f(t) = t \log(t)$ is clearly strictly convex and thrice differentiable at unity with $f(1) = 0$, $f^{\prime}(1) = 1$, $f^{\prime \prime}(1) = 1 > 0$, and $f^{\prime \prime \prime}(1) = -1$. Moreover, the function $g:(0,\infty) \rightarrow \R, \, g(x) = \frac{f(x) - f(0)}{x} = \log(x)$ is clearly concave (where $f(0) = \lim_{t \rightarrow 0^{+}}{f(t)} = 0$). So, to prove Corollary \ref{Cor:Contraction Coefficient Bound} using Theorem \ref{Thm:General Contraction Coefficient Bound}, it suffices to show that $f$ satisfies \eqref{Eq:General Pinsker Condition} for every $t \in (0,\infty)$ (cf. \cite{fDivergencePinsker}):
$$ \left(f(t) - f^{\prime}(1) (t-1)\right)\!\!\left(1 - \frac{f^{\prime \prime \prime}(1)}{3 f^{\prime \prime}(1)}(t-1)\right) \geq \frac{f^{\prime \prime}(1)}{2}(t-1)^2 $$
which simplifies to:
$$ 2t (t + 2) \log(t) - (5 t + 1)(t - 1) \geq 0 \, . $$
Define $h : (0,\infty) \rightarrow \R, \, h(t) = 2t (t + 2) \log(t) - (5 t + 1)(t - 1)$ and observe that:
\begin{align*}
h^{\prime}(t) & = 4(t+1)\log(t) - 8(t-1) \\
h^{\prime \prime}(t) & = 4 \log(t) + \frac{4}{t} - 4 \geq 0
\end{align*}
where the non-negativity of the second derivative follows from the well-known inequality:
$$ \forall x > 0, \enspace x \log(x) \geq x - 1 \, . $$
Since $h$ is convex (as its second derivative is non-negative) and $h(1) = h^{\prime}(1) = 0$, $t = 1$ is a global minimizer of $h$ and $h(t) \geq 0$ for every $t \in (0,\infty)$ as required. 

Finally, we can verify that the constant in Corollary \ref{Cor:Contraction Coefficient Bound} is:
$$ \frac{f^{\prime}(1) + f(0)}{\displaystyle{f^{\prime \prime}(1) \min_{x \in \X}{P_X(x)}}} = \frac{1}{\displaystyle{\min_{x \in \X}{P_X(x)}}} $$
which completes the proof.
\end{proof}

\section{Proof of \eqref{Eq:Looser Lower Bound on KL Divergence}}
\label{App: Proof of Looser Lower Bound on KL Divergence}

\begin{proof} 
Two proofs for \eqref{Eq:Looser Lower Bound on KL Divergence} are provided in the conference version of this paper \cite[Lemma 6]{BoundsbetweenContractionCoefficients}. We present the one with a convex analysis flavor. It involves recognizing that KL divergence is a Bregman divergence associated with the negative Shannon entropy function, and then exploiting the strong convexity of the negative Shannon entropy function to bound KL divergence. Let $H_{\textsf{\tiny neg}}:\Simplex_{\X} \rightarrow \R$ be the negative \textit{Shannon entropy} function, which is defined as:
$$ \forall Q_X \in \Simplex_{\X}, \enspace H_{\textsf{\tiny neg}}(Q_X) \triangleq \sum_{x \in \X}{Q_X(x) \log\!\left(Q_X(x)\right)} \, . $$
Since the \textit{Bregman divergence} corresponding to $H_{\textsf{\tiny neg}}$ is the KL divergence, cf. \cite{BregmanDivergence}, we have for all $S_X \in \Simplex_{\X}$ and $Q_X \in \Simplex_{\X}^{\circ}$:
$$ D(S_X||Q_X) = H_{\textsf{\tiny neg}}(S_X) - H_{\textsf{\tiny neg}}(Q_X) - J_X \nabla H_{\textsf{\tiny neg}}(Q_X) $$
where $J_X = S_X - Q_X$ is an additive perturbation vector, and $\nabla H_{\textsf{\tiny neg}}:\Simplex_{\X}^{\circ} \rightarrow \R^{|\X|}$ is the gradient of $H_{\textsf{\tiny neg}}$. Moreover, as $H_{\textsf{\tiny neg}}$ is twice continuously differentiable, we have: 
$$ \forall Q_X \in \Simplex_{\X}^{\circ}, \enspace \nabla^2 H_{\textsf{\tiny neg}}(Q_X) = \textsf{\small diag}\!\left(Q_X\right)^{-1} \succeq_{\textsf{\tiny PSD}} I $$
where $\nabla^2 H_{\textsf{\tiny neg}}:\Simplex_{\X}^{\circ} \rightarrow \R^{|\X| \times |\X|}$ denotes the Hessian matrix of $H_{\textsf{\tiny neg}}$, and $I \in \R^{|\X| \times |\X|}$ denotes the identity matrix. (Note that $\textsf{\small diag}\!\left(Q_X\right)^{-1} - I$ is positive semidefinite because it is a diagonal matrix with non-negative diagonal entries.) Recall from \cite[Chapter 9]{ConvexOptimization} that a twice continuously differentiable convex function $f : S \rightarrow R$ with open domain $S \subseteq \R^n$ is called \textit{strongly convex} if there exists $m > 0$ such that for all $x \in S$, $\nabla^2 f(x) \succeq m I$. This means that $H_{\textsf{\tiny neg}}$ is strongly convex on $\Simplex_{\X}^{\circ}$. A consequence of this strong convexity is the following quadratic lower bound \cite[Chapter 9]{ConvexOptimization}: 
\begin{align*}
H_{\textsf{\tiny neg}}(S_X) & \geq H_{\textsf{\tiny neg}}(Q_X) + J_X \nabla H_{\textsf{\tiny neg}}(Q_X) + \frac{1}{2}\left\|J_X\right\|_2^2 \\
\Leftrightarrow \enspace D(S_X||Q_X) & \geq \frac{1}{2}\left\|J_X\right\|_2^2
\end{align*}
for every $S_X \in \Simplex_{\X}$ and $Q_X \in \Simplex_{\X}^{\circ}$, where we allow $S_X \in \Simplex_{X}\backslash\Simplex_{X}^{\circ}$ due to the continuity of $H_{\textsf{\tiny neg}}$. This is precisely what we get if we loosen \eqref{Eq: Intermediate Pinsker} in the proof of Lemma \ref{Lemma:Distribution Dependent KL Divergence Lower Bound} using $\left\|J_X\right\|_1 \geq \left\|J_X\right\|_2$ and \eqref{Eq: phi Bound}. Finally, we have for every $S_X \in \Simplex_{\X}$ and $Q_X \in \Simplex_{\X}^{\circ}$:
$$ D(S_X||Q_X) \geq \frac{1}{2} \left\|J_X\right\|_2^2 \geq \frac{\displaystyle{\min_{x \in \X}{Q_X(x)} }}{2} \, \chi^2(S_X||Q_X) $$
where the second inequality follows from \eqref{Eq:Chi-Squared Divergence}. This trivially holds for all $Q_X \in \Simplex_{\X}\backslash\Simplex_{\X}^{\circ}$ as well.
\end{proof}

\section*{Acknowledgment}

A. Makur would like to thank Prof. Yury Polyanskiy for stimulating discussions related to Theorem \ref{Thm: Equivalent Characterizations of Less Noisy Preorder}, and more generally, for discussions regarding contraction coefficients.

\bibliographystyle{IEEEtran}
\bibliography{contractionrefs3}

\end{document}